\documentclass[acmsmall]{acmart}

\acmJournal{TOPLAS}
\acmVolume{1}
\acmArticle{1}
\acmYear{2021}
\acmMonth{1}
\acmDOI{} 
\startPage{1}

\setcopyright{none}

\usepackage{booktabs}   
\usepackage{subcaption} 

\usepackage{microtype}
\usepackage{subcaption}

\usepackage{xifthen}
\usepackage{multicol}
\usepackage{standalone}

\usepackage{andymacros}

\usepackage{tikz}
\usetikzlibrary{bayesnet}
\usetikzlibrary{positioning}

\makeatletter
\let\@authorsaddresses\@empty
\makeatother

\usepackage{algorithm}
\usepackage[noend]{algpseudocode} 

\makeatletter
\newcommand\fs@nobottomruled{\def\@fs@cfont{\bfseries}\let\@fs@capt\floatc@ruled
	\def\@fs@pre{\hrule height.8pt depth0pt \kern2pt}%
	\def\@fs@post{}
	\def\@fs@mid{\kern2pt\hrule\kern2pt}%
	\let\@fs@iftopcapt\iftrue}
\makeatother

\floatstyle{nobottomruled}
\restylefloat{algorithm}

\newcommand{\CommentVars}[1]{\leavevmode\hfill\makebox[.5\linewidth][l]{\textcolor{seaborngreen}{//~#1}}}

\algnewcommand{\IfThen}[2]{
	\State \algorithmicif\ #1\ \algorithmicthen\ #2\ }
\algnewcommand{\StateComment}[1]{\State \textit{\textcolor{darkblue}{#1}}}

\algnewcommand{\StateCommentVars}[1]{\State \textit{\textcolor{midred}{#1}}}

\algnewcommand{\Arguments}{\item[\textbf{Arguments:}]}
\algnewcommand{\Returns}{\item[\textbf{Returns:}]}

\allowdisplaybreaks
\begin{document}

\title{Conditional independence by typing}  



\author{Maria I. Gorinova}
\affiliation{
	\institution{University of Edinburgh}      
}
\author{Andrew D. Gordon}
\affiliation{
	\institution{Microsoft Research}            
}
\affiliation{     
	\institution{University of Edinburgh}       
}
\author{Charles Sutton}
\affiliation{     
	\institution{University of Edinburgh}       
}
\author{Matthijs V\'ak\'ar}
\affiliation{
	\institution{Utrecht University}                     
}

\begin{abstract}

A central goal of probabilistic programming languages (PPLs) is to separate modelling from inference. However, this goal is hard to achieve in practice. Users are often forced to re-write their models in order to improve efficiency of inference or meet restrictions imposed by the PPL. Conditional independence (CI) relationships among parameters are a crucial aspect of probabilistic models that capture a qualitative summary of the specified model and can facilitate more efficient inference. 

We present an information flow type system for probabilistic programming that captures conditional independence (CI) relationships, and show that, for a well-typed program in our system, the distribution it implements is guaranteed to have certain CI-relationships. Further, by using type inference, we can statically \emph{deduce} which CI-properties are present in a specified model.

As a practical application, we consider the problem of how to perform inference on models with mixed discrete and continuous parameters. Inference on such models is challenging in many existing PPLs,  but can be improved through a workaround, where the discrete parameters are used \textit{implicitly}, at the expense of manual model re-writing. We present a source-to-source semantics-preserving transformation, which uses our CI-type system to automate this workaround by eliminating the discrete parameters from a probabilistic program. The resulting program can be seen as a hybrid inference algorithm on the original program, where continuous parameters can be drawn using efficient gradient-based inference methods, while the discrete parameters are inferred using variable elimination.

We implement our CI-type system and its example application in SlicStan: a compositional variant of Stan.\footnotemark

\footnotetext{The implementation is available at 
\url{https://github.com/mgorinova/SlicStan}.}


\end{abstract}

	\ccsdesc[300]{Theory of computation~Random walks and Markov chains}
	\ccsdesc[500]{Theory of computation~Type structures}
	\ccsdesc[300]{Theory of computation~Operational semantics}
	\ccsdesc[500]{Mathematics of computing~Statistical software}
	\ccsdesc[500]{Computing methodologies~Learning in probabilistic graphical models}

\keywords{probabilistic programming, information flow types, static analysis, conditional independence, compiler correctness}  

\maketitle

\section{Introduction}

The number of probabilistic programming languages (PPLs) has grown far and wide, and so has the range of inference techniques they support.
Some focus on problems that can be solved analytically, and provide a symbolic solution \cite{PSI}, others are very flexible in the models they can express and use general-purpose inference algorithms \cite{Anglican}. Some use gradient-based methods \cite{StanJSS}, or message-passing methods \cite{InferNET} to provide an efficient solution at the cost of restricting the range of expressible programs.
Each option presents its own challenges, whether in terms of speed, accuracy or inference constraints, which is why PPL users often are required to learn a set of model re-writing techniques: to be able to change the program until it can be feasibly used within the backend inference algorithm.

Take for example Stan \cite{StanJSS}, which is used by
practitioners in a wide range of sciences and industries to analyse their data
using Bayesian inference.
While efficient inference algorithms exist for continuous-only and for some discrete-only models,  
it is much less clear what algorithm to use for arbitrary models with large numbers of both discrete
and continuous (latent, i.e., unobserved) parameters.
Stan has made a conscious choice \emph{not} to
support probabilistic models with discrete parameters, so as to
perform inference using (dynamic) Hamiltonian Monte Carlo 
(HMC) \cite{HMCfirst,HMChierarchical,NUTS}),
which provides efficient, gradient-based inference for differentiable models.
As a result, Stan has often been criticised \cite{StanOld} for its lack of support for discrete parameters.
What is usually overlooked is that many models with discrete parameters
can, in fact, be accommodated in Stan, by manually marginalising (summing) out
the discrete parameters and drawing them conditionally on the continuous parameters \cite[Chapter 7]{StanUserGuide}.
One of the core model rewriting
techniques is marginalisation: summing over all possible values that 
a random variable can take to obtain a marginal density function that 
does not involve that variable.
Marginalising efficiently is not always an obvious procedure, as it requires 
exploiting conditional independence relationships among the variables in the model.
For probabilistic graphical models, there are well-known
algorithms for enumerating all of the conditional independence
assumptions implied by a model. 
But probabilistic programs are much more general,
including control flow and assignment. For this more
general case, it is much less clear how to determine
conditional independence relationships automatically,
and doing so requires combining ideas from traditional
program analysis and from probabilistic graphical modelling.

In this paper, we introduce an information flow type system that can deduce conditional independence relationships between parameters in a probabilistic program. Finding such relationships can be useful in many scenarios.
As an example application, we implement a semantics-preserving source-to-source transformation that automatically marginalises discrete parameters.
We work in  SlicStan \cite{SlicStanPOPL}, a form of Stan with a more compositional syntax than the original language.
Our system extends SlicStan to support discrete parameters in the case when the discrete parameter space is bounded.
This transform corresponds to the \emph{variable elimination} algorithm \cite{zhang1994simple, koller2009probabilistic}: an exact inference algorithm, efficient in models with sparse structure. 
Combining this transformation with an efficient algorithm for continuous parameters, like HMC, gives us a model-specific, automatically derived inference strategy, which is a composition of variable elimination and the algorithm of choice. 
While we only focus on one application in this paper, our type system for conditional independence is applicable to program transformations of probabilistic programs more generally, and we believe it can enable other composed-inference strategies.

In short, we make the following contributions:
\begin{enumerate}
\item \emph{Factorised semantics for SlicStan}:
As a basis for proving correctness of our transformation, 
we extend SlicStan's type system, so that shredding (which slices a SlicStan program into Stan for execution) correctly separates well-typed programs into data preprocessing, main model, and purely generative code (\autoref{th:shred_gen}).
\item \emph{Main theoretical result}:
We show how a very simple, relatively standard information flow type system can be used to capture a conditional independence in probabilistic programs (\autoref{sec:theory}) and establish a correspondence between well-typed programs and conditional independence 
properties of the probability distribution it implements (\autoref{th:shred_discrete}, \autoref{th:ci}).
\item \emph{Main practical result}:  We describe and implement (in SlicStan) a source-to-source transformation that repeatedly uses the result from (2) to efficiently marginalise out the discrete parameters of the program, and we give a generative procedure for drawing these parameters (\autoref{sec:application}), thus automating inference 
for mixed discrete-continuous models. We prove that our transformation is semantics-preserving (\autoref{th:sempreservation}).
\end{enumerate}


\section{SlicStan: Extended syntax and semantics} \label{sec:background}
SlicStan \cite{SlicStanPOPL} is a Stan-like probabilistic programming language.
Compared to Stan, it provides extra compositionality by dropping the requirement that programs be block-structured.
SlicStan uses type inference in an information-flow type system
\cite{DBLP:journals/jcs/VolpanoIS96, 
      DBLP:conf/popl/AbadiBHR99, 
      DBLP:conf/esop/GordonRSBRGT15} 
to automatically rearrange the program into parts roughly corresponding to the block structure of Stan: pre-processing (data), model, and post-processing (generated quantities).
Originally, this \textit{shredding} was developed to compile SlicStan to Stan.
In this paper, we show that it can be used, more generally, to automatically compile to an efficient program-specific inference scheme.

Like Stan, SlicStan is imperative and allows for deterministic assignment, for-loops, if-statements, probabilistic assignment, and factor-statements.
One contribution of this work is that we present an updated version of SlicStan.

A key difference to the original version of SlicStan is the treatment of sampling ($\sim$) statements. 
In the original SlicStan paper \cite{SlicStanPOPL}, a statement such as $x \sim \normal(0, 1)$ was understood simply as a syntactic sugar for $\kw{factor}(\normal(x \mid 0, 1))$: adding a factor to the underlying density of the model, rather than performing actual sampling. 
In our updated version of SlicStan, sampling statements are part of the core syntax. The semantics of $x \sim \normal(0, 1)$ remains \emph{equivalent} to that of $\kw{factor}(\normal(x \mid 0, 1))$ in terms of density semantics, however it could be \emph{implemented} differently depending on the context. In particular, $x \sim \normal(0, 1)$ could be implemented as a simple call to a random number generator in Stan, $x = \normal_{rng}(0, 1)$, like in the example in \autoref{fig:first_example}. 

This way of treating $\sim$ statements differently is useful, as it allows for an increase of the functionality of the SlicStan's information-flow analysis. Consider, for example the SlicStan program on the right of \autoref{fig:first_example}. Using the original type system, both $\mu$ and $x_{\mathrm{pred}}$ will be of level $\lev{model}$, as they are both involved in a $\sim$ statement. Thus, when translated to Stan, both $\mu$ and $x_{\mathrm{pred}}$ must be inferred with HMC (or another similar algorithm), which is expensive.
However, the updated type system of this paper allows for $x_{\mathrm{pred}}$ to be of level $\lev{genquant}$, which is preferable: in the context of Stan, this means only $\mu$ needs to be inferred with HMC, while $x_{\mathrm{pred}}$ can be simply drawn using a random number generator.
More generally, the updated SlicStan type system allows for factorising the density defined by the program: for data $\data$, parameters $\params$ and generated quantities $Q$, a program defining a density $p(\data, \params, Q)$ can be sliced into two programs with densities $p(\data, \params)$ and $p(Q \mid \data, \params)$ respectively (\autoref{th:shred_gen}). The parameters $\params$ are inferred using HMC (or another general-purpose inference algorithm) according to $p(\data, \params)$, while the quantities $Q$ are directly generated according to $p(Q \mid \data, \params)$. 

\begin{figure*}
\vspace{-10pt}
\begin{multicols}{3} 
\begin{lstlisting}[basicstyle=\small]
// Stan target program
data{ real x; }
parameters{ real $\mu$; }
model{ x $\sim$ normal($\mu$, 1); }
generated quantities{ 
 real x_pred = normal_rng($\mu$, 1); }
\end{lstlisting}	
	
\begin{lstlisting}[basicstyle=\small]
// SlicStan from POPL'19

real $\mu$;
data real x $\sim$ normal($\mu$, 1);
real x_pred = normal_rng($\mu$, 1);

\end{lstlisting}	
\vspace{8pt}

\begin{lstlisting}[basicstyle=\small]
// Extended SlicStan

real $\mu$;
data real x $\sim$ normal($\mu$, 1);
real x_pred $\sim$ normal($\mu$, 1);

\end{lstlisting}
\end{multicols} \vspace{-10pt}
\caption{Example of difference to previous version of SlicStan}
\label{fig:first_example}
\end{figure*}

Treating $\sim$ statements differently based on context is very similar in spirit to existing effect-handling based PPLs \cite{ProbProg18} like Edward2 and Pyro, where $\sim$ can be handled in different ways. However, in our case, this difference in treatment is determined statically, automatically, and only in the translation to Stan or another backend. 

Another difference between \citet{SlicStanPOPL}'s SlicStan and our updated version is the $\kw{target}(S)$ expression, which we use to capture locally the density defined by statements.

These changes are a small but useful contribution of the current work: they are key to allowing us to decompose the program and compose different inference strategies for efficiency.

In the rest of this section, we give the updated formal syntax, typing and semantics of SlicStan and describe shredding --- the procedure key to the translation of Stan / inference composition. 

\subsection{Syntax} \label{ssec:syntax}
SlicStan has the following types, programs, L-values,
statements, and expressions. We highlight the difference with \cite{SlicStanPOPL} with boxes.

\vspace{10pt}

\noindent\begin{minipage}[c]{0.5\linewidth}
	\begin{display}[.50]{SlicStan Types:}
		\clause{\ell ::= \lev{data} \mid \lev{model} \mid \lev{genquant}}{level type}\\
		\clause{n\in\mathbb{N}}{size} \\
		\clause{\tau ::= \kw{real} \mid \kw{int} \mid \fbox{$\kw{int}\langle n \rangle$} \mid \tau []}{base type} \\
		\clause{T ::= (\tau,\ell)}{type}
	\end{display}
\end{minipage}
\begin{minipage}{0.49\linewidth}
	\begin{display}[.35]{SlicStan Program:}
		\clause{P ::= \Gamma,S}{program}
	\end{display}
	\begin{display}[.35]{SlicStan L-Values:}
		\clause{L ::= x[E_1]\cdots[E_n]}{L-value}
	\end{display}
\end{minipage}
\begin{display}{SlicStan Typing Environments:}
	\clause{\Gamma ::= \{x_1\mapsto T_1,\ldots, x_n\mapsto T_n\}}{typing environment}
\end{display}
\vspace{-10pt}

\vspace{-5pt}
\begin{multicols}{2}
	\begin{display}[.18]{SlicStan Statements:}
		\Category{S}{statement}\\
		\entry{L = E}{assignment}\\
		\entry{S_1; S_2}{sequence}\\
		\entry{\kw{for}(x\;\kw{in}\;E_1:E_2)\;S}{for loop} \\
		\entry{\kw{if}(E)\;S_1\,\kw{else}\;S_2}{if statement} \\
		\entry{\kw{skip}}{skip}\\
		\entry{\fbox{\kw{factor}(E)}}{{factor statement}} \\
		\entry{\fbox{$L \sim d(E_1, ..., E_n)$}}{sample statement}
	\end{display}
	\begin{display}[.18]{SlicStan Expressions:}
		\Category{E}{expression}\\
		\entry{x}{variable}\\
		\entry{c}{constant}\\
		\entry{[E_1,...,E_n]}{array}\\
		\entry{E_1[E_2]}{array element}\\ 
		\entry{f(E_1,\dots,E_n)}{function call}\\
		\entry{\fbox{$[E \mid x~\kw{in}~E_1:E_2]$}}{{array comprehension}}\\
		\entry{\fbox{\kw{target}(S)}}{{evaluating a density}}
	\end{display}
\end{multicols}
\vspace{-5pt}

SlicStan programs consist of a pair $\Gamma,S$ of a \emph{typing environment} $\Gamma$
(a finite map that assigns global variables $x$ to their types $T$)
and a \emph{statement} $S$.
Following the usual style of declaring variables in C-like languages,
we informally present programs $\Gamma,S$ in examples
by sprinkling the type declarations of $\Gamma$
throughout the statement $S$.
For example, we write $\kw{data}~\kw{real}~x \sim~ \kw{normal}(0, 1)$ for the program 
$\{x\mapsto (\kw{real},\lev{data})\}, x \sim \kw{normal}(0, 1)$.
Sometimes, we will leave out types or write incomplete types in our examples.
In this case, we intend for the missing types to be determined using 
type inference. 

As we discuss in detail in \autoref{ssec:stanop}, a $\kw{factor}(E)$ statement can be read as
multiplying the current \emph{weight} (contribution to the model's joint density) of the program trace by the value of $E$.
Conversely, a $\kw{target}(S)$ expression initialises the weight to $1$
and returns the weight that is accumulated after evaluating $S$.
For example, if: \vspace{-4pt}
\begin{align*}
S &= \kw{x ~ normal(0,1); y = 2 * x; z ~ normal(y,1);} \\
&=  \kw{factor(normal_pdf(x|0,1)); y = 2 * x; factor(normal_pdf(z|y,1));}
\end{align*}
\noindent Then $\kw{target}(S)$ is semantically equivalent to \kw{normal_pdf(x|0,1) * normal_pdf(z|2 * x,1)}.

We extend the base types of the language of \cite{SlicStanPOPL} with $\kw{int}\langle n \rangle$, which denotes a positive integer constrained from above by an integer $n$. For example if $x$ is of type $\kw{int}\langle 2 \rangle$, then $x$ can only be 1 or 2. These types allow us to specify the support of discrete variables, and they can easily be extended to include both upper and lower bounds. 
For the purpose of our typing rules, we treat $\kw{int}\langle n \rangle$ identically to $\kw{int}$. We only differentiate between these types in $\autoref{sec:application}$, where our transformation uses the size annotation to eliminate a discrete variable.

\subsection{Typing} \label{ssec:typing}

Types $T$ in SlicStan range over pairs $(\tau, \ell)$ of a base type $\tau$, and a level type $\ell$. The level types $\ell$ form a lattice $\left(\{\lev{data}, \lev{model}, \lev{genquant}\}, \leq \right)$, where $\lev{data} \leq \lev{model} \leq \lev{genquant}$.
We write $\bigsqcup_{i=1}^n \ell_i$ for the least upper bound of the levels $\ell_1,\ldots,\ell_n$.
We call variables of level $\lev{data}$ \emph{data (variables)}, 
of level $\lev{model}$ \emph{model parameters},
and of level $\lev{genquant}$ \emph{generated quantities}.
We refer to variables that are either of level $\lev{model}$ or $\lev{genquant}$ 
simply as \emph{parameters}.
Given a typing environment $\Gamma$,
we can consider the well-typedness of expressions and statements,
given the types assigned to variables by $\Gamma$.
The judgment $\Gamma \vdash E : (\tau, \ell)$ means that expression $E$ has type $\tau$ and reads only level $\ell$ and below.
The judgment $\Gamma \vdash S : \ell$ means that statement $S$ assigns only to level $\ell$ and above.
We write $\Gamma\vdash S$ as a shorthand for $\Gamma\vdash S:\lev{data}$.

The typing rules for expressions are those of \cite{SlicStanPOPL} with added rules 
for the two constructs of array comprehensions and $\kw{target}(S)$-expressions. 
The typing rules for statements are as in \cite{SlicStanPOPL}, with three differences (highlighted in boxes). \ref{Factor} and \ref{Sample} add typing rules for the now new language constructs $\kw{factor}(E)$ and $L \sim d(E_1, ..., E_n)$.
The language supports a finite number of 
built-in functions $f$ with type $\tau_1,\ldots,\tau_n\to \tau$ 
and (conditional) distributions $d\in\mathrm{Dist}(\tau_1,\ldots,\tau_n;\tau)$ over $\tau$ given values 
of types $\tau_1,\ldots,\tau_n$.

\begin{display}{Typing Rules for Expressions:}
	\squad
	\staterule{ESub}
	{ \Gamma \vdash E : (\tau,\ell) \quad \ell \leq \ell'}
	{ \Gamma \vdash E : (\tau,\ell') }
	\quad\,
	\staterule{Var}
	{ }
	{ \Gamma, x:T \vdash x:T}  \quad\,
	\staterule{Const}
	{ \kw{ty}(c) = \tau }
	{ \Gamma \vdash c : (\tau,\lev{data}) }\quad\,
	
	\staterule[($f: \tau_1,\dots,\tau_n \to \tau$)]
	{PrimCall}
	{ \Gamma \vdash E_i : (\tau_i,\ell_i) \quad \forall i \in 1..n}
	{ \Gamma \vdash f(E_1,\dots,E_n) : (\tau,\bigsqcup_{i=1}^n \ell_i ) } \quad
	
	\\[\GAP]\squad
	\staterule{ArrEl}
	{\Gamma \vdash E_1 : (\tau[], \ell) \quad \Gamma \vdash E_2 : (\kw{int}, \ell)}
	{\Gamma \vdash E_1[E_2] : (\tau,\ell)}\quad
	
	\fbox{\staterule{Target}
		{\Gamma\vdash S :\ell''\quad \forall \ell' > \ell. \readset_{\Gamma \vdash \ell'}(S)=\emptyset \footnotemark}
		{\Gamma\vdash \kw{target}(S) : (\kw{real},\ell)}}
	
	\\[\GAP]\squad

	\staterule{Arr}
	{\Gamma \vdash E_i : (\tau,\ell) \quad \forall i \in 1..n}
	{\Gamma \vdash [E_1,...,E_n] : (\tau [],\ell)}\squad
	
	\fbox{\staterule{ArrComp}
		{\Gamma\vdash E_1 : (\kw{int},\ell)\quad \Gamma\vdash E_2 : (\kw{int},\ell)\quad \Gamma, x:(\kw{int},\ell)\vdash E : (\tau,\ell)\quad x \notin \dom(\Gamma)}
		{\Gamma \vdash [E \mid x~\kw{in}~E_1:E_2] : (\tau[],\ell)}}
	\\[\GAP]\squad

\end{display}

\footnotetext{We use $\ell' > \ell$ as a shorthand for $\ell \leq \ell' \wedge \neg \ell' \leq \ell $ }

\begin{display}{Typing Rules for Statements:}
	\squad
	\staterule{SSub}
	{ \Gamma \vdash S : \ell' \quad \ell \leq \ell'}
	{ \Gamma \vdash S : \ell }\quad
		
	\staterule[\footnotemark]{Assign}
	{ \Gamma(L) = (\tau, \ell) \quad \Gamma \vdash E : (\tau,\ell)}
	{ \Gamma \vdash (L = E) : \ell }\qquad 
	
	\staterule{If}
	{ \Gamma \vdash E : (\kw{real},\ell) \quad \!\Gamma \vdash S_1\! :\! \ell \quad\! \Gamma \vdash S_2\! :\! \ell}
	{ \Gamma \vdash \kw{if}(E)\;S_1 \;\kw{else}\; S_2: \ell }\qquad
	
	\\[\GAP]\squad
	\fbox{
		\staterule{Seq}
		{ \Gamma \vdash S_1 : \ell \quad \Gamma \vdash S_2 : \ell \quad \shreddable(S_1, S_2)\,\wedge\, \generative(S_1, S_2)}
		{ \Gamma \vdash (S_1; S_2) : \ell }} \qquad
		
	\fbox{
		\staterule{Factor}
		{ \Gamma \vdash E : (\kw{real}, \lev{model})}
		{ \Gamma \vdash \kw{factor}(E) : \lev{model} }}\qquad
	
	\staterule{Skip}
	{ }
	{ \Gamma \vdash \kw{skip} : \ell } \qquad
		
	\\[\GAP]\squad
	\fbox{
		\staterule[\textsuperscript{\ref{footnote:lookup}}($d\in \mathrm{Dist}(\tau_1,\dots,\tau_n ;\tau)$)]{Sample}
		{
		\Gamma(L) = (\tau,\ell') \quad\! \Gamma \vdash E_i : (\tau_i, \ell), \,\, \forall i \in 1..n \quad\! \ell = \ell' \sqcup \lev{model}}
		{ \Gamma \vdash L \sim d(E_1, \dots, E_n) : \ell }} \quad\!

	\\[\GAP]\squad
	\staterule{For}
	{ \Gamma \vdash E_1 : (\kw{int},\ell) \quad \Gamma \vdash E_2 : (\kw{int},\ell) \quad \Gamma, x:(\kw{int}, \ell) \vdash S : \ell \quad x \notin \dom(\Gamma) \quad x \notin \assset(S)}
	{ \Gamma \vdash \kw{for}(x\;\kw{in}\;E_1:E_2)\;S : \ell } \qquad
\end{display}
\footnotetext{\label{footnote:lookup} Here we use $\Gamma(L)$ to look up the type of the L-value $L$ in $\Gamma$. Sometimes we will use an overloaded meaning of this notation (\autoref{def:gammaE}) to look-up the level type of a general expression. Which $\Gamma(.)$ we refer to will be clear from context.}


In these rules, we make use of the following notation (see \autoref{ap:proofs} for precise definitions).
\begin{itemize}
	\item $\assset(S)$: the set of variables $x$ that have been assigned to in $S$.
	\item $\readset_{\Gamma \vdash \ell}(S)$: the set of variables $x$ that are read at level $\ell$ in $S$. 
	\item $\assset_{\Gamma \vdash \ell}(S)$: the set of variables $x$ of level $\ell$ that have been assigned to in $S$.
	\item $\tildeset_{\Gamma\vdash \ell}(S)$: the set of variables $x$ of level $\ell$ that have been $\sim$-ed in $S$.
	\item $\asstildeset_{\Gamma\vdash \ell}(S)= \assset_{\Gamma\vdash \ell}(S)\cup \tildeset_{\Gamma\vdash \ell}(S)$ 
\end{itemize}

The intention in SlicStan is that statements of level $\ell$ are executed before those of 
$\ell'$ if $\ell<\ell'$.
In order to follow that implementation strategy without reordering possibly non-commutative pairs of statements, we impose the condition $\shreddable(S_1, S_2)$ when we sequence $S_1$ and $S_2$ in \ref{Seq}.

\begin{definition}[Shreddable seq]~ $\shreddable(S_1, S_2) \deq  \forall \ell_1,\ell_2. (\ell_2 < \ell_1) \implies \readset_{\Gamma \vdash \ell_1}(S_1) \cap \assset_{\Gamma \vdash \ell_2}(S_2) = \emptyset$.
\end{definition}

For example, this excludes the following problematic program:
\begin{lstlisting}
	data real sigma = 1; 
	model real mu ~ normal(0, sigma);
	sigma = 2;
\end{lstlisting}	

Above, \kw{sigma} and the statements \lstinline{sigma=1} and \lstinline{sigma=2} are of level \lev{data}, which means they should be executed before the statement \lstinline{mu ~ normal(0,sigma)}, which is of level \lev{model}. However, this would change the intended semantics of the program, giving \kw{mu} a $\normal(0, 2)$ prior instead of the intended $\normal(0, 1)$ prior. This problematic program fails to typecheck in SlicStan, as it is not shreddable:  $\lnot \shreddable(\kw{mu ~ normal(0,sigma)},\,\; \kw{sigma = 2})$.

\begin{definition}[Generative seq] ~ $\generative(S_1, S_2) \deq  \forall \ell\neq \lev{model}.\, \tildeset_{\Gamma\vdash \ell}(S_1)\cap \asstildeset_{\Gamma\vdash \ell}(S_2)=\emptyset \wedge \asstildeset_{\Gamma\vdash \ell}(S_1)\cap \tildeset_{\Gamma\vdash \ell}(S_2)=\emptyset$
\end{definition}
To be able to read $x\sim \normal(0,1)$ at level $\lev{genquant}$, depending on the context, either as a probabilistic assignment to $x$
or as a density contribution, we impose the condition $\generative(S_1, S_2)$
when we sequence $S_1$ and $S_2$.
This excludes problematic programs like the following, 
in which the multiple assignments to \kw{y} create a discrepancy between the density semantics of the program $p(y) = \normal(y \mid 0, 1)\normal(y \mid 0, 1)$ and the sampling-based semantics of the program \kw{y = 5}.
\begin{lstlisting}
	genquant real y ~ normal(0, 1); 
	y ~ normal(0, 1); 
	y = 5;
\end{lstlisting}
This problematic program fails to typecheck in SlicStan owing to the $\generative$ constraint: \newline 
$\lnot \generative(\kw{y ~ normal(0,1)},\,\; \kw{y ~ normal(0,1)})$, and also $\lnot \generative(\kw{y ~ normal(0,1)},\,\; \kw{y = 5})$.

\subsection{Operational Semantics of SlicStan Statements} \label{ssec:stanop}
In this paper, we use a modified version of the semantics given in \citet{SlicStanPOPL}.
We extend the call-by-value operational semantics given in that paper, and derive a more equational form that also includes the generated quantities. 

We define a standard big-step operational semantics for SlicStan expressions and statements:
\begin{display}[.2]{Big-step Relation}
	\clause{ (s, E) \Downarrow V }{expression evaluation} \\
	\clause{ (s, S) \Downarrow (s',w)}{statement evaluation} 
\end{display}
Here, $s$ and $s'$ are \emph{states}, $V$ is a \emph{value} and $w\in \weights$ is a \emph{weight}.
Our statements can read and write the state with arbitrary destructive updates.
The weight can be thought of as an element of state that stores a positive real value which only gets accessed by multiplying it with the value of an expression $E$, through the use of  $\kw{factor}(E)$-statements.
It can only be read through a $\kw{target}(S)$-statement which initialises the weight to $1$, evaluates the statement $S$ and returns the final weight.

Formally, states and values are defined as follows.
\begin{display}[0.4]{Values and States:}
	\Category{V}{value}\\
	\entry{c}{constant}\\
	\entry{[V_1,\dots,V_n]}{array}\\
	\clause{s ::= x_1 \mapsto V_1, \dots, x_n \mapsto V_n \quad x_i\textrm{ distinct}}{state (finite map from variables to values)}
\end{display}

In the rest of the paper, we use the notation for states $s = x_1 \mapsto V_1, \dots, x_n \mapsto V_n$:
\begin{itemize}
	\item $s[x \mapsto V]$ is the state $s$, but where the value of $x$ is updated to $V$ if $x \in \dom(s)$, or the element $x \mapsto V$ is added to $s$ if $x \notin \dom(s)$.
	\item $s[-x]$ is the state s, but where $x$ is removed from the domain of $s$ (if it were present).
\end{itemize}


We also define lookup and update operations on values:
\begin{itemize}
	\item If $U$ is an $n$-dimensional array value for $n \geq 0$
	and $c_1$, \dots, $c_n$ are suitable indexes into $U$,
	then the \emph{lookup} $U[c_1]\dots[c_n]$ is the value in $U$ indexed by $c_1$, \dots, $c_n$.
	\item If $U$ is an $n$-dimensional array value for $n \geq 0$
	and $c_1$, \dots, $c_n$ are suitable indexes into $U$,
	then the (functional) \emph{update} $U[c_1]\dots[c_n] := V$ is the array that is the same as $U$ except that the
	value indexed by $c_1$, \dots, $c_n$ is $V$.
\end{itemize}

The relation $\Downarrow$ is deterministic but partial, as we do not explicitly handle error states.
The purpose of the operational semantics is to define a density function in \autoref{ssec:standen}, and any errors lead to the density being undefined.
The big-step semantics is defined as follows.

\vspace{3pt}
\begin{display}{Operational Semantics of Expressions:}
	\squad	
	\staterule{Eval Const}
	{ }
	{ (s, c) \Downarrow c }\qquad
	
	\staterule{Eval Var}
	{ V = s(x)  \quad x \in \dom(s)}
	{ (s, x) \Downarrow V  }\qquad
	
	\staterule{Eval Arr}
	{ (s, E_i) \Downarrow V_i \quad \forall i \in 1.. n }
	{ (s, [E_1, \dots, E_n]) \Downarrow [V_1, \dots, V_n]}\qquad	
	
	\\[1.3ex]\squad
	\staterule{Eval ArrEl}
	{ (s, E_1 \Downarrow V) \quad (s, E_2 \Downarrow c)}
	{ (s, E_1[E_2]) \Downarrow V[c]}\qquad
	
	\staterule[\footnotemark]{Eval PrimCall}
	{ (s, E_i) \Downarrow V_i \quad \forall i \in 1 \dots n \quad V = f(V_1, \dots, V_n)}
	{ (s, f(E_1, \dots, E_n)) \Downarrow V}
	
	\\[1.3ex]\squad
	\staterule[\footnotemark]{Eval ArrComp}
	{ (s, E_1) \Downarrow n \quad (s, E_2) \Downarrow m \quad (s, E[i/x]) \Downarrow V_i, \forall n\leq i\leq m }
	{ (s, [E \mid x~\kw{in}~E_1:E_2]) \Downarrow [V_{n},\ldots, V_{m}] }\qquad
	
	\staterule{Eval Target}
	{ (s, S) \Downarrow (s', w)}
	{ (s, \kw{target}(S)) \Downarrow w}
\end{display}

\footnotetext[4]{$f(V_1, \dots, V_n)$ means applying the built-in function $f$ on the values $V_1, \dots, V_n$.}
\footnotetext{Here, we write $E[E'/x]$ for the usual capture avoiding substitution of $E'$ for $x$ in $E$.}
%
%
\begin{display}{Operational Semantics of Statements:}
	
	\staterule[~(where $L=x[E_1{]}\dots[E_n{]}$)]{Eval Assign}
	{ (s,E_i) \Downarrow V_i \quad \forall i \in 1..n \quad (s,E) \Downarrow V \quad
		U = s(x) \quad
		U' = (U[V_1]\dots[V_n] := V) }
	{ (s, L=E) \Downarrow (s[x \mapsto U'], 1)}
	
	\\[1.3ex]
	\staterule{Eval Skip}
	{ }
	{ (s, \kw{skip}) \Downarrow (s, 1) }\qquad
	
	\staterule{Eval Seq}
	{ (s, S_1) \Downarrow (s', w) \quad (s', S_2) \Downarrow (s'', w')}
	{ (s, S_1;S_2) \Downarrow (s'', w * w')}\qquad
	
	\staterule{Eval ForFalse}
	{ (s, E_1) \Downarrow c_1 \quad (s, E_2) \Downarrow c_2 \quad c_1 > c_2}
	{ (s, \kw{for}(x\;\kw{in}\;E_1:E_2)\;S) \Downarrow (s, 1) } \qquad
	
	\\[1.3ex]\squad
	\staterule{Eval ForTrue}
	{ \{(s, E_i) \Downarrow c_i\}_{i=1,2} \quad c_1 \leq c_2 \quad (s[x \mapsto c_1], S) \Downarrow (s', w) \quad (s'[-x], \kw{for}(x\;\kw{in}\;(c_1+1):c_2)\;S) \Downarrow (s'', w')}
	{ (s, \kw{for}(x\;\kw{in}\;E_1:E_2)\;S) \Downarrow (s'', w * w') } \qquad
	
	\\[1.3ex]\squad

	\staterule{Eval IfTrue}
	{ (s, E) \Downarrow c\neq 0.0 \quad (s, S_1) \Downarrow (s', w)}
	{ (s, \kw{if}(E)\; S_1\; \kw{else}\; S_2) \Downarrow (s', w)}\qquad
	
	\staterule{Eval IfFalse}
	{ (s, E) \Downarrow 0.0 \quad (s, S_2) \Downarrow (s', w)}
	{ (s, \kw{if}(E)\; S_1\; \kw{else}\; S_2) \Downarrow (s', w)}\qquad

	\\[1.3ex]\squad
	\staterule{Eval Factor}
	{ (s, E) \Downarrow V }
	{ (s, \kw{factor}(E)) \Downarrow (s, V)}\qquad
	\staterule[\footnotemark]{Eval Sample}
	{ (s, L) \Downarrow V \quad (s, E_i) \Downarrow V_i, \forall 1\leq i \leq n\quad V' =d(V| V_1,\ldots, V_n)  }
	{ (s, L \sim d(E_1, \ldots, E_n)) \Downarrow (s, V') }
\end{display}

\footnotetext{By $d(V|V_1,\ldots,V_n)$, we mean the result of evaluating the intended built-in conditional distribution $d$ on $V,V_1,\ldots,V_n$.}

Most rules of the big-step operational semantics are standard, with the exception of \ref{Eval Factor} and \ref{Eval Sample}, which correspond to the PPL-specific language constructs \kw{factor} and $L \sim d(E_1, \dots, E_n)$. While we refer to the latter construct as \textit{probabilistic assignment}, its formal semantics is not that of an assignment statement: both the left and the right hand-side of the ``assignment'' are evaluated to a value, in order for the density contribution $d(V \mid V_1, \dots, V_n)$ to be evaluated and factored into the weight of the current execution trace.
Contrary to \ref{Eval Assign}, there is no binding of a result to a variable in \ref{Eval Sample}.  
Of course, as is common in probabilistic programming, it might, at times\footnote{For example, in our Stan backend for SlicStan, if such a statement is of level \lev{model}, it will be executed as density contribution, while if it is  of level \lev{genquant}, it will be executed as a probabilistic assignment. }, 
be beneficial to execute these statements as actual probabilistic assignments.
Our treatment of these statements is agnostic of such implementation details, however.

The design of the type system ensures that information 
can flow from a level $\ell$ to a higher one $\ell'\geq \ell$,
but not a lower one $\ell'<\ell$: a noninterference result.
To state this formally, we introduce the notions of
\emph{conformance between a state $s$ and a typing environment $\Gamma$} and
\emph{$\ell$-equality} of states.

We define a conformance relation on states $s$ and typing environments $\Gamma$. A state $s$ \emph{conforms} to an environment $\Gamma$, whenever $s$ provides values of the correct types for the variables used in $\Gamma$:

\begin{display}[-0.03]{Conformance Relation:}
	\clause{ s \models \Gamma }{state $s$ conforms to environment $\Gamma$}
\end{display}
\begin{display}[.2]{Rule for the Conformance Relation:}
	\quad
	\staterule{Stan State}
	{ V_i \models \tau_i \quad \forall i \in I}
	{(x_i \mapsto V_i)^{i \in I} \models (x_i : \tau_i)^{i \in I}}
\end{display}

Here, $V \models \tau$ denotes that the value $V$ is of type $\tau$, and it has the following definition:
\begin{itemize}
	\item $c \models \kw{int}$, if $c \in \mathbb{Z}$, and $c \models \kw{real}$, if $c \in \mathbb{R}$. 
	\item $[V_1,\dots,V_n] \models \tau[n]$, if $\forall i \in 1\dots n. V_i \models \tau$. 
\end{itemize}

\begin{definition}[$\ell$-equal states]~\\ Given a typing environment $\Gamma$, states $s_1 \models \Gamma$ and $s_2 \models \Gamma$ are $\ell$-equal for level $\ell$ (written $s_1 \approx_{\ell} s_2$), if they differ only for variables of a level strictly higher than $\ell$:
	$$s_1 \approx_{\ell} s_2 \deq \forall x:(\tau, \ell') \in \Gamma. \left( \ell' \leq \ell \implies s_1(x) = s_2(x) \right)$$
\end{definition}

\begin{lemma}[Noninterference of $\vdash$]~\\ \label{lem:noninterf} Suppose $s_1 \models \Gamma$, $s_2 \models \Gamma$, and $s_1 \approx_{\ell} s_2$ for some $\ell$. Then for SlicStan statement $S$ and expression $E$:
	\begin{enumerate}
		\item If $~\Gamma \vdash E:(\tau,\ell)$ and $(s_1, E) \Downarrow V_1$ and $(s_2, E) \Downarrow V_2$ then $V_1 = V_2$. 
		\item If $~\Gamma \vdash S:\ell$ and $(s_1, S) \Downarrow s_1', w_1$ and $(s_2, S) \Downarrow s_2', w_2$ then $s_1' \approx_{\ell} s_2'$.
	\end{enumerate}
\end{lemma}


\begin{proof} (1)~follows by rule induction on the derivation $\Gamma \vdash E:(\tau, \ell)$, and using that if $\Gamma \vdash E:(\tau, \ell)$, $E$ reads $x$ and $\Gamma(x) = (\tau', \ell')$, then $\ell' \leq \ell$. (2)~follows by rule induction on the derivation $\Gamma \vdash S:\ell$ and using (1). We present more details of the proof in \autoref{ap:proofs}.
\end{proof}

\subsection{Density Semantics} \label{ssec:standen}
The semantic aspect of a SlicStan 
program $\Gamma, S$ that we are the most interested in is the final weight $w$
obtained after evaluating the program $S$.
This is the value the program computes for the
unnormalised joint density $p^*(\mathbf{x}) = p^*(\data,\params, \quants)$ over the 
data $\data$,
the model parameters $\params$, and generated quantities $\quants$
of the program (see \autoref{ssec:factorisation}).
Given a program $\Gamma, S$, we separate the typing environment 
$\Gamma$ into disjoint parts: $\Gamma_{\sigma}$ and $\Gamma_{\mathbf{x}}$,
such that $\Gamma_{\sigma}$ contains precisely the variables that are deterministically assigned 
in $S$ and $\Gamma_{\mathbf{x}}$ contains those which never get deterministically assigned; that is the variables $\mathbf{x}$ with respect to which we define the target unnormalised density $p^*(\mathbf{x})$:
\begin{align*}
	\Gamma_{\sigma}=\{(x:T)\in \Gamma\mid x\in W(S)\}\qquad
	\Gamma_{\mathbf{x}}=\Gamma\setminus \Gamma_{\sigma}.
\end{align*}
Similarly, any conforming state $s\models \Gamma$ separates as $\sigma \uplus \mathbf{x}$
with 
\begin{align*}
	\sigma =\{(x\mapsto V)\in s\mid x\in W(S)\}\quad \mathbf{x}=s\setminus \sigma.
\end{align*}
Then, $\sigma\models \Gamma_{\sigma}$ and $\mathbf{x}\models \Gamma_{\mathbf{x}}$.

The semantics of a SlicStan program $\Gamma_{\sigma}, \Gamma_{\mathbf{x}}, S$ is a function $\sem{S}$ on states $\sigma \models \Gamma_{\sigma}$ and $\mathbf{x} \models \Gamma_{\mathbf{x}}$ that yields a pair of a state $\sigma'$ and a weight $w$, such that:\vspace{-6pt}
$$\sem{S}(\sigma)(\mathbf{x}) = \sigma', w,
\quad\text{ where } \sigma \uplus \mathbf{x}, S \Downarrow \sigma'\uplus \mathbf{x}, w.$$

We will sometimes refer only to one of the two elements of the pair $\sigma, w$. In those cases we use the notation:
$\sems{S}(\sigma)(\mathbf{x}),\semp{S}(\sigma)(\mathbf{x})=\sem{S}(\sigma)(\mathbf{x})$.
We call $\sems{S}$ the \emph{state semantics} and $\semp{S}$ the \emph{density semantics} of $\Gamma, S$.
We will be particularly interested in the density semantics.

The function $\semp{S}(\sigma)$ is some positive function $\phi(\mathbf{x})$ of  $\mathbf{x}$.
If $\mathbf{x}_1, \mathbf{x}_2$ is a partitioning of $\mathbf{x}$ and $\int \phi(\mathbf{x}) \dif \mathbf{x}_1$ is finite, we say $\phi(\mathbf{x})$ is an unnormalised density corresponding to the normalised density $p(\mathbf{x}_1 \mid \mathbf{x}_2) = \phi(\mathbf{x}) /  \int \phi(\mathbf{x}) \dif \mathbf{x}_1$ over $\mathbf{x}_1$ and we write $\semp{S}(\sigma){(\mathbf{x})} \propto p(\mathbf{x}_1 \mid \mathbf{x}_2)$.
%
Sometimes, when $\sigma$ is clear from context, we will leave it implicit and 
simply write $p(\mathbf{x})$ for $p(\mathbf{x}; \sigma)$.

Next, we observe how the state and density semantics compose.

\begin{lemma}[Semantics composes] \label{lem:sem_properties}~
	The state and density semantics compose as follows:
$$\sems{S_1; S_2}(\sigma)(\mathbf{x}) = \sems{S_2}(\sems{S_2}(\sigma)(\mathbf{x}))(\mathbf{x}) \qquad
\semp{S_1; S_2}(\sigma)(\mathbf{x}) = \semp{S_1}(\sigma)(\mathbf{x}) \times \semp{S_2}(\sems{S_1}(\sigma)(\mathbf{x}))(\mathbf{x})$$

\end{lemma}

Throughout the paper we use the following notation to separate the store in a concise way. 
\begin{definition}[$\Gamma_{\ell}(s)$ or $s_{\ell}$] ~\\
For a typing environment $\Gamma$ and a store $s \models \Gamma$, let $\Gamma_{\ell}(s) = \{(x \mapsto V) \in s \mid \Gamma(x) = (\_, \ell)\}$. When it is clear which typing environment the notation refers to, we write simply $s_{\ell}$ instead of $\Gamma_{\ell}(s)$. 
\end{definition}

Using this definition, we re-state the noninterference result in the following  convenient form.
\begin{lemma}[Noninterference of $\vdash$ reformulated] \label{lem:noninterf-reform}~
Let $~\Gamma_{\sigma}, \Gamma_{\mathbf{x}}\vdash S$ be a well-typed SlicStan program.
For all levels $\ell \in \{\lev{data}, \lev{model}, \lev{genquant}\}$, there exist unique functions $f_{\ell}$, such that for all 
$\sigma \models \Gamma_{\sigma}$, $\mathbf{x} \models \Gamma_{\mathbf{x}}$ and $\sigma'$ such that $\sems{S}(\sigma)(\mathbf{x}) = \sigma'$,
\quad $ \sigma_{\ell}' = f_{\ell}(\{ \sigma_{\ell'}, \mathbf{x}_{\ell'} \mid \ell' \leq \ell \})$.
%
\end{lemma}



\subsection{Shredding and Translation to Stan} \label{ssec:shred}
A key aim of SlicStan is to rearrange the input program into three phases of execution,
corresponding to the levels of the type system: $\lev{data}$ preprocessing, core $\lev{model}$ code to run MCMC or another inference algorithm on, and $\lev{genquant}$, or generated quantities,
which amount to sample post-processing after inference is performed.
The motivation for these phases is that they all naturally appear in the workflow of probabilistic programming.
The blocks of the Stan are built around this phase distinction, and compilation 
of SlicStan to Stan and comparable back-ends requires it.

The phases impose different restrictions on the code and make it incur differing computational costs.
The model phase is by far the most expensive to evaluate: code in this phase tends to be executed repeatedly within the inner loop of an inference algorithm like an MCMC method.
Further, it tends to be automatically differentiated \cite{griewank2008evaluating} in case gradient-based inference algorithms are used, which restricts the available programming features and increases the space and time complexity of evaluation.
Type inference in SlicStan combined with shredding allows the user to write their code 
without worrying about the performance of different phases, as 
code will be shredded into its optimal phase of execution.

The shredding relation is in the core of this rearrangement.
Shredding takes a SlicStan statement $S$ and splits it into three \textit{single-level statements} (\autoref{def:singlelev}).
That is, $S \shred S_D, S_M, S_Q$ means we split $S$ into sub-statements $S_D, S_M, S_Q$, were $S_D$ mentions only \lev{data} variables, $S_M$ mentions \lev{data} and \lev{model} variables, and $S_Q$ is the rest of the program, and such that the composition $S_D; S_M; S_Q$ behaves the same as the original program $S$. 
When combined with type inference, shredding automatically determines optimal statement placement, such that only necessary work is executed in the `heavy-weight' \lev{model} part of inference.

We adapt the shredding from \cite{SlicStanPOPL}, so that the following holds for the three sub-statements of a shredded well-typed SlicStan program $\Gamma\vdash S$:
\begin{itemize}
\item $S_D$ implements \textit{deterministic data preprocessing}: no 
contributions to the density are allowed. 
\item $S_M$ is the \textit{inference core}: it is the least restrictive of the three slices --- either or both of $S_D$ and $S_Q$ can be merged into $S_M$. It can involve 
contributions to the density which require advanced inference for sampling. Therefore, this is the part of the program which requires the most computation during inference (in Stan, what is run inside HMC);
\item $S_Q$ represents \textit{sample post-processing}: any contributions to the  density are generative. That is, they can immediately be implemented 
using draws from random number generators.
\end{itemize}
In terms of inference, we can run $S_D$ once as a pre-processing step. Then use a suitable inference algorithm for $S_M$ (in the case of Stan, that's HMC, but we can use other MCMC or VI algorithms), and, finally, we use ancestral sampling for $S_Q$. \footnote{Ancestral (or forward) sampling refers to the method of sampling from a joint distribution by individually sampling variables from the factors constituting the joint distribution. For example, we can sample from $p(x, y) = p(x)p(y \mid x)$ by randomly generating $\hat{x}$ according to $p(x)$, and then randomly generating $\hat{y}$ according to $p(y \mid x = \hat{x})$.}
\begin{display}[.3]{Shredding Relation}
	\clause{S \shred \shredded}{statement shredding} 
\end{display}
\begin{display}{Shredding Rules for Statements:}
	\squad	
	\staterule{Shred Assign}
	{\Gamma(L) = (\_,\lev{data}) \rightarrow S_D = L = E, S_M = S_Q = \kw{skip} \\
		\Gamma(L) = (\_,\lev{model}) \rightarrow S_M = L = E, S_D = S_Q = \kw{skip} \\
		\Gamma(L) = (\_,\lev{genquant}) \rightarrow S_Q = L = E, S_D = S_M = \kw{skip}	}
	{ L = E \shred (S_D, S_M, S_Q)}\quad
	
	\staterule{Shred Seq}
	{ S_1 \shred S_{D_1}, S_{M_1}, S_{Q_1} \quad 
		S_2 \shred S_{D_2}, S_{M_2}, S_{Q_2}}
	{ S_1; S_2 \shred (S_{D_1};S_{D_2}), (S_{M_1};S_{M_2}), (S_{Q_1};S_{Q_2})  } 
	
	\\[\GAP]\squad
	\staterule{Shred Factor}
	{\Gamma(E) = \lev{data} \rightarrow S_D = \kw{factor}(E), S_M = S_Q = \kw{skip} \\
		\Gamma(E) = \lev{model} \rightarrow S_M = \kw{factor}(E), S_D = S_Q = \kw{skip} \\
		\Gamma(E) = \lev{genquant} \rightarrow S_Q = \kw{factor}(E), S_D = S_M = \kw{skip}	}
	{ \kw{factor}(E) \shred (S_D, S_M, S_Q)}\quad
	
	\staterule{Shred Skip}
	{}
	{\kw{skip} \shred (\kw{skip}, \kw{skip}, \kw{skip})}\qquad
	
	\\[\GAP]\squad
	\staterule{Shred Sample}
	{\Gamma(L, E_1, \dots, E_n) = \lev{data} \rightarrow S_D = L \sim d(E_1, \dots, E_n), S_M = S_Q = \kw{skip}) \\
		\Gamma(L, E_1, \dots, E_n) = \lev{model} \rightarrow S_M = L \sim d(E_1, \dots, E_n), S_D = S_Q = \kw{skip}) \\
		\Gamma(L, E_1, \dots, E_n) = \lev{genquant} \rightarrow S_Q = L \sim d(E_1, \dots, E_n), S_D = S_M = \kw{skip})	}
	{ L \sim d(E_1, \dots, E_n) \shred (S_D, S_M, S_Q)}\quad\hquad
	
	\\[\GAP]\squad	
	\staterule{Shred If}
	{   S_1 \shred \shredded[1] \quad 
		S_2 \shred \shredded[2] \quad}
	{ \kw{if}(g)\; S_1\; \kw{else}\; S_2 \shred  
		(\kw{if}(g)\; S_{D_1}\; \kw{else}\; S_{D_2}),  
		(\kw{if}(g)\; S_{M_1}\; \kw{else}\; S_{M_2}), 
		(\kw{if}(g)\; S_{Q_1}\; \kw{else}\; S_{Q_2})} \qquad	
	
	\\[\GAP]\squad
	\staterule{Shred For}
	{   S \shred \shredded  }
	{ \kw{for}(x\;\kw{in}\;g_1:g_2)\;S \shred  
		(\kw{for}(x\;\kw{in}\;g_1:g_2)\;S_D),  
		(\kw{for}(x\;\kw{in}\;g_1:g_2)\;S_M), 
		(\kw{for}(x\;\kw{in}\;g_1:g_2)\;S_Q)} \qquad
\end{display}



Here, $\Gamma(E)$ (\autoref{def:gammaE}) gives the principal type of an expression $E$, while $\Gamma(E_1, \dots, E_n)$ (\autoref{def:gammaEs}) gives the least upper bound of the principal types of $E_1, \dots, E_n$.  

The \ref{Shred If} and \ref{Shred For} rules make sure to shred if and for statements so that they are separated into parts which can be computed independently at each of the three levels. Note that the usage of $\kw{if}$ and $\kw{for}$ guards is simplified, to avoid stating rules for when the guard(s) are of different levels. For example, if we have a statement 
$\kw{if}(E)~S_1~\kw{else}~S_2$, where $E$ is of level \lev{model}, we cannot access $E$ at level $\lev{data}$, thus the actual shredding rule we would use is:

\vspace{6pt}
\staterule{Shred If Model Level}
{   S_1 \shred (S_{D_1}, S_{M_1}, S_{Q_1}) \quad 
	S_2 \shred (S_{D_2}, S_{M_2}, S_{Q_2}) \quad}
{ \kw{if}(g)\; S_1\; \kw{else}\; S_2 \shred  
	\kw{skip},  
	(\kw{if}(g)\; S_{D_1}; S_{M_1}\; \kw{else}\; S_{D_2}; S_{M_2}), 
	(\kw{if}(g)\; S_{Q_1}\; \kw{else}\; S_{Q_2})
}
\vspace{6pt}


These shredding rules follow very closely those given by \citet{SlicStanPOPL}. The main difference is that sample statements ($L \sim d(E_1, \dots, E_n)$) are allowed to be of \lev{genquant} level and can be included in the last, generative slice of the program (see rule \ref{Shred Sample}). In other words, such \lev{genquant} sample statements are those statements that can be interpreted as probabilistic assignment (using random number generator functions) to directly sample from the posterior distribution according to ancestral sampling.

We provide proofs for the following key results in \autoref{ap:proofs}: shredding produces \textit{single-level statements} (\autoref{def:singlelev} 
and \autoref{lem:shredisleveled}) and shredding is semantics preserving (Lemma~\ref{lem:shred}).

Intuitively, a single-level statement of level $\ell$ is one that updates only variables of level $\ell$. 

\begin{definition}[Single-level Statement $\Gamma \vdash \ell(S)$] \label{def:singlelev}
	We define single-level statements  $S$ of level $\ell$ with respect to $\Gamma$ (written $\Gamma \vdash \ell(S)$), by induction:
	\begin{display}{Single Level Statements:}
		\squad
		\staterule{Assign Single}
		{\Gamma(x)=(\_,\ell)}
		{\Gamma\vdash \ell(x[E_1]\cdots[E_n]=E)}\quad
		
		\staterule{Seq Single}
		{ \Gamma\vdash \ell(S)\quad \Gamma\vdash\ell(S') }
		{ \Gamma\vdash \ell(S;S') }\quad
		
		\staterule{For Single}
		{\Gamma,x:(\kw{int},\ell)\vdash \ell(S)}
		{\Gamma\vdash \ell(\kw{for}(x\; \kw{in}\; E_1: E_2) S)}\quad

		\staterule{If Single}
		{\Gamma\vdash \ell(S_1)\quad \Gamma\vdash \ell(S_2)}
		{\Gamma\vdash \ell(\kw{if}(E)\;S_1\;\kw{else}\; S_2)}
		\quad  
		
		\\[\GAP]\squad		
		\staterule{Skip Single}
		{~}
		{\Gamma\vdash \ell(\kw{skip})}\quad
		
		\staterule{Factor Single}
		{\Gamma \vdash E:\ell \quad \forall \ell'<\ell. \Gamma\not\vdash E:\ell'}
		{\Gamma\vdash \ell(\kw{factor}(E))}\\[\GAP]\squad
		
		\staterule{Sample Single}
		{\Gamma\vdash L\sim d(E_1,\ldots, E_n) : \ell \quad 
			\forall \ell'<\ell.\Gamma\not\vdash L \sim d(E_1,\ldots, E_n):\ell'}
		{\Gamma\vdash \ell(L\sim d(E_1,\ldots, E_n))}		
	\end{display}
\end{definition}

\begin{lemma}[Shredding produces single-level statements] \label{lem:shredisleveled}
	$$\Gamma\vdash S \;\wedge\;
	S \shred[\Gamma] \shredded \implies 
	\singlelevelS{\lev{data}}{S_D} \wedge \singlelevelS{\lev{model}}{S_M} \wedge \singlelevelS{\lev{genquant}}{S_Q}$$
\end{lemma}


We prove a result about the effect of single-level statements on the store and weight of well-typed programs (\autoref{lem:single-lev-prop}). Intuitively, this result shows that a single-level statement of level $\ell$ acts on the state and weight in a way that is independent of levels greater than $\ell$.
\begin{lemma}[Property of single-level statements] \label{lem:single-lev-prop} ~\\
	Let $~\Gamma_{\sigma}, \Gamma_{\mathbf{x}}, S$ be a SlicStan program, such that $S$ is a single-level statement of level $\ell$, $\Gamma \vdash \ell(S)$. Then there exist unique functions $f$ and $\phi$, such that for any $\sigma, \mathbf{x} \models \Gamma_{\sigma}, \Gamma_{\mathbf{x}}$: 
	$$  \sem{S}(\sigma)(x) = f(\sigma_{\leq \ell}, \mathbf{x}_{\leq \ell})\cup \sigma_{> \ell} , \hquad \phi(\sigma_{\leq \ell})(\mathbf{x}_{\leq \ell}), $$
	where we write $\sigma_{\leq \ell}=\{(x\mapsto V)\in \sigma\mid \Gamma_{\sigma}(x)=(\_,\ell)\}$ and $\sigma_{>\ell}=\sigma\setminus \sigma_{\leq \ell}$.
\end{lemma}

\begin{lemma}[Semantic Preservation of $\shred$] \label{lem:shred} ~ \\
	If $~\Gamma \vdash S:\lev{data} $ and $ S \shred[\Gamma] \shredded $ then $\sem{S} = \sem{S_D; S_M; S_Q}$.
\end{lemma}
\subsection{Density Factorisation} \label{ssec:factorisation}
As an extension of \cite{SlicStanPOPL},
we show that shredding induces a natural factorization of the density 
implemented by the program: $p(\data, \params, Q) = p(\params,\data)p(Q \mid \params,\data)$. \footnote{Here, $p(Q \mid \params,\data)$ denotes the \textit{conditional probability density} of $Q$, given the values of $\params$ and $\data$.}
This means that we can separate the program into a deterministic preprocessing part, a part that uses a `heavy-weight' inference algorithm, such as HMC, and a part that uses simple ancestral sampling. 

\begin{theorem}[Shredding induces a factorisation of the density]\label{th:shred_gen}~ \\
	Suppose $\Gamma \vdash S : \lev{data}$ and $~S \shred[\Gamma] S_D, S_M, S_Q$
	and $\Gamma = \Gamma_{\sigma}, \Gamma_{\data}, \Gamma_{\params}, \Gamma_{\quants}$.
	%
	For all $\sigma$, $\data$, $\params$, and $\quants$:
	if $\sigma, \data, \params, \quants \models \Gamma_{\sigma}, \Gamma_{\data}, \Gamma_{\params}, \Gamma_{\quants}$,
	and $\semp{S}(\sigma)(\data, \params, \quants) \propto p(\data, \params, Q)$
	and  $\tildeset(S_Q)=\dom(\Gamma_Q)$ then:
	\begin{enumerate}
		\item $\semp{S_M}(\sigma_D)(\data, \params, \quants) \propto p(\params, \data)$
		\item $\semp{S_Q}(\sigma_M)(\data, \params, \quants) = p(Q \mid \params, \data)$
	\end{enumerate}
	where $\sigma_D = \sems{S_D}(\sigma)(\data, \params, \quants)$,
	$\sigma_M = \sems{S_M}(\sigma_D)(\data, \params, \quants)$,
	and $p(\data, \params, Q) = p(\data, \params)p(Q \mid \data, \params)$.
\end{theorem}
\begin{proof}
This follows by proving a more general result using induction on the structure of $S$, \autoref{lem:shred}, \autoref{lem:sem_properties} and \autoref{lem:shredisleveled}. See \autoref{ap:proofs} for full proof.
\end{proof}

The given SlicStan program $S$ defines a joint density $p(\data, \params, Q)$.
By shredding we obtain a \lev{model} block $S_M$ that defines $p(\params, \data)$
and a \lev{genquant} block $S_Q$ that defines $p(Q \mid \params, \data)$.
Hence, inference in Stan using these blocks recovers the semantics $p(\data, \params, Q)$ of the SlicStan program.


\section{Theory: Conditional independence by typing} \label{sec:theory}

This section presents the main theoretical contribution of the paper: an information flow type system for conditional independence.
We present a type system and show that a well-typed program in that system is guaranteed to have certain conditional independencies in its density semantics.
As a reminder, determining the conditional independence relationships between variables is important, as such relationships capture a qualitative summary of the specified model and can facilitate more efficient inference. For example, in \autoref{sec:application} we present an application that uses our type system: a semantic-preserving transformation that allows for discrete parameters to be introduced in SlicStan, which was previously not possible due to efficiency constraints.  

Our aim is to optimise probabilistic programs by transforming abstract syntax trees or intermediate representations (as in the Stan compiler) that are close to abstract syntax.
Hence, we seek a way to compute conditional dependencies by a type-based source analysis, rather than by explicitly constructing a separate graphical representation of the probabilsitic model.

Given three disjoint sets of random variables (RVs) $A$, $B$ and $C$, we say that $A$ is \emph{conditionally independent} of $B$ given $C$, written $A \bigCI B \mid C$, if and only if
their densities factorise as $p(A, B \mid C) = p(A \mid C)p(B \mid C)$. (An alternative formulation states that $A \bigCI B \mid C$ if and only if $p(A, B, C) = \phi_1(A, C) \phi_2(B, C)$ for some functions $\phi_1$ and $\phi_2$.)
Deriving conditional independencies in the presence of a graphical model (such as a factor graph%
\footnote{A factor graph is a bipartite graph that shows the factorisation of a multivariable function. Variables are circular nodes, and each factor of the function is a square node. An edge exists between a variable node $x$ and a factor node $\phi$ if and only if $\phi$ is a function of $x$. See Program A and its corresponding factor graph as an example, or   \cite{koller2009probabilistic} for details.}%
) is straightforward, which is why some PPLs focus on building and performing inference on graphs (for example, Infer.NET \cite{InferNET}). However, building and manipulating a factor graph in generative PPLs (e.g. Gen \cite{Gen}, Pyro \cite{Pyro}, Edward2 \cite{Edward2}, PyMC3 \cite{PyMC3}) or imperative density-based PPLs (SlicStan, Stan) is not straightforward. Dependencies between modelled variables might be separated by various deterministic transformations, making it harder to track the information flow, and -- more importantly -- more difficult to isolate parts of the model needed for transformations such as variable elimination.
In the case of SlicStan, each program can still be thought of as specifying a factor graph \emph{implicitly}.
In this paper, we focus on the problem of how to work with conditional independence information implicitly encoded in a probabilistic program, without having access to an explicit factor graph.
For example, consider Program A: 

\vspace{-4pt}
\begin{multicols}{2}
\textbf{A.\,Simple Hidden Markov Model (HMM)}
\begin{lstlisting}
		int<2> z1 ~ bern($\theta_0$);
		real $\theta_1$ = foo($\theta_0$, z1);
		int<2> z2 ~ bern($\theta_1$);
		real $\phi_1$ = foo($1$, z1);
		real $\phi_2$ = foo($1$, z2);
		int<2> y1 ~ bern($\phi_1$);
		int<2> y2 ~ bern($\phi_2$);
\end{lstlisting}
\vspace{12pt}
\begin{tikzpicture}
	\node[latent]   (d1)   { $z_1$}; %
	\factor[left=of d1, xshift=-0.5cm]   {p1} {\scriptsize  $z_1 \sim b(\theta_0)$} {} {};
	\factor[below=of d1]  {yp1} {left:{\scriptsize $y_1 \sim b(\mathrm{foo}(1, z_1))$}} {} {};
	\node[latent, below=of yp1] (y1) {$y_1$};
	\factor[right=of d1, xshift=0.5cm] {p2} {\scriptsize $z_2 \sim b(\mathrm{foo}(\theta_0, z_1))$} {} {}; %
	\node[latent, right=of p2]   (d2)   { $z_2$}; %
	\factor[below=of d2]  {yp2} {left:{\scriptsize $y_2 \sim b(\mathrm{foo}(1, z_2))$}} {} {};
	\node[latent, below=of yp2] (y2) {$y_2$};
	
	\factoredge [] {} {p1} {d1};
	\factoredge [] {d1} {p2} {d2};
	\factoredge [] {d1} {yp1} {y1};
	\factoredge [] {d2} {yp2} {y2};
\end{tikzpicture}
\end{multicols}
\vspace{-4pt}

The factor graph above represents the factorisation of the joint density function over the parameters of the program: 
$p(z_1, z_2, y_1, y_2) = b(z_1 \mid \theta_0) b(y_1 \mid \mathrm{foo}(1, z_1))b(z_2 \mid \mathrm{foo}(\theta_0, z_1)) b(y_2 \mid \mathrm{foo}(1, z_2))$. Each of the four factors is represented by a square node in the graph, and it connects to the variables (circle nodes) that the factor depends on. 
This representation is useful for thinking about conditional independencies. 
For example, it is immediately evident from the graph that variables which connect to the same square node cannot be conditionally independent as they share a factor. More generally, if there is an (uninterrupted by observed variables) direct path between two variables, then these two variables are \textit{not} conditionally independent  \cite{FactorGraphs}. 

When looking at the factor graph, it is straightforward to see that $z_1$ and $z_2$ are \emph{not} conditionally independent, and neither are $z_1$ and $y_1$ nor $z_2$ and $y_2$, as there is a direct path between each of these pairs. When looking at the program, however, we need to reason about the information flow through the deterministic variables $\theta_1, \phi_1$ and $\phi_2$ to reach the same conclusion. 

Moreover, manipulation of the program based on conditional dependencies can also be more difficult without a factor graph. As an example, consider the problem of variable elimination (which we discuss in more details in \autoref{ssec:ve}). If we are to eliminate $z_1$ in the factor graph, using variable elimination, we would simply merge the factors directly connected to $z_1$, sum over $z_1$, and attach the new factors to all former neighbours of $z_1$ (in this case $y_1$ and $z_2$, but not $y_2$).
However, in the case of an imperative program, we need to isolate all the statements that depend on $z_1$, and group them together without changing the meaning of the program beyond the elimination:

\vspace{-4pt}
\begin{multicols}{2}
\textbf{B.\,HMM with $z_1$ marginalised out} \vspace{-2pt}
\begin{lstlisting}
	factor(sum([target(
		  z1 ~ bern($\theta_0$); real $\theta_1$ = foo($\theta_0$, z1);
		  z2 ~ bern($\theta_1$); real $\phi_1$ = foo($1$, z1);
		  y1 ~ bern($\phi_1$); ) | z1 in 1 : 2 ]));		  
	real $\phi_2$ = foo($1$, z2);
	int<2> y2 ~ bern($\phi_2$);
\end{lstlisting}
\vspace{12pt}

\textcolor{white}{.}

\hspace{-7pt}
\resizebox{!}{70pt}{%
\begin{tikzpicture}
\factor  {f} {} {} {};
\node[left=of f, align=right, xshift=1cm, yshift=-0.3cm] (hidden) {
	\scriptsize $	\sum_{z_1} \left[b(z_1 \mid \theta_0)\right.$ \\
	\scriptsize $\times b(y_1 \mid \mathrm{foo}(1, z_1))$ \\
	\scriptsize $\left.\times b(z_2 \mid \mathrm{foo}(\theta_0, z_1))\right]$
};
\node[latent, right=of f, xshift=7.5pt]   (d2)   { $z_2$}; %
\factor[below=of d2]  {yp2} {right:{\scriptsize $y_2 \sim b(\mathrm{foo}(1, z_2))$}} {} {};
\node[latent, below=of yp2, yshift=0.6cm] (y2) {$y_2$};
\node[latent, below=of f, left=of y2] (y1) {$y_1$};

\factoredge [] {y1, d2} {f} {};
\factoredge [] {d2} {yp2} {y2};
\end{tikzpicture}
}
\end{multicols}

We need a way to analyse the information flow to determine conditional independencies between variables. In the example above, we can leave $y_2$ out of the elimination of $z_1$, because $z_1$ and $y_2$ are conditionally independent given $z_2$, written $z_1 \bigCI y_2 \mid z_2$. 

To analyse the information flow, we introduce a novel type system, which we refer to via the relation $\typingspecial$. It works with a lower semi-lattice $(\{\lev{l1}, \lev{l2}, \lev{l3}\}, \leq)$ of levels, where $\lev{l1} \leq \lev{l2}$ and $\lev{l1} \leq \lev{l3}$ and $\lev{l2}$ and $\lev{l3}$ are unrelated.
(Recall that a lower semi-lattice is a partial order in which any two elements $\ell_1$, $\ell_2$ have a greatest lower bound $\ell_1 \sqcap \ell_2$ 
but do not always have an upper bound.)
A well-typed program induces a conditional independence relationship for the (random) variables (RVs) in the program: $\lev{l2}\text{-RVs} \bigCI \lev{l3}\text{-RVs} \mid \lev{l1}\text{-RVs}$. 

In the example above, this result allows us to eliminate $\lev{l2}$-variables ($z_1$), while only considering $\lev{l1}$-variables ($y_1$ and $z_2$) and knowing $\lev{l3}$-variables ($y_2$) are unaffected by the elimination.
We can use a shredding relation almost identical to that of \autoref{ssec:shred} to slice the program in a semantics-preserving way, and isolate the sub-statements needed for elimination. Here, $\theta_1$ and $\phi_1$ must be of level $\lev{l2}$ for the program to be well-typed. Thus, all statements involving $z_1, \theta_1$ or $\phi_1$ are of level $\lev{l2}$, and the shredding relation groups them together inside of the elimination loop for $z_1$. 

\autoref{fig:intuition} shows the relationship between the levels $\lev{l1}, \lev{l2}, \lev{l3}$ and the shredding relation. Information flows from $\lev{l1}$ to $\lev{l2}$ and $\lev{l3}$, but there is no flow of information between $\lev{l2}$ and $\lev{l3}$ (\autoref{fig:intuition_infoflow}). A $\typingspecial$-well-typed program $S$ is shredded by $\shred$ into $S_1$, $S_2$ and $S_3$, where $S_1$ only mentions $\lev{l1}$ variables, $S_2$ only mentions $\lev{l1}$ and $\lev{l2}$ variables, and $S_3$ only mentions $\lev{l1}$ and $\lev{l3}$ variables. This can be understood as a new factor graph formulation of the original program $S$, where each of the substatements $S_1, S_2, S_3$ defines a factor connected to any involved variables (\autoref{fig:intuition_factorgrpah}). 

\begin{figure*}
	\centering
	\begin{subfigure}[b]{0.45\textwidth}
		\centering
		\begin{tikzpicture}
		\node[latent] (l1) {$\mathbf{x}_{\lev{l1}}$};
		\factor[below=of l1] {s1} {right:{$\semp{S_1}$}} {} {};
		
		\coordinate[above=of l1, yshift=-15pt] (mid);
		
		\factor[right=of mid, xshift=26pt] {s2} {315:{$\semp{S_2}$}} {} {};
		\node[latent, above=of s2, yshift=-16pt] (lz) {$\mathbf{x}_{\lev{l2}}$};
		
		\factor[left=of mid, xshift=-26pt]  {s3} {225:{$\semp{S_3}$}} {} {};
		\node[latent, above=of s3, yshift=-16pt] (l3) {$\mathbf{x}_{\lev{l3}}$};
		
		\factoredge [] {s1, s2, s3} {l1} {};
		\factoredge [] {s2} {lz} {};
		\factoredge [] {s3} {l3} {};
		\end{tikzpicture}
		\caption{Factor graph of variables of different levels.}
		\label{fig:intuition_factorgrpah}
	\end{subfigure}
	\begin{subfigure}[b]{0.45\textwidth}
		\centering
		\begin{tikzpicture}
		\node[latent] (l1) {$\mathbf{x}_{\lev{l1}}$};
		\node[below=of l1, yshift=12pt] (blank) {$ $};
		\coordinate[above=of l1] (mid);
		\node[latent, right=of mid] (lz) {$\mathbf{x}_{\lev{l2}}$};
		\node[latent, left=of mid] (l3) {$\mathbf{x}_{\lev{l3}}$};
		
		\draw [->, seabornred] (l1) -- (l3) node[midway, above] {};
		\draw [->, seabornred] (l1) -- (lz) node[midway, above] {};
		\end{tikzpicture}
		\caption{Information flow between levels.}
		\label{fig:intuition_infoflow}
	\end{subfigure}
	\caption{Intuition for the semi-lattice case $\lev{l1} < \lev{l2}$ and $\lev{l1} < \lev{l3}$, where $\mathbf{x}_{\ell}$ is of level $\ell$. We get $\mathbf{x}_{\lev{l2}} \cicaption \mathbf{x}_{\lev{l3}} \mid \mathbf{x}_{\lev{l1}}$.}
	\label{fig:intuition}
	\vspace{-2pt}
\end{figure*}

Our approach relies on determining the $\lev{l1}, \lev{l2}, \lev{l3}$ level types by type inference, as they are \emph{not} intrinsic to the variables or program in any way, but are designed solely to determine conditional independence relationships. These types are not accessible by the probabilistic programming user. 
Our type system makes it possible to answer various questions about conditional independence in a program. Assuming a program defining a joint density $p(\mathbf{x})$, we can use the type system to: 
\begin{enumerate} 
	\item Check if $\mathbf{x}_2 \bigCI \mathbf{x}_3 \mid \mathbf{x}_1$ for some partitioning $\mathbf{x} = \mathbf{x}_1, \mathbf{x}_2, \mathbf{x}_3$.
	\item Find an optimal variable partitioning. Given a variable $x \in \mathbf{x}$, find a partitioning $\mathbf{x} = \mathbf{x}_1, \mathbf{x}_2, \mathbf{x}_3$, such that $x \in \mathbf{x}_2$, $\mathbf{x}_2 \bigCI \mathbf{x}_3 \mid \mathbf{x}_1$, and $\mathbf{x}_1$ and $\mathbf{x}_2$ are as small as possible.
	\item Ask questions about the Markov boundary of a variable. Given two variables $x$ and $x'$, find the partitioning $\mathbf{x} = x, \mathbf{x}_1, \mathbf{x}_2$, such that $x \bigCI \mathbf{x}_1 \mid \mathbf{x}_2$ and $\mathbf{x}_2$ is as small as possible. Is $x'$ in $\mathbf{x}_2$? In other words, is $x'$ in the Markov boundary of $x$?
\end{enumerate} 

In the rest of \autoref{sec:theory}, we give the $\typingspecial$ type system (\autoref{ssec:typingspecial}), state a noninterference result (\autoref{lem:noninterf2}, \autoref{lem:noninterf2-reform}) and show that semantics is preserved when shredding $\typingspecial$-well-typed programs (\autoref{lem:shred2}). We present the type system and transformation rules in a declarative style. The implementation relies on type inference, which we discuss in \autoref{ssec:mb}. 
We derive a result about the way shredding factorises the density defined by the program (\autoref{th:shred_discrete}). 
We prove a conditional independence result (\autoref{ssec:ci}, \autoref{th:ci}) and discuss the scope of our approach with examples (\autoref{ssec:scope}).


\vspace{-3pt}
\subsection{The $\typingspecial$ Type System} \label{ssec:typingspecial}
\vspace{-1pt}

We introduce a modified version of SlicStan's type system. 
Once again, types $T$ range over pairs $(\tau, \ell)$ of a base type $\tau$, and a level type $\ell$, but levels $\ell$ are one of $\lev{l1}, \lev{l2},$ or $\lev{l3}$, which form a lower semi-lattice $\left(\{\lev{l1}, \lev{l2}, \lev{l3}\}, \leq \right)$, where $\lev{l1} \leq \lev{l2}$ and $\lev{l1} \leq \lev{l3}$. This means, for example, that an \lev{l2} variable can depend on an \lev{l1} variable, but an \lev{l3} variable cannot depend on an \lev{l2} variable, as level types \lev{l2} and \lev{l3} are incomparable.

The type system is a standard information flow type system, very similar to the $\vdash$ system introduced in \autoref{ssec:typing}. We mark the only non-standard rules,
\ref{Sample2}, \ref{Factor2}, and \ref{Seq2}, which also differ from those of $\vdash$.
\ref{Sample2} and \ref{Factor2} both have the same effect as an assignment to an implicit weight variable that can be of any of the three levels. 
\ref{Seq2} is a less restrictive version of \ref{Seq} and exactly as in \cite{SlicStanPOPL}, and it makes sure the program can be sliced later. 

Note also that the non-interference between $\lev{l2}$ and $\lev{l3}$ relies on the \ref{PrimCall2} rule not being derivable when the least upper bound $\bigsqcup_{i=1}^n \ell_i$ does not exist. 

\vspace{1pt}
\begin{display}{Typing Rules for Expressions:}
	\squad
	\staterule{ESub2}
	{ \Gamma \typingspecial E : (\tau,\ell) \quad \ell \leq \ell'}
	{ \Gamma \typingspecial E : (\tau,\ell') }
	\quad\,
	\staterule{Var2}
	{ }
	{ \Gamma, x:T \typingspecial x:T}  \quad\,
	\staterule{Const2}
	{ \kw{ty}(c) = \tau }
	{ \Gamma \typingspecial c : (\tau,\lev{l1}) }\quad\,
	
	\staterule{Arr2}
	{\Gamma \typingspecial E_i : (\tau,\ell) \quad \forall i \in 1..n}
	{\Gamma \typingspecial [E_1,...,E_n] : (\tau [n],\ell)}
	
	\\[\GAP]\squad
	\staterule{ArrEl2}
	{\Gamma \typingspecial E_1 : (\tau[n], \ell) \quad \Gamma \vdash E_2 : (\kw{int}, \ell)}
	{\Gamma \typingspecial E_1[E_2] : (\tau,\ell)} \hquad
	
	\staterule[($f: \tau_1,\dots,\tau_n \to \tau$)]
	{PrimCall2}
	{ \Gamma \typingspecial E_i : (\tau_i,\ell_i) \quad \forall i \in 1..n }
	{ \Gamma \typingspecial f(E_1,\dots,E_n) : (\tau,\bigsqcup_{i=1}^n \ell_i ) } 
	
	\\[\GAP]\squad
	\staterule{ArrComp2}
	{\forall i=1,2.\Gamma\typingspecial E_i : (\kw{int},\ell)\quad \Gamma, x:(\kw{int},\ell)\vdash E : (\tau,\ell)\quad\! x \notin \dom(\Gamma)}
	{\Gamma \typingspecial [E \mid x~\kw{in}~E_1:E_2] : (\tau[n],\ell)}\hquad
	
	\staterule{Target2}
	{\Gamma\typingspecial S :\ell''\quad\! \forall \ell' > \ell. \readset_{\Gamma \vdash \ell'}(S)=\emptyset}
	{\Gamma\typingspecial \kw{target}(S) : (\kw{real},\ell)}
\end{display}

\vspace{-5pt}
\begin{display}{Typing Rules for Statements:}
	\squad
	\staterule{SSub2}
	{ \Gamma \typingspecial S : \ell' \quad \ell \leq \ell'}
	{ \Gamma \typingspecial S : \ell }\qquad
	
	\staterule{Assign2}
	{ \Gamma(L) = (\tau, \ell) \quad \Gamma \typingspecial E : (\tau,\ell)}
	{ \Gamma \typingspecial (L = E) : \ell }\qquad 
	
	\\[\GAP]\squad\!\!
	\fbox{
		\staterule{Sample2}
		{ \Gamma \typingspecial \kw{factor}(\mathrm{D}(L \mid E_1, \dots, E_n)) : \ell}
		{ \Gamma \typingspecial L \sim \mathrm{D_{dist}}(E_1, \dots E_n) : \ell }} \quad\;
	
	\fbox{
		\staterule{Factor2}
		{ \Gamma \typingspecial E : (\kw{real}, \ell)}
		{ \Gamma \typingspecial \kw{factor}(E) : \ell }}\quad \;
	
	\fbox{
		\staterule{Seq2}
		{ \Gamma \typingspecial S_1 : \ell \quad \Gamma \typingspecial S_2 : \ell \quad \shreddable(S_1, S_2)}
		{ \Gamma \typingspecial (S_1; S_2) : \ell }} \qquad
	
	\\[\GAP]\squad	
	\staterule{If2}
	{ \Gamma \typingspecial E : (\kw{bool},\ell) \quad \Gamma \typingspecial S_1 : \ell \quad \Gamma \typingspecial S_2 : \ell}
	{ \Gamma \typingspecial \kw{if}(E)\;S_1 \;\kw{else}\; S_2: \ell }\qquad
	
	\staterule{Skip2}
	{ }
	{ \Gamma \typingspecial \kw{skip} : \ell } \qquad
	
	\\[\GAP]\squad
	\staterule{For2}
	{ \Gamma \typingspecial E_1 : (\kw{int},\ell) \quad \Gamma \typingspecial E_2 : (\kw{int},\ell) \quad \Gamma, x:(\kw{int}, \ell) \typingspecial S : \ell \quad x \notin \dom(\Gamma) \quad x \notin \assset(S)}
	{ \Gamma \typingspecial \kw{for}(x\;\kw{in}\;E_1:E_2)\;S : \ell } \qquad
\end{display}

We state and prove a noninterference result for $\typingspecial$, which follows similarly to the result for $\vdash$.
\begin{lemma}[Noninterference of $\typingspecial$] \label{lem:noninterf2} Suppose $s_1 \models \Gamma$, $s_2 \models \Gamma$, and $s_1 \approx_{\ell} s_2$ for some $\ell$. Then for a SlicStan statement $S$ and expression $E$:
	\begin{enumerate}
		\item If $~\Gamma \typingspecial E: (\tau,\ell)$ and $(s_1, E) \Downarrow V_1$ and $(s_2, E) \Downarrow V_2$ then $V_1 = V_2$. 
		\item If $~\Gamma \typingspecial S: \ell$ and $(s_1, S) \Downarrow s_1', w_1$ and $(s_2, S) \Downarrow s_2', w_2$ then $s_1' \approx_{\ell} s_2'$.
	\end{enumerate}
\end{lemma}
%
%
\begin{proof} (1)~follows by rule induction on the derivation $\Gamma \typingspecial E:(\tau, \ell)$, and using that if $\Gamma \typingspecial E:(\tau, \ell)$, $x \in \readset(E)$ and $\Gamma(x) = (\tau', \ell')$, then $\ell' \leq \ell$. (2)~follows by rule induction on the derivation $\Gamma \typingspecial S:\ell$ and using (1).
\end{proof}

Once again we derive a more convenient form of the noninterference result. Because the level types $\lev{l2}$ and $\lev{l3}$ are not comparable in the order $\leq$, changes in the store at $\lev{l2}$ do not affect the store at $\lev{l3}$ and vice versa.
\begin{lemma}[Noninterference of $\typingspecial$-well-typed programs] \label{lem:noninterf2-reform} ~\\
	Let $~\Gamma_{\sigma}, \Gamma_{\mathbf{x}}, S$ be a SlicStan program, and
	$\Gamma \typingspecial S : \lev{l1}$.
	There exist unique functions $f, g$ and $h$, such that for all 
	$\sigma \models \Gamma_{\sigma}$, $\mathbf{x} \models \Gamma_{\mathbf{x}}$ and $\sigma'$ such that $\sems{S}(\sigma)(\mathbf{x}) = \sigma'$: 
	$$\sigma_{\lev{l1}}' = f(\sigma_{\lev{l1}}, \mathbf{x}_{\lev{l1}}), \hquad
	\sigma_{\lev{l2}}' = g(\sigma_{\lev{l1}},\sigma_{\lev{l2}}, \mathbf{x}_{\lev{l1}},  \mathbf{x}_{\lev{l2}}), \hquad
	\sigma_{\lev{l3}}' = h(\sigma_{\lev{l1}},\sigma_{\lev{l3}}, \mathbf{x}_{\lev{l1}}, \mathbf{x}_{\lev{l3}})
	$$
\end{lemma}
%
%
\begin{proof}
	Follows from noninterference (\autoref{lem:noninterf2}).
\end{proof}

Next, we extend the shredding relation from $\autoref{ssec:shred}$, and the concept of single-level statements, to SlicStan programs that are well-typed with respect to $\typingspecial$. This is done by simply treating $\lev{l1}$ as $\lev{data}$, $\lev{l2}$ as $\lev{model}$, and $\lev{l3}$ as $\lev{genquant}$ for the purpose of shredding. We include the full definition of shredding with respect to $\typingspecial$ for completeness below. We use the same notation $\shred$, and we generally treat the \textit{standard shredding relation} from \ref{ssec:shred} and the \textit{conditional independence shredding relation} presented here, as the same relation, as there is no difference between the two, other than the naming of levels.     

\begin{display}{Shredding Rules for Statements:}
	\squad	
	\staterule{Shred2 Assign}
	{\Gamma(L) = \lev{l1} \rightarrow S_1 = L = E, S_2 = S_3 = \kw{skip} \\
		\Gamma(L) = \lev{l2} \rightarrow S_2 = L = E, S_1 = S_3 = \kw{skip} \\
		\Gamma(L) = \lev{l3} \rightarrow S_3 = L = E, S_1 = S_2 = \kw{skip}	}
	{ L = E \shred (S_1, S_2, S_3)}\quad
	
	\staterule{Shred2 Seq}
	{ S_1 \shred S_{1}^{(1)}, S_{2}^{(1)}, S_{3}^{(1)} \quad 
		S_2 \shred S_{1}^{(2)}, S_{2}^{(2)}, S_{3}^{(2)}}
	{ S_1; S_2 \shred (S_{1}^{(1)};S_{1}^{(2)}), (S_{2}^{(1)};S_{2}^{(2)}), (S_{3}^{(1)};S_{3}^{(2)})  } 
	
	\\[\GAP]\squad
	\staterule{Shred2 Factor}
	{\Gamma(E) = \lev{l1} \rightarrow S_1 = \kw{factor}(E), S_2 = S_3 = \kw{skip} \\
		\Gamma(E) = \lev{l2} \rightarrow S_2 = \kw{factor}(E), S_1 = S_3 = \kw{skip} \\
		\Gamma(E) = \lev{l3} \rightarrow S_3 = \kw{factor}(E), S_1 = S_2 = \kw{skip}	}
	{ \kw{factor}(E) \shred (S_1, S_2, S_3)}\quad
	
	\staterule{Shred2 Skip}
	{}
	{\kw{skip} \shred (\kw{skip}, \kw{skip}, \kw{skip})}\qquad
	
	\\[\GAP]\squad
	\staterule{Shred2 Sample}
	{\Gamma(L, E_1, \dots, E_n) = \lev{l1} \rightarrow S_1 = L \sim d(E_1, \dots, E_n), S_2 = S_3 = \kw{skip} \\
		\Gamma(L, E_1, \dots, E_n) = \lev{l2} \rightarrow S_2 = L \sim d(E_1, \dots, E_n), S_1 = S_3 = \kw{skip} \\
		\Gamma(L, E_1, \dots, E_n) = \lev{l3} \rightarrow S_3 = L \sim d(E_1, \dots, E_n), S_1 = S_2 = \kw{skip}	}
	{ L \sim d(E_1, \dots, E_n) \shred (S_1, S_2, S_3)}\quad\hquad
	
	\\[\GAP]\squad	
	\staterule{Shred2 If}
	{   S_1 \shred S_{1}^{(1)}, S_{2}^{(1)}, S_{3}^{(1)} \quad
		S_2 \shred S_{1}^{(2)}, S_{2}^{(2)}, S_{3}^{(2)} \quad}
	{ \kw{if}(g)\; S_1\; \kw{else}\; S_2 \shred  
		(\kw{if}(g)\; S_{1}^{(1)}\; \kw{else}\; S_{1}^{(2)}),  
		(\kw{if}(g)\; S_{2}^{(1)}\; \kw{else}\; S_{2}^{(2)}), 
		(\kw{if}(g)\; S_{3}^{(1)}\; \kw{else}\; S_{3}^{(2)})} \qquad	
	
	\\[\GAP]\squad
	\staterule{Shred2 For}
	{   S \shred S_{1}, S_{2}, S_{3}  }
	{ \kw{for}(x\;\kw{in}\;g_1:g_2)\;S \shred  
		(\kw{for}(x\;\kw{in}\;g_1:g_2)\;S_1),  
		(\kw{for}(x\;\kw{in}\;g_1:g_2)\;S_2), 
		(\kw{for}(x\;\kw{in}\;g_1:g_2)\;S_3)} \qquad
\end{display}

As before, shredding produces single-level statements, and shredding preserves semantics with respect to $\typingspecial$-well-typed programs. 

\begin{lemma}[Shredding produces single-level statements, $\typingspecial$] \label{lem:shredisleveled2} ~\\ b
	If $~S \shred[\Gamma] S_1, S_2, S_3$ then $ \singlelevelS{\lev{l1}}{S_1}$, $\singlelevelS{\lev{l2}}{S_2}$, and $\singlelevelS{\lev{l3}}{S_3}$.
\end{lemma}
\begin{lemma}[Semantic preservation of $\shred$, $\typingspecial$] \label{lem:shred2} ~\\
	If $~\Gamma \typingspecial S:\lev{l1} $ and $ S \shred[\Gamma] S_1, S_2, S_3 $ then $\sem{S} = \sem{S_1; S_2; S_3}$.
\end{lemma}

In addition, we derive a result about the effect of single-level statements on the store and weight of $\typingspecial$-well-typed programs.
\begin{lemma}[Property of $\typingspecial$ single-level statements] \label{lem:single-lev-prop2} ~\\
	Let $~\Gamma_{\sigma}, \Gamma_{\mathbf{x}}, S$ be a SlicStan program, and $\Gamma \typingspecial S : \lev{l1}$, and $S$ be single-level statement of level $\ell$, $\Gamma \typingspecial \ell(S)$. Then there exist unique functions $f$ and $\phi$, such that for any $\sigma, \mathbf{x} \models \Gamma_{\sigma}, \Gamma_{\mathbf{x}}$: 
	\begin{enumerate}
	\item If $\ell = \lev{l1}$, then $\sem{S}(\sigma)(x) \hquad = \hquad \left(f(\sigma_{\lev{l1}}, \mathbf{x}_{\lev{l1}}), \sigma_{\lev{l2}}, \sigma_{\lev{l3}} \right), \qquad \qquad \; \phi(\sigma_{\lev{l1}})(\mathbf{x}_{\lev{l1}})$
	\item If $\ell = \lev{l2}$, then $\sem{S}(\sigma)(x) \hquad = \hquad \left(\sigma_{\lev{l1}}, f(\sigma_{\lev{l1}}, \sigma_{\lev{l2}}, \mathbf{x}_{\lev{l1}}, \mathbf{x}_{\lev{l2}}), \sigma_{\lev{l3}} \right), \hquad \phi(\sigma_{\lev{l1}}, \sigma_{\lev{l2}})(\mathbf{x}_{\lev{l1}}, \mathbf{x}_{\lev{l2}})$
	\item If $\ell = \lev{l3}$, then $\sem{S}(\sigma)(x) \hquad = \hquad \left(\sigma_{\lev{l1}}, \sigma_{\lev{l2}}, f(\sigma_{\lev{l1}}, \sigma_{\lev{l3}}, \mathbf{x}_{\lev{l1}}, \mathbf{x}_{\lev{l3}}) \right), \hquad \phi(\sigma_{\lev{l1}}, \sigma_{\lev{l3}})(\mathbf{x}_{\lev{l1}}, \mathbf{x}_{\lev{l3}})$ 
	\end{enumerate}
\end{lemma}

We give proofs for Lemma~\ref{lem:shredisleveled2}, \ref{lem:shred2}, and \ref{lem:single-lev-prop2} in \autoref{ap:proofs}.
These results allows us to derive the second key theorem of this paper, \autoref{th:shred_discrete}, which, similarly to \autoref{th:shred_gen}, gives us a result on the way shredding factorises the density defined by the program. 

Here, and throughout the paper, we use subscripts to refer to specific subsets of $\Gamma$. For example, $\Gamma_{\lev{l1}}$ stands for the subset of the parameters $\Gamma_{\mathbf{x}}$, such that $x:(\tau, \ell) \in \Gamma_{\lev{l1}}$ if and only if $x:(\tau, \ell) \in \Gamma_{\mathbf{x}}$ and $\ell = \lev{l1}$.

\begin{theorem}[Shredding induces a factorisation of the density (2)] \label{th:shred_discrete} ~ \\	
	Suppose $\Gamma \typingspecial S : \lev{l1}$ with 
	$\Gamma = \Gamma_{\sigma}, \Gamma_{\lev{l1}}, \Gamma_{\lev{l2}}, \Gamma_{\lev{l3}}$, 
    $~S \shred[\Gamma] S_1, S_2, S_3$.    
	Then for $\sigma, \params_1, \params_2, \params_3 \models \Gamma_{\sigma}, \Gamma_{1}$, $\Gamma_{2}$, $\Gamma_{3}$,
	and $~\sigma', \sigma''$ such that $\sem{S_1}(\sigma)(\params_1, \params_2, \params_3) = \sigma'$, and $~\sem{S_2}(\sigma')(\params_1, \params_2, \params_3) = \sigma''$ we have:
	\begin{enumerate}
		\item $\semp{S_1}(\sigma)(\params_1, \params_2, \params_3) = \phi_1(\params_1)$
		\item $\semp{S_2}(\sigma')(\params_1, \params_2, \params_3) = \phi_2(\params_1, \params_2)$
		\item $\semp{S_3}(\sigma'')(\params_1, \params_2, \params_3) = \phi_3(\params_1, \params_3)$
	\end{enumerate}
\end{theorem}
%
%
\begin{proof}
By applying \autoref{lem:single-lev-prop2} to each of $S_1, S_2, S_3$, which are single-level statements (\autoref{lem:shredisleveled2}).
\end{proof}

\subsection{Conditional Independence Result for $\typingspecial$-Well-Typed Programs} \label{ssec:ci}

\autoref{th:ci} states the key theoretical result of this paper: the typing in programs well-typed with respect to $\typingspecial$ corresponds to a conditional independence relationship. In our proofs, we use the factorisation characterisation of conditional independence stated by
\autoref{def:ci}. This is a well-known result in the literature (e.g. \cite[Theorem 2.2.1.]{Murphy}). 

\begin{definition}[Characterisation of conditional independence as factorisation] \label{def:ci} ~\\
	For variables $x, y, z$ and a density $p(x, y, z)$, $x$ is conditionally independent of $y$ given $z$ with respect to $p$, written $x \bigCI_p y \mid z$, if and only if
	$~ \exists \phi_1, \phi_2 \text{ such that } p(x, y, z) = \phi_1(x, z) \phi_2(y, z)$.
	
	An equivalent formulation is $p(x, y \mid z) = p(x \mid z)p(y \mid z)$.
	
	We extend the notion of conditional independence to apply to a general function $\phi(x, y, z)$, using the notation  $x \perp_{\phi} y \mid z$ to mean
	$\exists \phi_1, \phi_2 \text{ such that } \phi(x, y, z) = \phi_1(x, z) \phi_2(y, z)$.
\end{definition}

\begin{theorem}[$\typingspecial$-well-typed programs induce a conditional independence relationship] \label{th:ci} ~ \\
	For a SlicStan program $\Gamma, S$ such that 
	$\Gamma \typingspecial S : \lev{l1}$, 
	$\Gamma = \Gamma_\sigma,\Gamma_{\lev{l1}},\Gamma_{\lev{l2}},\Gamma_{\lev{l3}}$,
	and for $\sigma, \params_1, \params_2, \params_3 \models \Gamma_{\sigma}, \Gamma_{\lev{l1}},\Gamma_{\lev{l2}},\Gamma_{\lev{l3}}$,
	we have 
	$\params_2 \perp_{\phi} \params_3 \mid \params_1$.
	
	When $\semp{S}(\sigma)(\params_1, \params_2, \params_3) \propto p(\params_1, \params_2, \params_3)$, we have $\params_2 \bigCI_p \params_3 \mid \params_1$.
\end{theorem}
%
%
\begin{proof}
	Let $\params = \params_1, \params_2, \params_3$, $S \shred S_1, S_2, S_3$, and let $\sigma'$ and $\sigma''$ be such that $\sigma' = \sems{S_1}(\sigma)(\params)$, and $\sigma'' = \sems{S_2}(\sigma')(\params)$.
	Then, by semantic preservation of shredding (\autoref{lem:shred2}), we have 
	\begin{align*}
	\semp{S}(\sigma)(\params) &= \semp{S_1; S_2; S_3}(\sigma)(\params) && \text{by \autoref{lem:shred2}}\\
	&= \semp{S_1}(\sigma)(\params) \times  \semp{S_2}(\sigma')(\params) \times \semp{S_3}(\sigma'')(\params) && \text{by \autoref{lem:sem_properties}} \\
	&= \phi_1(\params_1) \times \phi_2(\params_1, \params_2) \times \phi_3(\params_1, \params_3) && \text{by \autoref{th:shred_discrete}} \\
	&= \phi'(\params_1, \params_2) \times \phi_3(\params_1, \params_3)
	\end{align*}
	for some $\phi_1, \phi_2,$ and $\phi_3$, $\phi'(\params_1, \params_2) = \phi_1(\params_1) \times \phi_2(\params_1, \params_2)$.
	Thus $\params_2 \perp_{\phi} \params_3 \mid \params_1$ by definition of $\perp_{\phi}$.
	
	Suppose $\phi(\params_1, \params_2, \params_3) \propto p(\params_1, \params_2, \params_3)$. Then $p(\params_1, \params_2, \params_3) = \phi(\params_1, \params_2, \params_3) \times Z = \phi'(\params_1, \params_2) \times \phi_3(\params_1, \params_3) \times Z = \phi'(\params_1, \params_2) \times \phi''(\params_1, \params_3)$, where $Z$ is a constant and $\phi''(\params_1, \params_3) = \phi_3(\params_1, \params_3) \times Z$. Therefore, $\params_2 \bigCI_p \params_3 \mid \params_1$.
\end{proof}


\subsection{Scope of the Conditional Independence Result} \label{ssec:scope}

We have shown that $\typingspecial$-well-typed programs exhibit a conditional independence relationship
in their density semantics. However, it is \emph{not} the case that every conditional independence relationship can be derived from the type system. In particular, we can only derive results of the form $\params_2 \bigCI \params_3 \mid \params_1$, where $\params_1, \params_2, \params_3$ is a partitioning of $\params \models \Gamma_{\mathbf{x}}$ for a SlicStan program $\Gamma_{\sigma}, \Gamma_{\mathbf{x}}, S$. That is, the relationship includes \emph{all} parameters in the program.

We discuss the scope of our approach using an example and show a situation where trying to derive a conditional independence result that \emph{does not} hold results in a failure to type check.

\subsubsection{Example of $\typingspecial$-well-typed program $\rightarrow$ conditional independence} ~ \\
Consider the Cross Model in \autoref{fig:cross-model}, its SlicStan program (a), its directed graphical model (b) and the conditional independence (CI) relationships that hold for that model (c). 
\begin{figure}
\textbf{C.\,Cross Model}\vspace{10pt}\\
\begin{subfigure}{0.48\linewidth}
\begin{lstlisting}
 real x1 ~ normal(0,1)
 real x2 ~ normal(0,1)
 real x3 ~ normal(x1+x2,1)
 real x4 ~ normal(x3,1)
 real x5 ~ normal(x3,1)
\end{lstlisting}
\caption{A simple `cross' model.}
\end{subfigure}
\begin{subfigure}{0.48\linewidth}
\scalebox{0.9}{
\begin{tikzpicture}
\node[latent] (x3) {$x_3$};
\coordinate[above=of x3, yshift=-10pt] (midup);
\node[latent, left=of midup] (x1) {$x_1$};
\node[latent, right=of midup] (x2) {$x_2$};
\coordinate[below=of x3, yshift=10pt] (middown);
\node[latent, left=of middown] (x4) {$x_4$};
\node[latent, right=of middown] (x5) {$x_5$};
\draw [->] (x1) -- (x3) node[midway, above] {};
\draw [->] (x2) -- (x3) node[midway, above] {};
\draw [->] (x3) -- (x4) node[midway, above] {};
\draw [->] (x3) -- (x5) node[midway, above] {};
\end{tikzpicture}}
\caption{Graphical model.}
\end{subfigure}

\begin{subfigure}{\linewidth}
\small
\begin{align*}
&x_1 \bigCI x_2 
&x_1 \bigCI x_4 \mid \{x_3\} \cup A, \forall A \subseteq \{x_2, x_5\} \\
&x_1 \bigCI x_5 \mid \{x_3\} \cup A, \forall A \subseteq \{x_2, x_4\} 
&x_2 \bigCI x_4 \mid \{x_3\} \cup A, \forall A \subseteq \{x_1, x_5\} \\
&x_2 \bigCI x_5 \mid \{x_3\} \cup A, \forall A \subseteq \{x_1, x_4\} 
&x_4 \bigCI x_5 \mid \{x_3\} \cup A, \forall A \subseteq \{x_1, x_2\} 
\end{align*}
\caption{CI relationships.}
\end{subfigure}
\caption{\label{fig:cross-model} The cross model, as written in SlicStan (a) with its DAG (b) and CI relationships (c). }
\end{figure}

Out of the many relationships above, we can derive \emph{all} relationships that involve \emph{all} the variables. That is, we can use our type system to derive all conditional independence relationships that hold and are of the form $A \bigCI B \mid C$, where $A, B, C$ is some partitioning of $\{x_1, \dots, x_5\}$. 
However, note the following properties of conditional independence:
$$A \bigCI B \mid C \iff B \bigCI A \mid C \quad \text{ and } \quad A \bigCI B_1, B_2 \mid C \iff A \bigCI B_1 \mid C \text{ and } A \bigCI B_2 \mid C$$

Some of the relationships above can be combined and written in other ways, e.g. $x_1 \bigCI x_4 \mid x_2, x_3$ and $x_1 \bigCI x_5 \mid x_2, x_3$ can be written as a single relationship $x_1 \bigCI x_4, x_5 \mid x_2, x_3$, thus expressing them as a single relationship that includes all variables in the program. 

Exploring different mappings between the parameters $x_1, \dots, x_5$ and the type levels $\lev{l1}, \lev{l2}, \lev{l3}$, for which the above program typechecks, we can derive \emph{all} CI relationships that hold for this model, except for one: $x_1 \bigCI x_2$, which we cannot derive with our approach.  


\subsubsection{Conditional independence relationship does not hold $\rightarrow$ type error}  ~ \\
Suppose that we try to derive the result $x_1 \bigCI x_2 \mid x_3, x_4, x_5$. This \emph{does not hold} for Program C. By \autoref{th:ci}, we have that a program being $\typingspecial$-well-typed implies that $\lev{l2} \bigCI \lev{l3} \mid \lev{l1}$. So, we can derive $x_1 \bigCI x_2 \mid x_3, x_4, x_5$ using \autoref{th:ci} if we show that $\Gamma \typingspecial S: \lev{l1}$, for $\Gamma = \{ x_1 : \lev{l2}, x_2 : \lev{l3}, x_3 : \lev{l1}, x_4 : \lev{l1}, x_5 : \lev{l1} \}$ and $S$ being Program C. 

To typecheck $\Gamma\typingspecial S:\lev{l1}$, we need to typecheck $x_3 \sim \kw{normal}(x_1+x_2,1)$ at some level $\ell$. Thus, by \ref{Sample2} and \ref{PrimCall2}, $x_1$, $x_2$ and $x_3$ need to typecheck at $\ell$. The types of $x_1$, $x_2$ and $x_3$ are $\lev{l2}, \lev{l3}$ and $\lev{l1}$, respectively. So, using \ref{ESub2}, it must be the case that
$ \lev{l2} \leq \ell$, \emph{and} $\lev{l3} \leq \ell$, \emph{and} $\lev{l1} \leq \ell$.
However, no such level exists in our lower semi-lattice, as $\lev{l2}$ and $\lev{l3}$ have no upper bound. Therefore, typechecking fails and  we cannot derive $x_1 \bigCI x_2 \mid x_3, x_4, x_5$.







\section{Application: Discrete parameters support through a semantics-preserving transformation} \label{sec:application}

This section presents the main practical contribution of our work: a semantics-preserving procedure for transforming a probabilistic program to enable combined inference of discrete and continuous model parameters, which we have implemented for SlicStan.
The procedure corresponds to variable elimination (VE) for discrete parameters implemented in the probabilistic program itself, which can be combined with gradient-based methods, such as HMC, to perform inference on all parameters. 

PPLs that have gradient-based methods in the core of their inference strategy do not, in general, support directly working with discrete parameters. Stan disallows discrete model parameters altogether, while Pyro \cite{Pyro} and Edward2 \cite{Edward2} throw a runtime error whenever discrete parameters are used within a gradient-based method. 
%
However, working with discrete parameters in these languages is still possible, albeit in an implicit way. In many cases, discrete parameters can be marginalised out manually, and then drawn conditionally on the continuous parameters. Stan's user guide shows many examples of this approach \cite[Chapter 7]{StanManual}.
Pyro provides an on-request marginalisation functionality, which automates this implicit treatment for plated factor graphs \cite{PyroDiscrete}.

The key idea of the workaround is to marginalise out the discrete parameters by hand, so that the resulting program corresponds to a density function that does not depend on any discrete parameters. That is,
the user writes a program that computes $\sum_{\params_d}p(\params_d, \params_c) = p(\params_c)$,
where the density semantics of the original program was $p(\params_d, \params_c)$ for
 discrete parameters $\params_d$ and continuous parameters $\params_c$.
This allows for continuous parameters of the program to be sampled with HMC, or other gradient-based inference algorithms, whereas that would have not been possible for the program with both discrete and continuous latent variables.

Because a SlicStan program computes a density directly, it is easy to modify it to marginalise a variable. For a SlicStan program $\Gamma, S$, with parameters $\mathbf{x} \models \Gamma_{\mathbf{x}}$, and a discrete parameter $z$ of type $\kw{int}\langle K \rangle$, the program $\kw{elim}(\kw{int}\langle K \rangle z)~S \deq \kw{factor}(\mathrm{sum}([\kw{target}(S) \mid z~\kw{in}~ 1 : K])$\footnote{Here, we assume the function \textit{sum} is available in the language.} marginalises $z$:
\begin{align*}
\semp{\kw{factor}(\mathrm{sum}([\kw{target}(S) \mid z~\kw{in}~ 1 : K]))}(\sigma)(\mathbf{x}) = 
\sum_{z = 1}^{K} \semp{S}(\sigma)(\mathbf{x})  
\propto \sum_{z = 1}^{K} p(\mathbf{x}) = p(\mathbf{x} \setminus \{z\})
\end{align*}


In other words, we can easily marginalise out all discrete variables in a probabilistic program, by encapsulating the entire program in nested loops (nested array comprehension expressions in our examples). However, this approach becomes infeasible for more than a few variables. Variable elimination \cite{zhang1994simple, koller2009probabilistic} exploits the structure of a model to do as little work as possible. 
Consider the HMM snippet (Program D) with three discrete (binary) hidden variables $z_1$, $z_2$ and $z_3$, and observed outcomes $y_1$, $y_2$ and $y_3$.
Naively marginalising out the hidden variables results in nested loops around the original program (Program E). In the general case of $N$ hidden variables, the resulting program is of complexity $O(2^N)$. 

\begin{figure*}
\begin{multicols}{2} \label{hmm}
\centering
\textbf{D.\,A Hidden Markov Model (HMM)}\vspace{-12pt}
\begin{lstlisting}
		...
		int<2> z1 ~ bernoulli(theta[1]);
		int<2> z2 ~ bernoulli(theta[z1]);
		int<2> z3 ~ bernoulli(theta[z2]);
		data real y1 ~ normal(phi[z1], 1);
		data real y2 ~ normal(phi[z2], 1);
		data real y3 ~ normal(phi[z3], 1);		
\end{lstlisting}

\textbf{E.\,Inefficient marginalisation}\vspace{-12pt}
\begin{lstlisting}
		...
		factor(sum [target(
		  factor(sum [target(
		    factor(sum [target(
			      z1 ~ bernoulli(theta[1]);
			      z2 ~ bernoulli(theta[z1]);
			      z3 ~ bernoulli(theta[z2]);
			      y1 ~ normal(phi[z1], 1);
			      y2 ~ normal(phi[z2], 1);
			      y3 ~ normal(phi[z3], 1);)
		    | z1 in 1:2]);
		  | z2 in 1:2]);
		| z3 in 1:2]);
\end{lstlisting}
\vspace{3cm}

\textbf{F.\,Efficient marginalisation}\vspace{-12pt}
\begin{lstlisting}
		...
		real[2] f1 = // new factor on z2
		  [sum([target(
			        z1 ~ bernoulli(theta[1]);
			        z2 ~ bernoulli(theta[z1]);
			        y1 ~ normal(phi[z1], 1); ) 
			     | z1 in 1:2])
		  | z2 in 1:2] 	
		
		real[2] f2 = // new factor on z3
		  [sum([target(
			        factor(f1[z2]);
			        y2 ~ normal(phi[z2], 1);
			        z3 ~ bernoulli(theta[z2]); )
			     | z2 in 1:2])
		  | z3 in 1:2] 	
			     
		factor(sum [target(
		    factor(f2[z3]);
		    y3 ~ normal(phi[z3], 1); )
		| z3 in 1:2]);
\end{lstlisting}
\end{multicols}
\vspace{-20pt}
\end{figure*}

However, this is wasteful: expressions like $z_3 \sim \mathrm{bernoulli}(\theta[z_2])$ do not depend on $z_1$, and so do not need to be inside of the $z_1$-elimination loop. Variable elimination (VE) avoids this problem by pre-computing some of the work.
Program F implements VE for this model: when eliminating a variable, say $z_1$, we pre-compute 
statements that involve $z_1$ for each possible value of $z_1$ and store the resulting density contributions in a new factor, $f_1$. This new factor depends on the variables involved in those statements --- the neighbours of $z_1$ --- in this case that is solely $z_2$. We then repeat the procedure for the other variables, re-using the already computed factors where possible.

In the special case of an HMM, and given a suitable elimination order, variable elimination recovers the celebrated forward algorithm \cite{RabinerTutorial}, which has time complexity $O(N)$.
Our goal is to automatically translate the source code of Program D to Program F, exploiting statically detectable independence properties in the model.

\subsection{Goal}


Our ultimate goal is to transform a program $S$ with continuous parameters $\params_c$, discrete parameters $\params_d$, data $\data$ and density semantics $\semp{S}(\sigma)(\params_d, \params_c, \data) \propto p(\params_d, \params_c \mid \data)$, into two subprograms: $S_{\textsc{hmc}}$ and $S_{\textsc{gen}}$, such that: 
\begin{itemize}
	\item The density defined by $~S_{\textsc{hmc}}$ is the marginal $p(\params_c \mid \data)$, with the discrete parameters $\params_d$ marginalised out. This first statement, $S_{\textsc{hmc}}$, represents the marginalisation part of the program (see \autoref{ssec:ve}) and allows for \textbf{Hamiltonian Monte Carlo} (HMC) sampling of $\params_c$, as it does not involve any discrete parameters. 
	
	\item The density defined by $S_{\textsc{gen}}$ is the conditional $p(\params_d \mid \params_c, \data)$. This second statement, $S_{\textsc{gen}}$, represents the generative part of the program (\autoref{ssec:bs}) and it encodes a way to draw $\params_d$ generatively, without using HMC or another heavy-weight inference algorithm.
\end{itemize}

Similarly to the extended SlicStan slicing based on information-flow type inference, here we also want to transform and slice into sub-programs, each focusing on a subset of the parameters, and preserving the overall meaning:
$$\semp{S} \propto p(\params_d, \params_c \mid \data) = p(\params_c \mid \data) \times p(\params_d \mid \params_c, \data) \propto \semp{S_{\textsc{hmc}}}\times\semp{S_{\textsc{gen}}} = \semp{S_{\textsc{hmc}}; S_{\textsc{gen}}}\footnotemark$$
\footnotetext{This expression is simplified for readability.
}

Our approach performs a semantics-preserving transformation, guided by information-flow and type inference, which creates an efficient program-specific inference algorithm  automatically, combining HMC with variable elimination.

\subsection{Key Insight}

The key practical insight of this work is to use the adaptation of SlicStan's level types of \autoref{sec:theory} and its information flow type system to rearrange the program in a semantics-preserving way, so that discrete parameters can be forward-sampled, instead of sampled using a heavy-weight inference algorithm. We achieve this by a program
transformation for each of the discrete variables. Assuming  that we are applying the transformation with respect to a variable $z$, we use: 
\begin{itemize}
	\item The \emph{top-level information flow type system} $\Gamma \vdash S : \lev{data}$ from \autoref{ssec:typing}, which involves the level types $\lev{data} \leq \lev{model} \leq \lev{genquant}$. This partitions the modelled variables $\mathbf{x}$ into data $\data$, model parameters $\params$ and generated quantities $Q$. When we use type inference for $\vdash$ in conjunction with shredding $S \shred S_D, S_M, S_Q$ (\autoref{ssec:shred}), we slice the statement $S$ into a data part $S_D$ (involving only variables in $\data$), a non-generative part $S_M$ (involving $\data$ and $\params$) and a generative part $S_Q$ (involving $\data$, $\params$ and $Q$). 
	\item The \emph{conditional independence information flow type system}, $\Gamma \typingspecial S : \lev{l1}$ from \autoref{sec:theory}, which uses a lower semi-lattice of level types $\lev{l1} \leq \lev{l2}$, $\lev{l1} \leq \lev{l3}$. A $\typingspecial$-well-typed program induces a conditional independence relationship: $\lev{l2}$-variables are conditionally independent of $\lev{l3}$-variables given $\lev{l1}$-variables: $\mathbf{x}_{\lev{l2}} \bigCI \mathbf{x}_{\lev{l3}} \mid \mathbf{x}_{\lev{l1}}$, where $\mathbf{x} = \mathbf{x}_{\lev{l1}}, \mathbf{x}_{\lev{l2}}, \mathbf{x}_{\lev{l3}} = \params, \data$. When we use type inference for $\typingspecial$ in conjunction with shredding $S \shred S_1, S_2, S_3$ (\autoref{ssec:shred}), we isolate $S_2$: a part of the program that does not interfere with $S_3$. We can marginalise out $\lev{l2}$-variables in that sub-statement only, keeping the rest of the program unchanged.
	\item The \emph{discrete variable transformation relation} $\Gamma, S \xrightarrow{z} \Gamma', S'$ (defined in \autoref{sssec:elim-gen}), which takes a SlicStan program $\Gamma, S$ that has discrete model parameter $z$, and transforms it to a SlicStan program $\Gamma', S'$, where $z$ is no longer a $\lev{model}$-level parameter but instead one of level $\lev{genquant}$. We define the relation in terms of $\vdash$ and $\typingspecial$ as per the \ref{Elim Gen} rule.
\end{itemize}

\subsection{Variable Elimination} \label{ssec:ve}

Variable elimination (VE) \cite{zhang1994simple, koller2009probabilistic} is an exact inference algorithm 
often phrased in terms of factor graphs.
It can be used to compute prior or posterior marginal distributions by eliminating, one by one, variables that are irrelevant to the distribution of interest.
VE uses  dynamic programming combined with a clever use of the distributive law of multiplication over addition to efficiently compute 
a nested sum of a product of expressions.

We already saw an example of variable elimination in $\autoref{sec:theory}$ (Programs A and B). The idea is to eliminate (marginalise out) variables one by one. To eliminate a variable $z$, we multiply all of the factors connected to $z$ to form a single expression, then sum over all possible values for $z$ to create a new factor, remove $z$ from the graph, and finally connect the new factor to all former neighbours\footnotemark of $z$. 
\footnotetext{`Neighbours' refers to the variables which are connected to a factor
which connects to $z$.}
Recall Program D, with latent variables $z_1, z_2, z_3$ and observed data $\mathbf{y} = y_1, y_2, y_3$.  
\autoref{fig:varelim} shows the VE algorithm step-by-step applied to this program. We eliminate $z_1$ to get the marginal on $z_2$ and $z_3$ (\ref{fig:varelim_a} and \ref{fig:varelim_b}), then eliminate $z_2$ to get the marginal on $z_3$ (\ref{fig:varelim_c} and \ref{fig:varelim_d}).

\begin{figure*}
	\centering 
	\begin{subfigure}{0.45\textwidth} 
	\hspace{-7pt} 
	\begin{tikzpicture}
	\node[latent, fill=seabornred!20]   (d1)   {$z_1$}; %
	\factor[left=of d1, xshift=-0.4cm, fill=seaborndarkred]   {p1} {\color{seaborndarkred}\scriptsize $b(z_1 \mid \theta_1)$} {} {};
	\factor[below=of d1, yshift=0.2cm, fill=seaborndarkred]  {yp1} {left:{\color{seaborndarkred}\scriptsize $b(y_1 \mid \phi_{z_1})$}} {} {};
	\node[obs, below=of yp1, yshift=0.6cm,minimum size=12pt] (y1) {\scriptsize$y_1$};
	\factor[right=of d1, xshift=0.1cm, fill=seaborndarkred] {p2} {\color{seaborndarkred}\scriptsize $b(z_2 \mid \theta_{z_1})$} {} {}; %
	\node[latent, right=of p2, xshift=-0.2cm]   (d2)   {$z_2$}; %
	\factor[below=of d2, yshift=0.2cm]  {yp2} {left:{\scriptsize $b(y_2 \mid \phi_{z_2})$}} {} {};
	\node[obs, below=of yp2, yshift=0.6cm,minimum size=12pt] (y2) {\scriptsize$y_2$};
	\factor[right=of d2, xshift=0.1cm] {p3} {\scriptsize $b(z_3 \mid \theta_{z_2})$} {} {}; %
	\node[latent, right=of p3, xshift=-0.2cm]   (d3)   {$z_3$}; %
	\factor[below=of d3, yshift=0.2cm]  {yp3} {left:{\scriptsize $b(y_3 \mid \phi_{z_3})$}} {} {};
	\node[obs, below=of yp3, yshift=0.6cm,minimum size=12pt] (y3) {\scriptsize$y_3$};
	
	\factoredge [] {} {p1} {d1};
	\factoredge [] {d1} {p2} {d2};
	\factoredge [] {d2} {p3} {d3};
	\factoredge [] {d1} {yp1} {y1};
	\factoredge [] {d2} {yp2} {y2};
	\factoredge [] {d3} {yp3} {y3};	
	\end{tikzpicture}
	\caption{To eliminate $z_1$, we remove $z_1$ and all its neighbouring factors (in red). Create a new factor $f_1$, by summing out $z_1$ from the product of these factors.}
	\label{fig:varelim_a}
	\end{subfigure}
	\hspace{0.08\textwidth} 
	\begin{subfigure}{0.45\textwidth} 		
		\begin{tikzpicture}		
		\factor[fill=seaborndarkgreen] {d1} { } {} {};
		
		\node[left=of d1, align=right, xshift=1cm, yshift=-0.3cm] (hidden) {
			\color{seaborndarkgreen}\scriptsize $f_1(z_2, y_1) = $ \\
			\color{seaborndarkgreen}\scriptsize $\sum_{z_1} \left[b(z_1 \mid \theta_1)\right.$ \\
			\color{seaborndarkgreen}\scriptsize $\times b(y_1 \mid \phi_{z_1}) $ \\
			\color{seaborndarkgreen}\scriptsize $\left.\times b(z_2 \mid \theta_{z_1}) \right]$
		};
			
		\node[latent, right=of d1]   (d2)   {$z_2$}; %
		\factor[below=of d2, yshift=0.2cm]  {yp2} {left:{\scriptsize $b(y_2 \mid \phi_{z_2})$}} {} {};
		\node[obs, below=of yp2, yshift=0.6cm,minimum size=12pt] (y2) {\scriptsize$y_2$};
		\node[obs, below=of d1, left=of y2, minimum size=12pt] (y1) {\scriptsize$y_1$};
		\factor[right=of d2, xshift=0.1cm] {p3} {\scriptsize $b(z_3 \mid \theta_{z_2})$} {} {}; %
		\node[latent, right=of p3, xshift=-0.2cm]   (d3)   {$z_3$}; %
		\factor[below=of d3, yshift=0.2cm]  {yp3} {left:{\scriptsize $b(y_3 \mid \phi_{z_3})$}} {} {};
		\node[obs, below=of yp3, yshift=0.6cm,minimum size=12pt] (y3) {\scriptsize$y_3$};
		
		\factoredge [] {d2,y1} {d1} {};
		\factoredge [] {d2} {p3} {d3};
		\factoredge [] {d2} {yp2} {y2};
		\factoredge [] {d3} {yp3} {y3};	
		\end{tikzpicture}
		\caption{Connect $f_1$ (in green) to the former neighbours of $z_1$. The remaining factor graph defines the marginal $p(z_2, z_3 \mid \mathbf{y})$.}
		\label{fig:varelim_b}
	\end{subfigure}

	\begin{subfigure}{0.45\textwidth}
		\hspace{29pt} 
		\begin{tikzpicture}
		\factor[fill=seaborndarkred] {d1} {left:{\color{seaborndarkred}\scriptsize $f_1(z_2, y_1)$}} {} {};		
		\node[latent, right=of d1, fill=seabornred!20]   (d2)   {$z_2$}; %
		\factor[below=of d2, yshift=0.2cm, fill=seaborndarkred]  {yp2} {left:{\color{seaborndarkred}\scriptsize $b(y_2 \mid \phi_{z_2})$}} {} {};
		\node[obs, below=of yp2, yshift=0.6cm,minimum size=12pt] (y2) {\scriptsize$y_2$};
		\node[obs, below=of d1, left=of y2, minimum size=12pt] (y1) {\scriptsize$y_1$};
		\factor[right=of d2, xshift=0.1cm, fill=seaborndarkred] {p3} {\color{seaborndarkred}\scriptsize $b(z_3 \mid \theta_{z_2})$} {} {}; %
		\node[latent, right=of p3, xshift=-0.2cm]   (d3)   {$z_3$}; %
		\factor[below=of d3, yshift=0.2cm]  {yp3} {left:{\scriptsize $b(y_3 \mid \phi_{z_3})$}} {} {};
		\node[obs, below=of yp3, yshift=0.6cm,minimum size=12pt] (y3) {\scriptsize$y_3$};
		
		\factoredge [] {d2,y1} {d1} {};
		\factoredge [] {d2} {p3} {d3};
		\factoredge [] {d2} {yp2} {y2};
		\factoredge [] {d3} {yp3} {y3};	
		\end{tikzpicture}
		\caption{To eliminate $z_2$, we remove $z_2$ and all its neighbouring factors (in red). Create a new factor $f_2$, by summing out $z_2$ from the product of these factors.}
		\label{fig:varelim_c}
	\end{subfigure}	
	\hspace{0.08\textwidth}
	\raggedright 
	\begin{subfigure}{0.45\textwidth} 
		\begin{tikzpicture}
		\factor[fill=seaborndarkgreen]  {d2} {} {} {};
		
		\node[left=of d2, align=right, xshift=1cm, yshift=-0.3cm] (hidden) {
			\color{seaborndarkgreen}\scriptsize $f_2(z_3, y_2, y_1) = \sum_{z_2} \left[ f_1(z_2, y_1)\right.$ \\
			\color{seaborndarkgreen}\scriptsize $\left.\times b(y_2 \mid \phi_{z_2}) \times b(z_3 \mid \theta_{z_2}) \right]$
		};
		
		\node[latent, right=of d2, xshift=0.7cm]   (d3)   {$z_3$}; %
		\factor[below=of d3, yshift=0.2cm]  {yp3} {left:{\scriptsize $b(y_3 \mid \phi_{z_3})$}} {} {};
		\node[obs, below=of yp3, yshift=0.6cm,minimum size=12pt] (y3) {\scriptsize$y_3$};
		\node[obs, left=of y3 ,minimum size=12pt] (y2) {\scriptsize$y_2$};
		\node[obs, left=of y2, minimum size=12pt] (y1) {\scriptsize$y_1$};
		
		\factoredge [] {d3,y2,y1} {d2} {};
		\factoredge [] {d3} {yp3} {y3};	
		\end{tikzpicture}
		\caption{Connect $f_2$ (in green) to the former neighbours of $z_2$. The remaining factor graph defines the marginal $p(z_3 \mid \mathbf{y})$.}
		\label{fig:varelim_d}
	\end{subfigure}
	\caption{Step by step example of variable elimination.}
	\label{fig:varelim}
\end{figure*}

\subsection{Conditional Independence Relationships and Inferring the Markov Blanket} \label{ssec:mb}

The key property we are looking for, in order to be able to marginalise out a variable independently of another, is conditional independence given neighbouring variables. 
If we shred a $\typingspecial$-well-typed program into $S_1, S_2$ and $S_3$, and think of $\semp{S_1}, \semp{S_2}$ and $\semp{S_3}$ as factors, it is easy to visualise the factor graph corresponding to the program: it is as in \autoref{fig:varelim_levels_a}. Eliminating all $\mathbf{x}_{\lev{l2}}$ variables, ends up only modifying the $\semp{S_2}$ factor (\autoref{fig:varelim_levels_b}).

When using VE to marginalise out a parameter $z$, we want to find the smallest set of other parameters $A$, such that $z \bigCI B \mid A$, where $B$ is the rest of the parameters. The set $A$ is also called $z$'s minimal \emph{Markov blanket} or \emph{Markov boundary}. 
Once we know this set, we can ensure that we involve the smallest possible number of variables in $z$'s elimination, which is important to achieve a performant algorithm. 


For example, when we eliminate $z_1$ in Program D, both $z_2$ and $y_1$ need to be involved, as $z_1$ shares a factor with them. By contrast, there is no need to include $y_2, z_3, y_3$ and the statements associated with them, as they are unaffected by $z_1$, given $z_2$. The variables $y_1$ and $z_2$ form $z_1$'s Markov blanket: given these variables, $z_1$ is conditionally independent of all other variables. That is, $z_1 \bigCI z_3, y_2, y_3 \mid z_2, y_1$.

The type system we present in $\autoref{sec:theory}$ can tell us if the conditional independence relationship $\mathbf{x}_{\lev{l2}} \bigCI \mathbf{x}_{\lev{l3}} \mid \mathbf{x}_{\lev{l1}}$ holds for a concrete partitioning of the modelled variables $\mathbf{x} = \mathbf{x}_{\lev{l1}}, \mathbf{x}_{\lev{l2}}, \mathbf{x}_{\lev{l3}}$.
But to find the Markov blanket of a variable $z$ we want to eliminate, we rely on type inference. We define a performance ordering between the level types $\lev{l3} \prec \lev{l1} \prec \lev{l2}$, where our first preference is for variables to be of level $\lev{l3}$,  level $\lev{l1}$ is our second preference, and $\lev{l2}$ is our last resort. 
In our implementation, we use bidirectional type-checking \cite{pierce2000local} to synthesise hard constraints imposed by the type system, and resolve them, while optimising for the soft constraints given by the $\prec$ ordering. This maximises the number of variables that are conditionally independent of $z$ given its blanket ($\lev{l3}$) and minimises the number of variables forming the blanket ($\lev{l1}$). 
Fixing $z$ to be of $\lev{l2}$ level, and $\lev{l2}$ being the least preferred option, ensures that only $z$ and variables dependent on $z$ through deterministic assignment are of that level.  

\begin{figure}
	\begin{subfigure}{0.49\textwidth}
		\centering
		\begin{tikzpicture}
		\node[latent] (l1) {$\mathbf{x}_{\lev{\small l1}}$};
		\factor[above=of l1] {s1} {right:{$\semp{S_1}$}} {} {};
		
		\factor[left=of l1, fill=seaborndarkred] {s2} {below:{\color{seaborndarkred}$\semp{S_2}$}} {} {};
		\node[latent, left=of s2, fill=seabornred!20] (lz) {$\mathbf{x}_{\small \lev{l2}}$};
		
		\factor[right=of l1]  {s3} {below:{$\semp{S_3}$}} {} {};
		\node[latent, right=of s3] (l3) {$\mathbf{x}_{\small \lev{l3}}$};
		
		\factoredge [] {} {s1} {l1};
		\factoredge [] {l1} {s2} {lz};
		\factoredge [] {l1} {s3} {l3};
		\end{tikzpicture}
		\caption{A $\typingspecial$-well-typed program with parameters $\mathbf{x}$.}
		\label{fig:varelim_levels_a}
	\end{subfigure}
	\begin{subfigure}{0.49\textwidth}
		\centering
		\begin{tikzpicture}
		\node[latent] (l1) {$\mathbf{x}_{\small \lev{l1}}$};
		\factor[above=of l1] {s1} {right:{$\semp{S_1}$}} {} {};
		
		\factor[left=of l1, fill=seaborndarkgreen] {s2} {below:{\color{seaborndarkgreen}$
				\sum_{\mathbf{x}_{\lev{\tiny l2}}} \semp{S_2}$}} {} {};
		
		\factor[right=of l1]  {s3} {below:{$\semp{S_3}$}} {} {};
		\node[latent, right=of s3] (l3) {$\mathbf{x}_{\small \lev{l3}}$};
		
		\factoredge [] {} {s1} {l1};
		\factoredge [] {l1} {s2} {};
		\factoredge [] {l1} {s3} {l3};
		\end{tikzpicture}
		\caption{Eliminating $\mathbf{x}_{\lev{l2}}$ consists of modifying only $\semp{S_2}$.}
		\label{fig:varelim_levels_b}
	\end{subfigure}
	\caption{The factor graph and VE induced by the shedding $S\shred S_1,S_2,S_3$ according to the semi-lattice $\lev{l1}\leq \lev{l2},\lev{l3}$.}
	\label{fig:varelim_levels}
\end{figure}

\subsection{Sampling the Discrete Parameters} \label{ssec:bs}

Variable elimination gives a way to efficiently marginalise out a variable $z$ from a model defining density $p(\mathbf{x})$, to obtain a new density $p(\mathbf{x} \setminus \{z\})$. In the context of SlicStan, this means we have the tools to eliminate all discrete parameters $\params_d$, from a density $p(\data, \params_c, \params_d)$ on data $\data$, continuous parameters $\params_c$ and discrete parameters $\params_d$. The resulting marginal $\sum_{\params_d}p(\data, \params_c, \params_d) = p(\data, \params_c)$ does not involve discrete parameters, and therefore we can use gradient-based methods to infer $\params_c$. However, the method so far does not give us a way to infer the discrete parameters $\params_d$.

To infer these, we observe that 
$p(\mathbf{x}) = p(\mathbf{x} \setminus \{z\} ) p(z \mid \mathbf{x} \setminus \{z\})$, which means that we can preserve the semantics of the original model (which defines $p(\mathbf{x})$), by finding an expression for the conditional $p(z \mid \mathbf{x} \setminus \{z\})$. If $\mathbf{x}_1, \mathbf{x}_2$ is a partitioning of $\mathbf{x} \setminus \{z\}$ such that $z \bigCI \mathbf{x}_2 \mid \mathbf{x}_1$, then (from \autoref{def:ci}) $p(\mathbf{x}) = \phi_1(z, \mathbf{x}_1)\phi_2(\mathbf{x}_1, \mathbf{x}_2)$ for some functions $\phi_1$ and $\phi_2$. Thus,
$p(z \mid \mathbf{x} \setminus \{z\}) =
\phi_1(z, \mathbf{x}_1) \cdot 
\left(\phi_2(\mathbf{x}_1,\mathbf{x}_2)/p(\mathbf{x} \setminus \{z\})\right) \propto \phi_1(z, \mathbf{x}_1)$.

In the case when $z$ is a discrete variable of finite support, we can calculate the conditional probability exactly: $p(z \mid \mathbf{x} \setminus \{z\}) = \frac{\phi_1(z, \mathbf{x}_1)}{\sum_z \phi_1(z, \mathbf{x}_1)}$.
We can apply this calculation to the factorisation of a program $\Gamma\typingspecial S$ that is induced 
by shredding (\autoref{th:shred_discrete}).
In that case,  $\mathbf{x}_{\lev{l2}}$, $\mathbf{x}_{\lev{l1}}$, $\semp{S_2}$ play the roles of $z$, $\mathbf{x}_1$, and
  $\phi_1$, respectively. Consequently, we obtain a formula for 
 drawing $\mathbf{x}_{\lev{l2}}$ conditional on the other parameters:
 $\mathbf{x}_{\lev{l2}} \sim \mathrm{categorical}\left(\left[ \frac{\semp{S_2}(\mathbf{x}_{\tiny\lev{l2}},  \mathbf{x}_{\tiny\lev{l1}} )}{\sum_{\mathbf{x}_{\tiny\lev{l2}}} \semp{S_2}(\mathbf{x}_{\tiny\lev{l2}},  \mathbf{x}_{\tiny\lev{l1}} )} \mid \mathbf{x}_{\tiny\lev{l2}} \in \mathrm{supp}(\mathbf{x}_{\lev{l2}}) \right]\right)$.


\subsection{A Semantics-Preserving Transformation Rule}

In this section we define a source-to-source transformation that implements a single step of variable elimination. 
The transformation re-writes a SlicStan program $\Gamma, S$ with a discrete $\lev{model}$-level parameter $z$, to a SlicStan program, where $z$ is a $\lev{genquant}$-level parameter. Combining the rule with the shredding presented in $\autoref{sec:background}$ results in support for efficient inference (see \autoref{ssec:ve_relation} for discussion of  limitations) of both discrete and continuous random variables, where continuous variables can be inferred using gradient-based methods, such as HMC or variational inference, while discrete variables are generated using ancestral sampling. The transformation allows for SlicStan programs with explicit use of discrete parameters to be translated to Stan. We show a step-by-step example of our discrete parameter transformation in \autoref{sssec:elimgen_example}.

\subsubsection{The $\phi$, \kw{elim} and \kw{gen} derived forms} ~ \\
We introduce three derived forms that allow us to state the rule concisely.  
\begin{display}[.59]{Variable Elimination Derived Forms}
	\clause{\kw{elim}(\kw{int}\langle K \rangle~z)~S \deq \kw{factor}(\kw{sum}(\,[\kw{target}(S) \mid z~\kw{in}\,1\!:\!K]\,))}{
		} \\ 
	\clause{\phi(\kw{int}\langle K_1 \rangle~z_1, \dots, \kw{int}\langle K_N \rangle~z_N)~S}{
		} \\ 
	\clause{\qquad \deq  [\dots[\kw{target}(S) \mid z_1~\kw{in}~1:K_1] \mid \dots \mid z_N~\kw{in}~1 : K_N ]}{ } \\ 
	\clause{\kw{gen}(\kw{int}\langle K \rangle\,z)~S \deq z \sim \kw{categorical}(\,[\kw{target}(S) \mid z~\kw{in}~1:K]\,)}{
		}
\end{display} 


The elimination expression $\kw{elim}(\kw{int}\langle K\rangle z)~S$ adds a new factor that is equivalent to marginalising $z$ in $S$. In other words, $\semp{\kw{elim}(\kw{int}\langle K\rangle z)~S}(\sigma)(\mathbf{x}) = \sum_{z = 1}^{K}\semp{S}(\sigma)(\mathbf{x})$ (see \autoref{lem:s2}).
A $\phi$-expression $\phi(\kw{int}\langle K_1 \rangle~z_1, \dots, \kw{int}\langle K_N \rangle~z_N)~S$ simply computes the density of the statement $S$ in a multidimensional array for all possible values of the variables $z_1, \dots z_N$. In other words, $\semp{\left(f = \phi(\kw{int}\langle K_1 \rangle~z_1, \dots, \kw{int}\langle K_N \rangle~z_N)~S\right); \kw{factor}(f[z_1] \dots [z_N])}(\sigma)(\mathbf{x}) = \semp{S}(\sigma)(\mathbf{x})$ (\autoref{lem:s2}). The $\phi$-expression allows us to pre-compute all the work that we may need to do when marginalising other discrete variables, which results in efficient nesting. 
%
%
Finally, the generation expression computes the conditional of a variable $z$ given the rest of the parameters, as in \autoref{ssec:bs} (see \autoref{lem:sgen}).

\subsubsection{Eliminating a single variable $z$}\label{sssec:elim-gen}
The \ref{Elim Gen} rule below specifies a semantics-preserving transformation that takes a SlicStan program with a discrete $\lev{model}$-level parameter $z$, and transforms it to one where $z$ is $\lev{genquant}$-level parameter. 
In practice, we apply this rule once per discrete $\lev{model}$-level parameter, which eliminates those parameters one-by-one, similarly to the variable elimination algorithm. And like in VE, the ordering in which we eliminate those variables can impact performance. 

The  \ref{Elim Gen} rule makes use of two auxiliary definitions that we define next.
Firstly, the neighbours of $z$, $\Gamma_{\mathrm{ne}}$, are defined by the relation $\mathrm{ne}(\Gamma, \Gamma', z)$ (\autoref{def:ne}), which looks for non-data and non-continuous $\lev{l1}$-variables in $\Gamma'$.
\begin{definition}[Neighbours of $z$, $\mathrm{ne}(\Gamma, \Gamma', z)$] \label{def:ne}
	For a $\vdash$ typing environment $\Gamma$, a $\typingspecial$ typing environment $\Gamma' = \Gamma_{\sigma}', \Gamma_{\mathbf{x}}'$ and a variable $z \in \dom(\Gamma_{\mathbf{x}}')$, the neighbours of $z$ are defined as: 
	$$\mathrm{ne}(\Gamma, \Gamma', z) \deq \{x : (\tau, \ell) \in \Gamma_{\mathbf{x}}' \mid \ell = \lev{l1} \text{ and } \Gamma(x) = (\kw{int}\langle K\rangle, \lev{model}) \text{ for some } K \}$$
\end{definition}
Secondly, $\store(S_2)$ (\autoref{def:store}) is a statement that has the same store semantics as $S_2$, but density semantics of $1$: $\sems{\store(S_2)} = \sems{S_2}$, but $\semp{\store(S_2)} = 1$. This ensures that the transformation preserves both the density semantics and the store semantics of $S$ and is needed because $\kw{gen}(z)S_2$ discards any store computed by $S_2$, thus only contributing to the weight.  

\begin{definition} \label{def:store}
	Given a statement $S$, we define the statement $\store(S)$ 
	by replacing all $\kw{factor}(E)$- and $L\sim d(E_1,\ldots,E_n)$-substatements 
	in $S$ by $\kw{skip}$ (see \autoref{ap:proofs} for the precise definition).
\end{definition}

\begin{display}{The elim-gen rule:}
	\hquad 
	\staterule{Elim Gen}
	{	
		\Gamma(z) = (\kw{int}\langle K \rangle, \lev{model}) \quad
		\Gamma_{\mathrm{ne}} = \mathrm{ne}(\Gamma, \Gamma_M, z) \quad
		S' = S_D; S_M'; S_Q
		\\
		S_M' = 
		S_1; 
		f = \phi(\Gamma_{\mathrm{ne}}) \{\kw{elim}(\kw{int}\langle K \rangle z)~S_2 \}; 
		\kw{factor}(f[\dom(\Gamma_{\mathrm{ne}})]); 
		S_3; \kw{gen}(z) S_2; \store(S_2)  \\
		\Gamma \vdash S : \lev{data} \quad
		S \shred S_{D}, S_{M}, S_{Q} \quad
		\Gamma \xrightarrow{z} \Gamma_M \quad 
		\Gamma_M \typingspecial S_M : \lev{l1} \quad 
		S_{M} \shred[\Gamma_M] S_1, S_2, S_3 \quad 
		\Gamma' \vdash S' : \lev{data}
	}
	{\Gamma, S \xrightarrow{z} \Gamma', S' }
\end{display}

We can use the \ref{Elim Gen} rule to transforms a SlicStan program, with respect to a parameter $z$, as described by \autoref{alg:elimgen}. This involves three main steps:
\begin{enumerate}
	\item Separate out $S_M$ --- the \lev{model}-level sub-part of $S$ --- using the top-level type system $\vdash$ (line 1 of \autoref{alg:elimgen}). 
	\item Separate out $S_2$ --- the part of $S_M$ that involves the discrete parameter $z$ --- using the conditional independence type system $\typingspecial$ (lines 2--8). 
	\item Perform a single VE step by marginalising out $z$ in $S_2$ and sample $z$ from the conditional probability specified by $S_2$ (lines 10--11).
\end{enumerate}

\begin{algorithm}[!]         
	\caption{Single step of applying \ref{Elim Gen}}
	\label{alg:elimgen}
	\begin{flushleft} 
		\hspace*{\algorithmicindent} \textbf{Arguments:} $(\Gamma, S), z$ \CommentVars {A program $(\Gamma, S)$; the variable $z$ to eliminate}\\
		\hspace*{\algorithmicindent} \textbf{Requires:} $\Gamma \vdash S : \lev{data}$ \CommentVars{$(\Gamma, S)$ is well-typed} \\
		\hspace*{\algorithmicindent} \textbf{Returns:} $(\Gamma', S')$ \CommentVars{The transformed program}\\ 
	\end{flushleft}    
	\vspace{6pt}
	\begin{algorithmic}[1]			
		\State Slice $(\Gamma, S)$ into $S_D, S_M, S_Q$ according to $S \shred S_D, S_M, S_Q$.
		\State Derive incomplete $\Gamma_M$ from $\Gamma$ based on $\Gamma \xrightarrow{z} \Gamma_M$. 
		\CommentVars{\lev{data} of $\Gamma$ is of level \lev{l1} in $\Gamma_M$.}
		\State 
		\CommentVars{Continuous \lev{model} var. of $\Gamma$ are \lev{l1} in $\Gamma_M$.}
		\State 
		\CommentVars{$z$ is of level \lev{l2} in $\Gamma_M$.}
		\State 
		\CommentVars{All other \lev{model} variables are given}
		\State
		\CommentVars{a type level placeholder in $\Gamma_M$.}
		\State Infer missing types of $\Gamma_M$ according to $\Gamma_M \typingspecial S_M : \lev{l1}$.
		\State Slice $(\Gamma_M, S_M)$ into $S_1, S_2, S_3$ according to $S_M \shred[\Gamma_M] S_1, S_2, S_3$.
		\\
		\State $\Gamma_{\mathrm{ne}} = \mathrm{ne}(\Gamma, \Gamma_M, z)$   
		\CommentVars{Determine the discrete neighbours of $z$.}
		\State \vspace{-16pt}  \CommentVars{Eliminate $z$ and re-generate $z$.}
		\begin{flalign*}
		\hspace*{\algorithmicindent}\hspace{2pt} S_M' = (&S_1; &\\
		&f = \phi(\Gamma_{\mathrm{ne}}) \{\kw{elim}(\kw{int}\langle K \rangle z)~S_2 \}; &\\ 
		&\kw{factor}(f[\dom(\Gamma_{\mathrm{ne}})]); &\\ 
		&S_3; &\\
		&\kw{gen}(z)\,S_2; &\\
		&\store(S_2)) &
		\end{flalign*}
		\State $S' = S_D; S_M'; S_Q$
		\State Infer an optimal $\Gamma'$ according to $\Gamma' \vdash S': \lev{data}$
		%
		%
		\\
		\Return $(\Gamma', S')$
	\end{algorithmic}
	\rule{\linewidth}{0.4pt}
\end{algorithm}

All other sub-statements of the program, $S_D, S_1, S_3$ and $S_Q$, stay the same during the transformation.
By isolating $S_2$ and transforming only this part of the program, we make sure we do not introduce more work than necessary when performing variable elimination. 

To efficiently marginalise out $z$, we want to find the \textit{Markov boundary} of $z$ given all $\lev{data}$ and continuous $\lev{model}$ parameters: the data is given, and marginalisation happens inside the continuous parameters inference loop, so we can see continuous parameters as given for the purpose of discrete parameters marginalisation. Thus we are looking for the relationship: $ z \bigCI \params_{d2} \mid \data, \params_c, \params_{d1}$,
where $\data$ is the data, $\params_{c}$ are the continuous $\lev{model}$-level parameters, $\params_{d1}$ is a subset of the discrete $\lev{model}$-level parameters that is as small as possible (the Markov blanket), and $\params_{d2}$ is the rest of the discrete $\lev{model}$-level parameters.
We can find an optimal partitioning of the discrete parameters $\params_{d1}, \params_{d2}$ that respects this relationship of interest using the type system from \autoref{sec:theory} together with type inference.

The judgement $\Gamma_M \typingspecial S_M : \lev{l1}$ induces a conditional independence relationship of the form $\mathbf{x}_{\lev{l2}} \bigCI \mathbf{x}_{\lev{l3}} \mid \mathbf{x}_{\lev{l1}}$, where $\mathbf{x} \models \Gamma_{\mathbf{x}}$ (\autoref{th:ci}).
The relation $\Gamma \xrightarrow{z} \Gamma_M$ (\autoref{def:gammas}) constrains the form of $\Gamma_M$ based on $\Gamma$. 
This is needed to make sure we are working with a relationship of the form we are interested in --- $z \bigCI \params_{d2} \mid \data, \params_c, \params_{d1}$ --- and that base types $\tau$ are the same between $\Gamma$ and $\Gamma_M$. 
In particular, $\Gamma \xrightarrow{z} \Gamma_M$ constrains $z$ to be the only $\lev{l2}$ parameter in $\Gamma_M$ and all $\lev{data}$ and continuous $\lev{model}$-level parameters of $\Gamma$ are $\lev{l1}$ in $\Gamma_M$.
Note, $\mathrm{dom}(\Gamma_M) \subseteq \mathrm{dom}(\Gamma)$ and $\Gamma_M$ only contains variables that are of level \lev{model} and below in $\Gamma$. Variables that are of level \lev{genquant} in $\Gamma$ are not in $\Gamma_M$.

\begin{definition}[$\Gamma \xrightarrow{z} \Gamma'$]\label{def:gammas} ~ \\
	For a $\vdash$ typing environment $\Gamma$ and a $\typingspecial$ typing environment $\Gamma'$, a variable $z$ and a statement $S$, we have:
	$$\Gamma \xrightarrow{z} \Gamma' = 
	\begin{cases}
	\Gamma(z) = (\tau, \lev{model}) \text{ and } \Gamma_{\mathbf{x}, \lev{l2}}' = \{z : \tau, \lev{l2}\} \text{ for some } \tau \\
	x : (\tau, \ell) \in \Gamma \text{ such that } \ell \leq \lev{model} \iff x : (\tau, \ell') \in \Gamma' \text{ for some } \ell' \in \{\lev{l1}, \lev{l2}, \lev{l3}\} \\
	x : (\tau, \lev{data}) \in \Gamma \rightarrow x : (\tau, \lev{l1}) \in \Gamma' \\
	x : (\tau, \lev{model}) \in \Gamma_{\mathbf{x}} \text{ and } \tau = \kw{real} \text{ or } \tau = \kw{real}[]...[] \rightarrow x : (\tau, \lev{l1}) \in \Gamma'
	\end{cases}$$
\end{definition}


Following convention from earlier in the paper, we use level subscripts to refer to specific subsets of $\Gamma$: in the above definition, $\Gamma'_{\mathbf{x}, \lev{l2}}$ refers to the subset of parameters $\mathbf{x}$ in $\Gamma'$, which are of level $\lev{l2}$.

\subsection{Marginalising Multiple Variables: An example} \label{sssec:elimgen_example}

To eliminate more than one discrete parameter, we apply the \ref{Elim Gen} rule repeatedly.
Here, we work through a full example, showing the different steps of this repeated \ref{Elim Gen} transformation. 

Consider an extended version of the HMM model from the beginning of this section (Program D), reformulated to include transformed parameters:

\vspace{-10pt}
\begin{multicols}{2}
\centering
\textbf{G. An extended HMM}
\begin{lstlisting}
	$S \hspace{1pt}=\hspace{2pt}$ real[2] phi ~ beta(1, 1);
			 real[2] theta ~ beta(1, 1); 
			 real theta0 = theta[0];
			 int<2> z1 ~ bernoulli(theta0);
			 real theta1 = theta[z1];
			 int<2> z2 ~ bernoulli(theta1);
			 real theta2 = theta[z2];
			 int<2> z3 ~ bernoulli(theta2);
			 real phi1 = phi[z1];
			 real phi2 = phi[z2];
			 real phi3 = phi[z3];
			 data real y1 ~ normal(phi1, 1);
			 data real y2 ~ normal(phi2, 1);
			 data real y3 ~ normal(phi3, 1);
			 real theta3 = theta[z3];
			 int genz ~ bernoulli(theta4);
\end{lstlisting}

\vspace{60pt}

\textbf{The typing environment}
\begin{align*}
	\Gamma = \{ &y_{1,2,3} : (\kw{real}, \lev{data}), \\
	& \phi : (\kw{real[2]}, \lev{model}), 	\\ 
	& \theta : (\kw{real[2]}, \lev{model}), \\
	& \theta_{0,1,2} : (\kw{real}, \lev{model}), \\ 
	& \phi_{1,2,3} : (\kw{real}, \lev{model}), \\ 
	& z_{1,2,3} : (\kw{int<2>}, \lev{model}), \\
	& \theta_3 : (\kw{real}, \lev{genquant}), \\
	& genz : (\kw{int<2>}, \lev{genquant})\}
\end{align*}
\end{multicols}
\vspace{-6pt}

The variables we are interested in transforming are $z_1, z_2$ and $z_3$: these are the $\lev{model}$-level discrete parameters of Program G. The variable $\kw{genz}$ is already at $\lev{genquant}$ level, so we can sample this with ancestral sampling (no need for automatic marginalisation).

We eliminate $z_1, z_2$ and $z_3$ one by one, in that order. The order of elimination generally has a significant impact on the complexity of the resulting program (see also \autoref{ssec:ve_relation}), but we do not focus on how to choose an ordering here. The problem of finding an optimal ordering is 
well-studied \cite{Arnborg87, kjaerulff1990triangulation, amir2010approximation} and is orthogonal to the focus of our work. 

\subsubsection{Eliminating $z_1$}

Eliminating a single variable happens in three steps, as shown in \autoref{fig:elim_z1}: standard shredding into $S_D, S_M$ and $S_Q$, conditional independence shredding of $S_M$ into $S_1, S_2$ and $S_3$, and combining everything based on \ref{Elim Gen}.

\begin{figure}[!b]
\vspace{-6pt}
\begin{multicols}{2}
\centering
\textbf{(1) Standard shredding of $S$}
\begin{lstlisting}[numbers=left,numbersep=-3pt,numberstyle=\tiny\color{darkgray}]
		$S_D = $ skip;
		$S_M = $ phi ~ beta(1, 1);
			    theta ~ beta(1, 1); 
			    theta0 = theta[0];
			    z1 ~ bernoulli(theta0);
			    theta1 = theta[z1];
			    z2 ~ bernoulli(theta1);
			    theta2 = theta[z2];
			    z3 ~ bernoulli(theta2);
			    phi1 = phi[z1];
			    phi2 = phi[z2];
			    phi3 = phi[z3];
			    y1 ~ normal(phi1, 1);
			    y2 ~ normal(phi2, 1);
			    y3 ~ normal(phi3, 1);
		$S_Q  = $ theta3 = theta[z3];
			    genz ~ bernoulli(theta3);
\end{lstlisting}

\vspace{1pt}

\centering
\textbf{(2) CI shredding of $S_M$}\vspace{-1pt}
\begin{lstlisting}[numbers=left,numbersep=-3pt,numberstyle=\tiny\color{darkgray}]
		$S_1 \hspace{1pt} = \hspace{2pt}$ phi ~ beta(1, 1);
			    theta ~ beta(1, 1);
			    theta0 = theta[0];  
		$S_2 \hspace{1pt} = \hspace{2pt}$ z1 ~ bernoulli(theta0);
			    theta1 = theta[z1];
			    z2 ~ bernoulli(theta1);
			    phi1 = phi[z1];
			    y1 ~ normal(phi1, 1);
		$S_3 \hspace{1pt} = \hspace{2pt}$ theta2 = theta[z2];
			    z3 ~ bernoulli(theta2);
			    phi2 = phi[z2];
			    phi3 = phi[z3];
			    y2 ~ normal(phi2, 1);
			    y3 ~ normal(phi3, 1);
\end{lstlisting}

\centering
\textbf{(3) Applying \ref{Elim Gen}: Program G-1}
\begin{lstlisting}[numbers=left,numbersep=-3pt,numberstyle=\tiny\color{darkgray}]
		phi ~ beta(1, 1);
		theta ~ beta(1, 1);
		theta0 = theta[0];
		
		f1 = $\phi$([int<2> z2]){
			elim(int<2> z1){
				z1 ~ bernoulli(theta0);
				theta1 = theta[z1];
				z2 ~ bernoulli(theta1);
				phi1 = phi[z1];
				y1 ~ normal(phi1, 1);
			}}
		factor(f1[z2]);
		
		theta2 = theta[z2];
		z3 ~ bernoulli(theta2);
		phi2 = phi[z2];
		phi3 = phi[z3];
		y2 ~ normal(phi2, 1);
		y3 ~ normal(phi3, 1);
		
		gen(int z1){
			z1 ~ bernoulli(theta0);
			theta1 = theta[z1]
			z2 ~ bernoulli(theta1);
			phi1 = phi[z1];
			y1 ~ normal(phi1, 1);
		}
		theta1 = theta[z1];
		phi1 = phi[z1];
		
		theta3 = theta[z3];
		genz ~ bernoulli(theta3);
\end{lstlisting}
\end{multicols}
\vspace{-8pt}
\caption{Step-by-step elimination of $z_1$ in Program G.}
\label{fig:elim_z1}
\end{figure}

\begin{enumerate}
	\item \textit{Standard shredding:} $S \shred S_{D}, S_{M}, S_{Q}$.
	Firstly, we separate out the parts of the program that depend on discrete parameters \textit{generatively}. That is any part of the program that would be in generated quantities with respect to the original program. In our case, this includes the last two lines in $S$. This would also include the \kw{gen} parts of the transform program, that draw discrete parameters as generated quantities. 
	Thus, $S \shred S_{D}, S_{M}, S_{Q}$, where $S_{D}$ is empty,
	$S_{Q} = \left(\kw{theta3} = \kw{theta[z3];} \kw{genz} \sim \kw{bernoulli}(\kw{theta3})\right)$, and $S_M$ is the rest of the program (see \autoref{fig:elim_z1}, (1)).
	
	\item \textit{Conditional independence shredding:} $S_M \shred[\Gamma_M] S_1, S_2, S_3$.
	In the next step, we want to establish a conditional independence relationship $z_1 \mid A \bigCI \mathbf{y}, \phi_0, \theta_0, B$, where $z_1$ is some discrete parameter and  $A, B$ is a partitioning of the rest of the discrete parameters in the model: $\{z_2, z_3\}$. 
	We derive a new, $\typingspecial$ typing environment $\Gamma_M$, using $\Gamma \xrightarrow{z} \Gamma_M$:
	\begin{align*} 
		\Gamma_M = \{ & y_{1,2,3} : (\kw{real}, \lev{l1}), \phi : (\kw{real[2]}, \lev{l1}), \theta: (\kw{real[2]}, \lev{l1}), \\
		& z_{1} : (\kw{int<2>}, \lev{l2}), z_2 : (\kw{int<2>}, \lev{?}), z_3 : (\kw{int<2>}, \lev{?}) \\
		& \theta_{0} : (\kw{real}, \lev{l1}), \theta_{1,2} : (\kw{real}, \lev{?}), \phi_{1,2,3} : (\kw{real}, \lev{?})\}
	\end{align*}
	
	Here, we use the notation $\lev{?}$ for a type placeholder, which will be inferred using type inference.  
	
	The optimal $\Gamma_M$ under the type inference soft constraint $\lev{l3} \prec \lev{l1} \prec \lev{l2}$ such that $\Gamma_M \typingspecial S_M : \lev{l1}$ is such that the levels of $\theta_{1}$ and $\phi_{1}$ are \lev{l2}, $z_2$ is \lev{l1} and $\theta_2, \phi_2$ and $\phi_3$ are \lev{l3}.
	Shredding then gives us
	$S_M \shred[\Gamma_M] S_1, S_2, S_3 $, as in \autoref{fig:elim_z1}, (2).
	
	\item \textit{Combining based on \ref{Elim Gen}}.
	Having rearranged the program into suitable sub-statements, we use \ref{Elim Gen} to get Program G-1 (\autoref{fig:elim_z1}, (3)) and:
	\begin{align*}
		\Gamma' = \{ &y_{1,2,3} : (\kw{real}, \lev{data}), \phi : (\kw{real}, \lev{model}), \\
		& \theta_1 : (\kw{real}, \lev{genquant}), \theta_{0,2} : (\kw{real}, \lev{model}), \\
		& \phi_1 : (\kw{real}, \lev{genquant}), \phi_{2,3} : (\kw{real}, \lev{model}), \\
		& z_{1} : (\kw{int}, \lev{genquant}), z_{2,3} : (\kw{int<2>}, \lev{model}), \\
		& \theta_3 : (\kw{real}, \lev{genquant}), genz : (\kw{int<2>}, \lev{genquant})\}
	\end{align*}
	
\end{enumerate}

\paragraph{Eliminating $z_2$} 

We apply the same procedure to eliminate the next variable, $z_2$, from the updated Program G-1. 
The variable $z_1$ is no longer a \lev{model}-level parameter, thus the only neighbouring parameter of $z_2$ is $z_3$. Note also that the computation of the factor $f_1$ does not include any free discrete parameters (both $z_1$ and $z_2$ are local to the computation due to $\kw{elim}$ and $\phi$). Thus, we do not need to include the computation of this factor anywhere else in the program (it does not get nested into other computations).
We obtain a new program, Program G-2:
\begin{multicols}{2}
\centering
\textbf{Program G-2}
\begin{lstlisting}[numbers=left,numbersep=-3pt,numberstyle=\tiny\color{darkgray}]
		phi ~ beta(1, 1);
		theta ~ beta(1, 1);
		theta0 = theta[0];
		
		f1 = $\phi$([int<2> z2]){ elim(int<2> z1){
				z1 ~ bernoulli(theta0);
				theta1 = theta[z1];
				z2 ~ bernoulli(theta1);
				phi1 = phi[z1];
				y1 ~ normal(phi1, 1);
			}}			
		f2 = $\phi$([int<2> z3]){ elim(int<2> z2){
				factor(f1[z2]);
				theta2 = theta[z2];
				z3 ~ bernoulli(theta2);
				phi2 = phi[z2];
				y2 ~ normal(phi2, 1);
			}}
			
		factor(f2[z3]);
		phi3 = phi[z3];
		y3 ~ normal(phi3, 1);
	
		gen(int z2){
			factor(f1[z2]);
			theta2 = theta[z2];
			z3 ~ bernoulli(theta2);
			phi2 = phi[z2];
			y2 ~ normal(phi2, 1);
		}
		theta2 = theta[z2];
		phi2 = phi[z2];
				
		gen(int z1){
			z1 ~ bernoulli(theta0);
			theta1 = theta[z1];
			z2 ~ bernoulli(theta1);
			phi1 = phi[z1];
			y1 ~ normal(phi1, 1);
		}
		theta1 = theta[z1];
		phi1 = phi[z1];
		
		theta3 = theta0[z3];
		genz ~ bernoulli(theta3);
\end{lstlisting}
\vspace{-18pt}
\textcolor{white}{.}
\end{multicols}

\paragraph{Eliminating $z_3$}

Finally, we eliminate $z_3$, which is the only discrete \lev{model}-level parameter left in the program. Thus, $z_3$ has no neighbours and $f_3$ is of arity 0: it is a real number instead of a vector.  

The final program generated by our implementation is Program G-3:

\vspace{-6pt}
\begin{multicols}{2}
\centering
\textbf{Program G-3}
\begin{lstlisting}[numbers=left,numbersep=-3pt,numberstyle=\tiny\color{darkgray}]
		phi0 ~ beta(1, 1);
		theta0 ~ beta(1, 1);
		
		f1 = $\phi$(int<2> z2){	elim(int<2> z1){
				z1 ~ bernoulli(theta0);
				theta1 = theta[z1];
				z2 ~ bernoulli(theta1);
				phi1 = phi[z1];
				y1 ~ normal(phi1, 1);
			}}
			
		f2 = $\phi$(int<2> z3){elim(int<2> z2){
				factor(f1[z2]);
				theta2 = theta[z2];
				z3 ~ bernoulli(theta2);
				phi2 = phi[z2];
				y2 ~ normal(phi2, 1);
			}}
			
		f3 = $\phi$(){elim(int<2> z3){
				factor(f2[z3]);
				phi3 = phi[z3];
				y3 ~ normal(phi3, 1);
			}}
		
		factor(f3);	
		gen(int z3){
			factor(f2[z3]);
			phi3 = phi[z3];;
			y3 ~ normal(phi3, 1);
		}
		phi3 = phi[z3];
		gen(int z2){
			factor(f1[z2]);
			theta2 = theta[z2];;
			z3 ~ bernoulli(theta2);
			phi2 = phi[z2];;
			y2 ~ normal(phi2, 1);
		}
		theta2 = theta[z2];
		phi2 = phi[z2];
		gen(int z1){
			z1 ~ bernoulli(theta0);
			theta1 = theta[z1];
			z2 ~ bernoulli(theta1);
			phi1 = phi[z1];
			y1 ~ normal(phi1, 1);
		}
		theta1 = theta[z1];
		phi1 = phi[z1];
		
		gen3 = theta[z3];
		genz ~ bernoulli(theta3);
\end{lstlisting}
\end{multicols}\vspace{-8pt}

\subsection{Relating to Variable Elimination and Complexity Analysis} \label{ssec:ve_relation}

Assume $\data, \params_d$, and $\params_c$ are the data, discrete model-level parameters, and continuous model-level parameters, respectively.
As $S_2$ is a single-level statement of level $\lev{l2}$, the density semantics of $S_2$ is of the form $\psi(\mathbf{x}_{\lev{l1}}, \mathbf{x}_{\lev{l2}}) = \psi(\data, \params_c, \params_{d, \lev{l1}}, z)$ (\autoref{lem:single-lev-prop2}). 

As $\kw{elim}(\kw{int}\langle K \rangle z)~\_$ binds the variable $z$
and  $\phi(\Gamma_{\mathrm{ne}}) \{\_\}$ 
binds the variables in  $\dom(\Gamma_{\mathrm{ne}})$,
the expression $\phi(\Gamma_{\mathrm{ne}}) \{\kw{elim}(\kw{int}\langle K \rangle z)~S_2$ depends only on continuous parameters and data, and it contains no free mentions of any discrete variables. 
This means that the expression will be of level $\lev{l1}$ and shredded into $S_1$ during the marginalisation of any subsequent discrete variable $z'$. The substatement $S_2$ will always be some sub-statement of the original program (prior to any transformations), up to potentially several constant factors of the form $\kw{factor}(f[\dom(\Gamma_{\mathrm{ne}})])$.

This observation makes it easy to reason about how repeated application of the \ref{Elim Gen} transform changes the complexity of the program. If the complexity of a SlicStan program with $N$ discrete parameters of support $1, \dots, K$, is $\mathcal{O}(S)$, then the complexity of a program where we naively marginalised out the discrete variables (Program E) will be $\mathcal{O}(S \times K^N)$.
In contrast, transforming with \ref{Elim Gen} gives us a program of complexity at most $\mathcal{O}(N \times S \times K^{M+1})$
where $M$ is the largest number of direct neighbours in the factor graph induced by the program. Further, the complexity could be smaller depending on the elimination ordering of choice. 
This result is not surprising, as we conjecture that repeated application of \ref{Elim Gen} is equivalent to variable elimination (though we do not formally prove this equivalence), which is of the same complexity. 

It is clear from this complexity observation that VE is not always efficient. When the dependency graph is dense, $M$ will be close to $N$, thus inference will be infeasible for large $N$. Fortunately, in many practical cases (such as those discussed in \autoref{sec:examples}), this graph is sparse ($M \ll N$) and our approach is suitable and efficient.
We note that this is a general limitation of exact inference of discrete parameters, and it is not a limitation of our approach in particular. 

\subsection{Semantic Preservation of the Discrete Variable Transformation}

The result we are interested in is the semantic preservation of the transformation rule $\xrightarrow{z}$.

\begin{theorem}[Semantic preservation of $\xrightarrow{z}$] \label{th:sempreservation}$ $ \\
	For SlicStan programs $\Gamma, S$ and $\Gamma', S'$, and a discrete parameter $z$: 
	$\Gamma, S \xrightarrow{z} \Gamma', S'$ implies $\sem{S} = \sem{S'}$.	
\end{theorem}
\begin{proof}
Note that shredding preserves semantics with respect to both $\vdash$ and $\typingspecial$ (\autoref{lem:shred} and \ref{lem:shred2}), 
examine the meaning of derived forms (\autoref{lem:s2} and~\ref{lem:sgen}),
note properties of single-level statements (\autoref{lem:single-lev-prop2}),
and apply the results on factorisation of shredding (\autoref{th:shred_gen}) and conditional independence (\autoref{th:ci}).
We present the full proof in \autoref{ap:proofs}.
\end{proof}

In addition, we also show that it is always possible to find a program derivable with \ref{Elim Gen}, such that a \lev{model}-level variable $z$ is transformed to a \lev{genquant}-level variable.  

\begin{lemma}[Existence of \lev{model} to \lev{genquant} transformation] \label{lem:exists}
	For any SlicStan program $\Gamma, S$ such that $\Gamma \vdash S : \lev{l1}$, 
	and a variable $z \in \dom(\Gamma)$ such that $\Gamma(z) = (\kw{int}\langle K \rangle, \lev{model})$,
	there exists a SlicStan program $\Gamma', S'$, such that:
	$$\Gamma, S \xrightarrow{z} \Gamma', S' 
	\quad \text{and} \quad 
	\Gamma'(z) = (\kw{int}\langle K \rangle, \lev{genquant})$$
\end{lemma}
\begin{proof}
By inspecting the level types of variables in each part of a program derivable using \ref{Elim Gen}. We include the full proof in \autoref{ap:proofs}.
\end{proof}

The practical usefulness of~\autoref{th:sempreservation} stems from the fact that it allows us to separate inference for discrete and continuous parameters. After applying \ref{Elim Gen} to each discrete $\lev{model}$-level parameter, we are left with a program that only has $\lev{genquant}$-level discrete parameters (\autoref{lem:exists}). We can then slice the program into $S_{\textsc{hmc}}$ and $S_{\textsc{gen}}$ and infer continuous parameters by using HMC (or other algorithms) on $S_{\textsc{hmc}}$ and, next, 
draw the discrete parameters using ancestral sampling by running forward $S_{\textsc{gen}}$.
\autoref{th:sempreservation} tells us that this is a correct inference strategy.

When used in the context of a model with only discrete parameters, our approach corresponds to exact inference through VE. In the presence of discrete and continuous parameters, our transformation gives an analytical sub-solution for the discrete parameters in the model.

A limitation of our method is that, due to its density-based nature, it can only be applied to models of fixed size. It cannot, in its current form, support models where the number of random variables changes during inference, such as Dirichlet Processes. However, this is a typical constraint adopted in Bayesian inference for efficiency. 
Another limitation is that discrete variables need to have finite (and fixed) support. For example, the method cannot be applied to transform a Poisson-distributed variable.
In some but not all applications, truncating unbounded discrete parameters at a realistic upper bound would suffice to make the method applicable.

An advantage of our method is that it can be combined with any inference algorithm that requires a function proportional to the joint density of variables. This includes gradient-based algorithms, such as HMC and variational inference, but it can also be used with methods that allow for (e.g. unbounded) discrete variables as an analytical sub-solution that can optimise inference. For example, consider a Poisson variable $n \sim \mathrm{Poisson}(\lambda)$ and a Binomial variable $k \sim \mathrm{Binomial}(n, p)$. While $n$ is of infinite support, and we cannot directly sum over all of its possible values, analytically marginalising out $n$ gives us $k \sim \mathrm{Poisson}(\lambda p)$. Future work can utilise such analytical results in place of explicit summation where possible.


\subsection{Scope and limitations of \ref{Elim Gen}}

\begin{figure}[!h]
\begin{multicols}{3}
\centering
\textbf{Program H}
\begin{lstlisting}
data int K;
real[K][K] phi;
real[K] mu;
int<K> z1 ~ 
			categorical(phi[0]);
int<K> z2;
int<K> z3;

if (z1 > K/2) {
	z2 ~ categorical(phi[z1]);
	z3 ~ categorical(phi[z2]);
}
else {
	z2 ~ categorical(phi[z1]);
	z3 ~ categorical(phi[z1]);    
}

data real y1 ~ 
			normal(mu[z1],1);
data real y2 ~ 
			normal(mu[z2],1);
data real y3 ~ 
			normal(mu[z3],1);    
\end{lstlisting}
\vspace{40pt}

\centering
\textbf{Program H-A: Optimal transformation}

$\dots$
\begin{lstlisting}
factor(elim(int<2> z1) {
 z1 ~ categorical(phi[0]);
 y1 ~ normal(mu[z1], 1));
 
 if (z1 > K/2) {
  elim(int<2> z2) {
   elim(int<2> z3) {
    z2 ~ categorical(phi[z1]);
    z3 ~ categorical(phi[z2]);
    y2 ~ normal(mu[z2],1);
    y3 ~ normal(mu[z3],1);
 }}} 
 
 else {
  elim(int<2> z2) {
   z2 ~ categorical(phi[z1]));
   y2 ~ normal(mu[z2],1);
  }
  elim(int<2> z3) {
   z3 ~ categorical(phi[z1]));
   y3 ~ normal(mu[z3],1);
}}});
\end{lstlisting}
\vspace{30pt}

\centering
\textbf{Program H-B: Our transformation}

$\dots$
\begin{lstlisting}
f1 = $\phi$(int<2> z2, int<2> z3){
	elim(int<2> z1){
	 z1 ~ categorical(phi[0]);
	 if(z1 > K/2){
	   z2 ~ categorical(phi[z1]);
	   z3 ~ categorical(phi[z2]);
	 }
	 else{
	   z2 ~ categorical(phi[z1]);
	   z3 ~ categorical(phi[z1]);
	 }
	 y1 ~ normal(mu[z1], 1);
	 y2 ~ normal(mu[z2], 1);
	 y3 ~ normal(mu[z3], 1);
}}
f2 = $\phi$(int<2> z3){
		elim(int<2>z2) 
			factor(f1[z2,z3]);}
f3 = $\phi$(){
		elim(int<2> z3) 
			factor(f2[z3]);}
factor(f3);
\end{lstlisting}
\end{multicols}
\caption{A program with different conditional dependencies depending on control flow.}
\end{figure}

Previously, we discussed the scope of the conditional independence result of the paper (\autoref{ssec:scope}). Similarly, here we demonstrate with an example, a situation where our approach of eliminating variables one-by-one using \ref{Elim Gen} is not optimal.

Consider the simple control-flow Program H below. In this example $z_2$ and $z_3$ are \textit{not} conditionally independent given $z_1 = 1$, but they are conditionally independent given $z_1 > K/2$. 
This independence is also referred to as context-specific independence \cite{Gates, CSI}.
We can use different elimination strategy depending on which \lstinline{if}-branch of the program we find ourselves. Program H-A demonstrates this: its complexity is $\mathcal{O}(\frac{K}{2} \times K^2 + \frac{K}{2} \times 2 \times K) = \mathcal{O}(\frac{1}{2}K^3 + K^2)$.

The typing relation $\typingspecial$ can only detect overall (in)dependencies, where sets of variables are conditionally independent given some $X$, regardless of what value $X$ takes. 
Thus, our static analysis is not 
able to detect that $z_2 \bigCI z_3 \mid z_1 = 0$. This results in Program H-B, which has complexity $\mathcal{O}(K^3 + K^2 + K)$: the same complexity as the optimal Program H-A, but with a bigger constant. 

Even if we extend our approach to detect that $z_2$ and $z_3$ are independent in one branch, it is unclear how to incorporate this new information. Our strategy is based on computing intermediate factors that allow re-using already computed information: eliminating $z_1$ requires computing a new factor $f_1$ that no longer depends on $z_1$. We represented $f_1$ with a multidimensional array indexed by $z_2$ and $z_3$, and we need to define each element of that array, thus we cannot decouple them for particular values of $z_1$. 

Runtime systems that compute intermediate factors in a similar way, such as Pyro \cite{Pyro}, face that same limitation. Birch \cite{Birch}, on the other hand, will be able to detect the conditional independence in the case $z_1 > K/2$, but it will not marginalise $z_1$, as it cannot (analytically) marginalise over branches. Instead, it uses Sequential Monte Carlo (SMC) to repeatedly sample $z_1$ and proceed according to its value.

\section{Implementation and empirical evaluation} \label{sec:examples}

The transformation we introduce can be useful for variety of models, and it can be adapted to PPLs to increase efficiency of inference and usability. 
Most notably, it can be used to extend Stan to allow for direct treatment of discrete variables, where previously that was not possible. 

In this section, we present a brief overview of such a discrete parameter extension for SlicStan 
(\autoref{ssec:impl}). 
To evaluate the practicality of \ref{Elim Gen}, we build a partial NumPyro \cite{NumPyro} backend for SlicStan, and compare our static approach to variable elimination for discrete parameters to the dynamic approach of NumPyro (\autoref{ssec:eval}). 
We find that our static transformation strategy speeds up inference compared to the dynamic approach, but that for models with a large number of discrete parameters performance gains could be diminished by the exponentially growing compilation time (\autoref{ssec:results}).

In addition to demonstrating the practicality of our contribution through empirical evaluation, we also discuss the usefulness of our contribution through examples, in \autoref{ap:examples}.

\subsection{Implementation} \label{ssec:impl}

We update the original SlicStan\footnote{Available at \url{https://github.com/mgorinova/SlicStan}.} according to the modification described in \autoref{sec:background}, and extend it to support automatic variable elimination through the scheme outlined in \autoref{sec:application}. As with the first version of SlicStan, the transformation produces a new SlicStan program that is then translated to Stan.

The variable elimination transformation procedure works by applying \ref{Elim Gen} iteratively, once for each discrete variable, as we show in \autoref{sssec:elimgen_example}.
The level types \lev{l1}, \lev{l2} and \lev{l3} are not exposed to the user, and are inferred automatically. Using bidirectional type-checking, we are able to synthesise a set of hard constraints that the levels must satisfy. These hard constraints will typically be satisfied by more than one assignment of variables to levels. We search for the optimal types with respect to the soft constraints $\lev{l3} \prec \lev{l1} \prec \lev{l2}$, using the theorem prover Z3 \cite{Z3}.

\subsection{Empirical evaluation} \label{ssec:eval}


To evaluate the practicality of our approach, we compare to the prior work most closely related to ours: that of \citet{PyroDiscrete}, who implement efficient variable-elimination for plated factor graphs in Pyro \cite{Pyro}. Their approach uses effect-handlers and dynamically marginalises discrete variables, so that gradient-based inference schemes can be used for the continuous parameters. This VE strategy has also been implemented in NumPyro \cite{NumPyro}.

As both ours and Pyro's strategies correspond to VE, we do not expect to see differences in complexity of the resulting programs.
However, as in our case the VE algorithm is determined and set up at compile time, while in the case of Pyro/NumPyro, this is done at run time.
The main question we aim to address is whether setting up the variable elimination logistics at compile time results in a practical runtime speed-up.

To allow for this comparison, we built a partial NumPyro backend for SlicStan. For each model we choose, we compare the runtime performance of three NumPyro programs:
\begin{enumerate}
	\item The NumPyro program obtained by translating a SlicStan program with discrete parameters to NumPyro directly (labelled `NumPyro'). This is the baseline: we leave the discrete parameter elimination to NumPyro. 
	\item The NumPyro program obtained by translating a transformed SlicStan program, where all discrete parameters have been eliminated according to \ref{Elim Gen} (labelled `SlicStan'). The variable elimination set-up is done at compile time; NumPyro does not do any marginalisation.
	\item A hand-optimised NumPyro program, which uses the \lstinline{plate} and \lstinline{markov} program constructs to specify some of the conditional independencies in the program (labelled `NumPyro-Opt'). 
\end{enumerate}

In each case, we measure the time (in seconds) for sampling a single chain consisting of $2 500$ warm-up samples, and $10 000$ samples using NUTS \cite{NUTS}.

In addition, we report three compilation times: 
\begin{enumerate}
	\item The compilation time of the NumPyro program obtained by translating a SlicStan program with discrete parameters to NumPyro directly (labelled `NumPyro').
	\item The compilation time of the NumPyro program obtained by translating a transformed SlicStan program, where all discrete parameters have been eliminated (labelled `SlicStan').
	\item The time taken for the original SlicStan program to be transformed using \ref{Elim Gen} and translated to NumPyro code (labelled `SlicStan-to-NumPyro').
\end{enumerate}
  
We consider different numbers of discrete parameters for each model, up to $25$ discrete parameters. 
We do not consider more than $25$ parameters due to constraints of the NumPyro baseline, which we discuss in more detail in \autoref{ssec:results}.
We run experiments on two classes of model often seen in practice: hidden Markov models (\autoref{sssec:hmms_eval}) and mixture models (\autoref{sssec:mixture_eval}).
To ensure a fair comparison, the same elimination ordering was used across experiments. 
Experiments were run on a dual-core 2.30GHz Intel Xeon CPU and a Tesla T4 GPU (when applicable).
All SlicStan models used in the experiments are available at the SlicStan repo.

\subsubsection{Hidden Markov models} \label{sssec:hmms_eval}
We showed several examples of simple HMMs throughout the paper (Program A, Program D, Program G) and worked through a complete example of VE in an HMM (\ref{sssec:elimgen_example}).
We evaluate our approach on both the simple first-order HMM seen previously, and on two additional ones: second-order HMM and factorial HMM. 

\paragraph{First-order HMM}

The first-order HMM is a simple chain of $N$ discrete variables, each taking a value from 1 to $K$ according to a categorical distribution. The event probabilities for the distribution of $z_n$ are given by $\boldsymbol{\theta}_{z_{n-1}}$, where $\boldsymbol{\theta}$ is some given $K \times K$ matrix. Each data point $\mathbf{y}$ is modelled as coming from a Gaussian distribution with mean $\mu_{z_n}$ and standard deviation $1$, where $\boldsymbol{\mu}$ is a $K-$dimensional continuous parameter of the model.
%
\begin{gather*}
	\mu_k \sim \mathcal{N}(0, 1) \quad\text{for } k \in 1, \dots, K \\
	z_1 \sim \mathrm{categorical}(\boldsymbol{\theta}_{1}) \\
	z_n \sim \mathrm{categorical}(\boldsymbol{\theta}_{z_{n-1}}) \quad\text{for } n \in 2, \dots, N \\
	y_n \sim \mathcal{N}(\mu_{z_{n}}, 1) \quad\text{for } n \in 1, \dots, N \\
\end{gather*}

\vspace{-6pt}
We measure the compilation time and the time taken to sample $1$ chain with each of the 3 NumPyro programs corresponding to this model. We use $K = 3$ and different values for $N$, ranging from $N=3$ to $N=25$. 
\autoref{fig:hmm-1} shows a summary of the results. We see that both on CPU and GPU, the program transformed using SlicStan outperforms the automatically generated NumPyro and also the manually optimised NumPyro-Opt. 
Each of the three programs has compilation time exponentially increasing with the number of variables, however SlicStan's compilation time increases the fastest. We discuss this drawback in more detail in \autoref{ssec:results}, highlighting the importance of an extended loop-level analysis being considered in future work. 

\begin{figure}[!p]
	\begin{subfigure}{0.32\textwidth}
	\includegraphics[width=\textwidth]{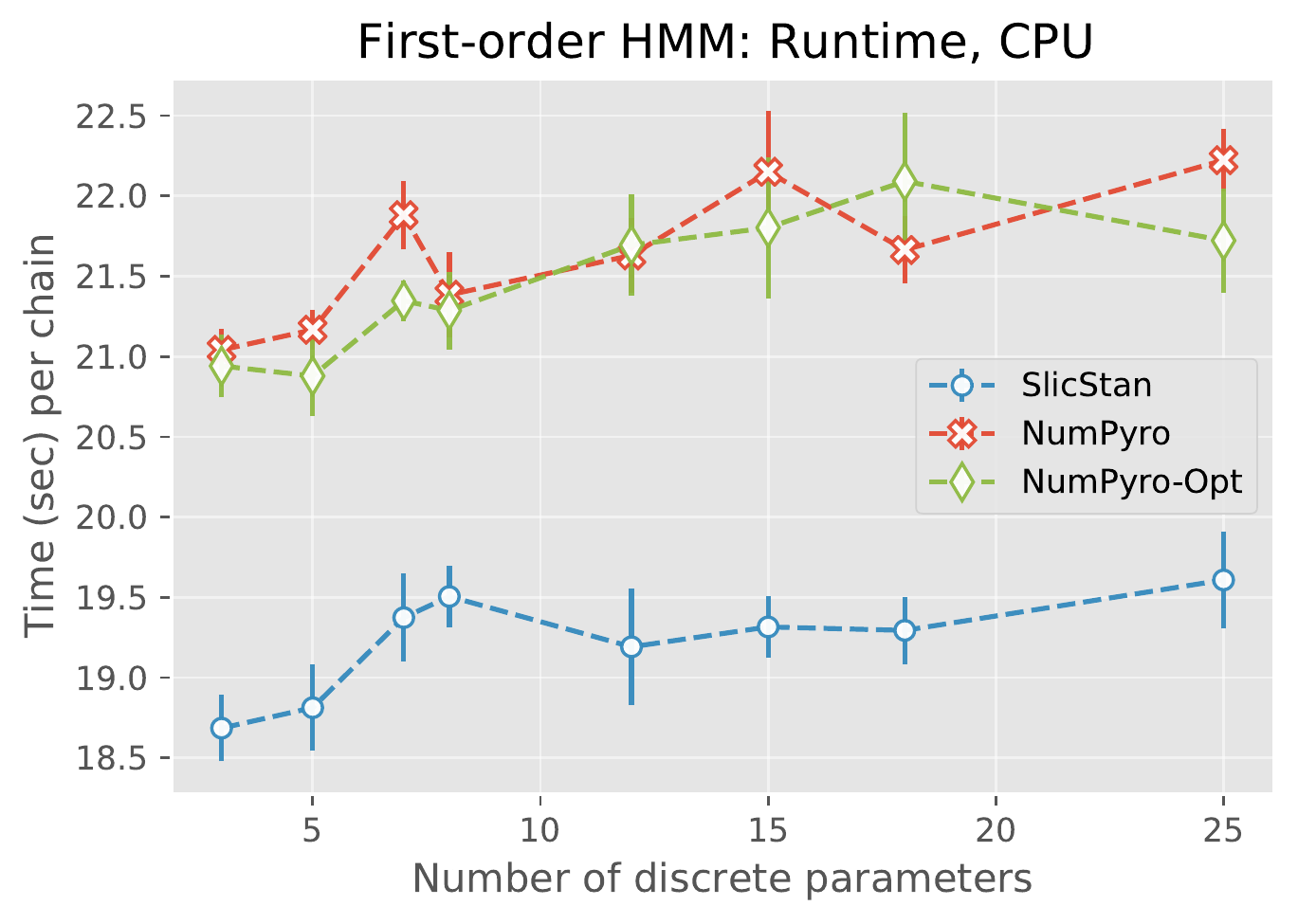}
	\end{subfigure}
	\begin{subfigure}{0.32\textwidth}
		\includegraphics[width=\textwidth]{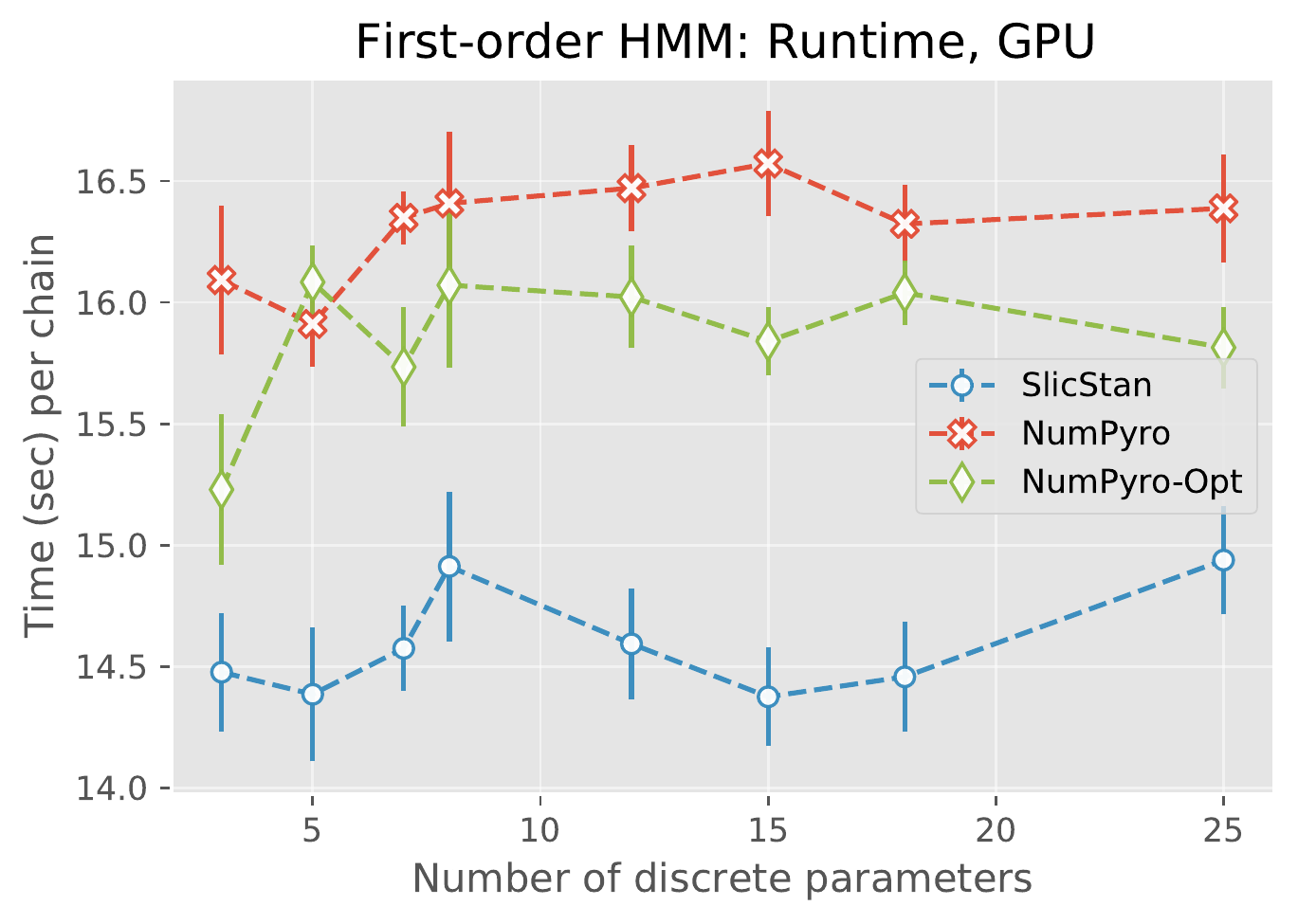}
	\end{subfigure}
	\begin{subfigure}{0.32\textwidth}
	\includegraphics[width=\textwidth]{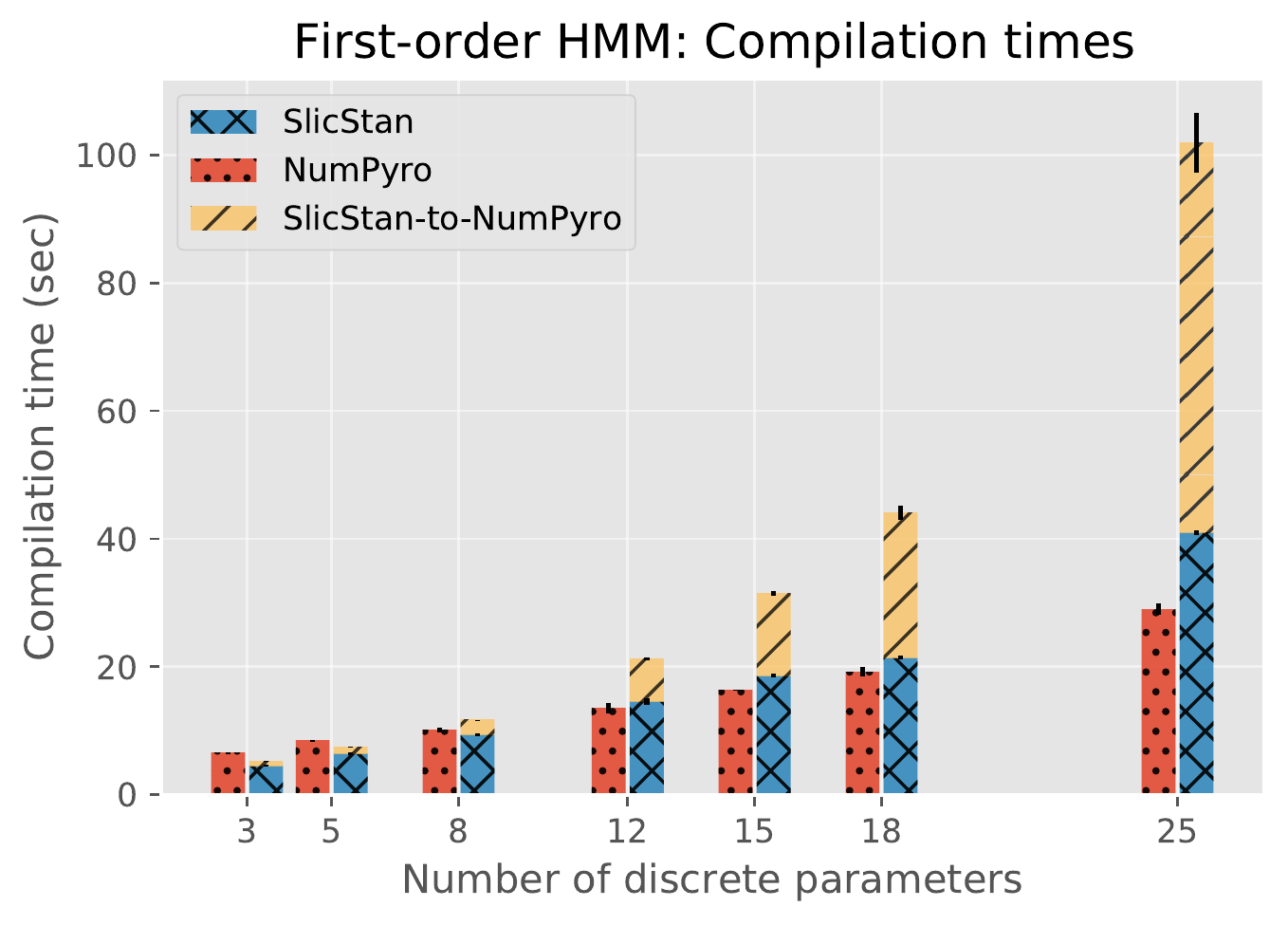}
	\end{subfigure}
	\vspace{-8pt}
	\caption{HMM results}
	\label{fig:hmm-1}
\end{figure}

\paragraph{Second-order HMM}

The second-order HMM is very similar to the first-order HMM, but the discrete variables depend on the previous $2$ variables, in this case taking the maximum of the two.
%
\begin{gather*}
	\mu_k \sim \mathcal{N}(0, 1) \quad\text{for } k \in 1, \dots, K \\
	z_1 \sim \mathrm{categorical}(\theta_{1}), \quad
	z_2 \sim \mathrm{categorical}(\theta_{z_1}) \\
	z_n \sim \mathrm{categorical}(\theta_{\mathrm{max}(z_{n-2}, z_{n-1})}) \quad\text{for } n \in 3, \dots, N \\
	y_n \sim \mathcal{N}(\mu_{z_{n}}, 1) \quad\text{for } n \in 1, \dots, N \\
\end{gather*}

\vspace{-8pt}
Similarly to before, we run the experiment for $K = 3$ and different values for $N$, ranging from $N=3$ to $N=25$. 
We show the results in \autoref{fig:hmm-2}, which once again shows SlicStan outperforming NumPyro and NumPyro-Opt in terms of runtime, but having slower compilation time for a larger number of discrete parameters.

\begin{figure}[!p]
	\begin{subfigure}{0.32\textwidth}
		\includegraphics[width=\textwidth]{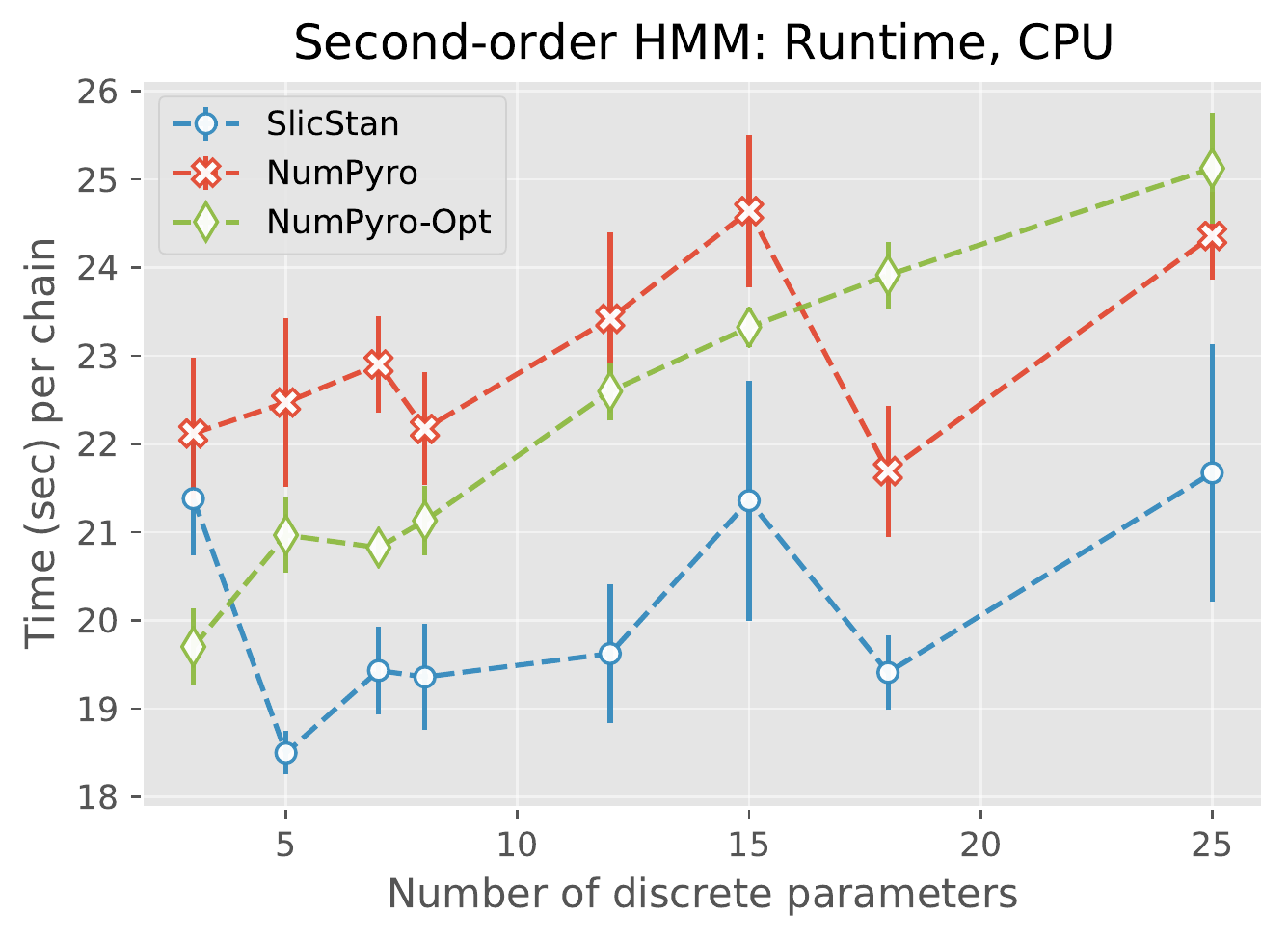}
	\end{subfigure}
	\begin{subfigure}{0.32\textwidth}
		\includegraphics[width=\textwidth]{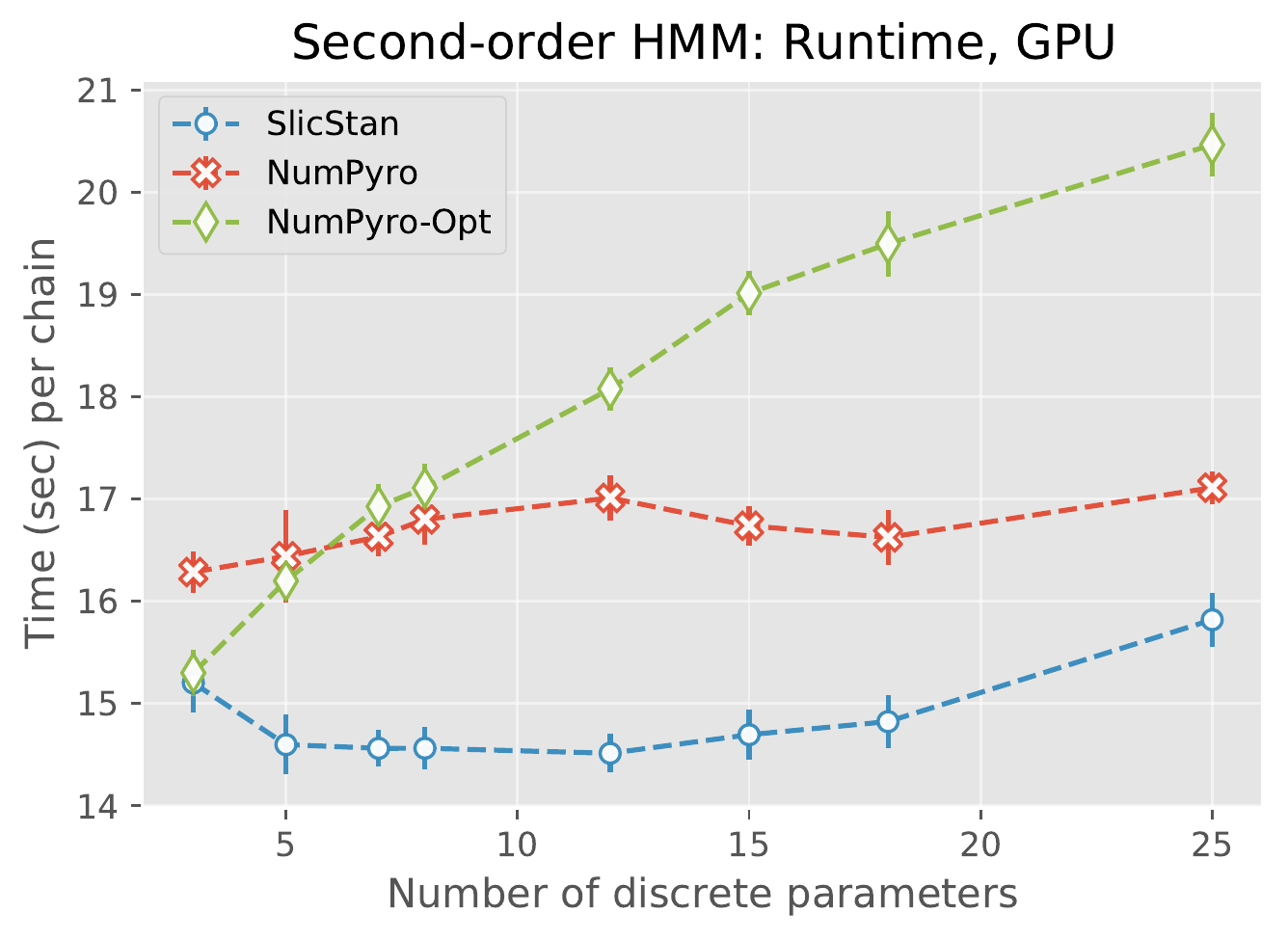}
	\end{subfigure}
	\begin{subfigure}{0.32\textwidth}
		\includegraphics[width=\textwidth]{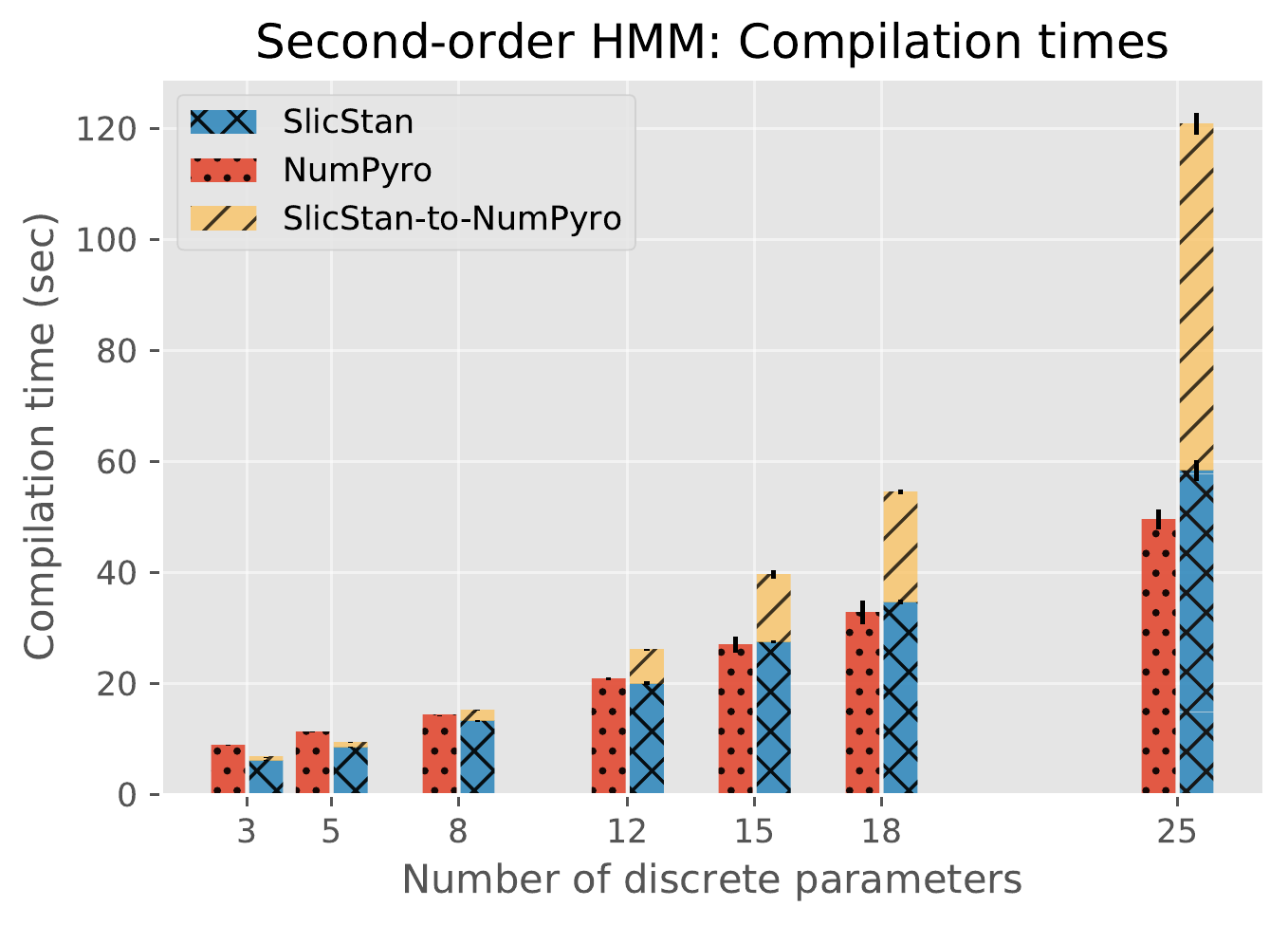}
	\end{subfigure}
	\vspace{-8pt}
	\caption{Second-order HMM results}
	\label{fig:hmm-2}
\end{figure}

\begin{figure}[!p]
	\begin{subfigure}{0.32\textwidth}
		\includegraphics[width=\textwidth]{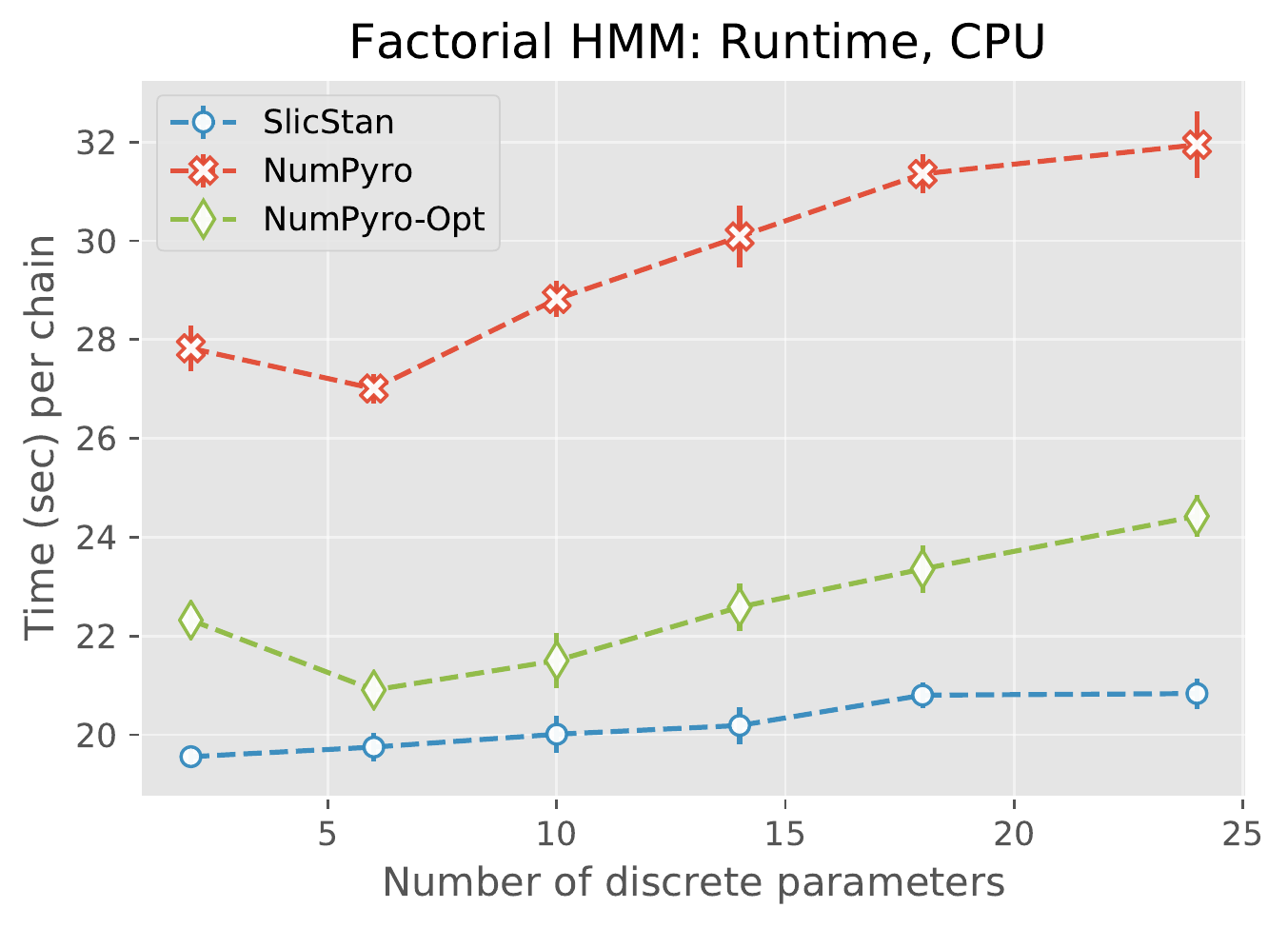}
	\end{subfigure}
	\begin{subfigure}{0.32\textwidth}
		\includegraphics[width=\textwidth]{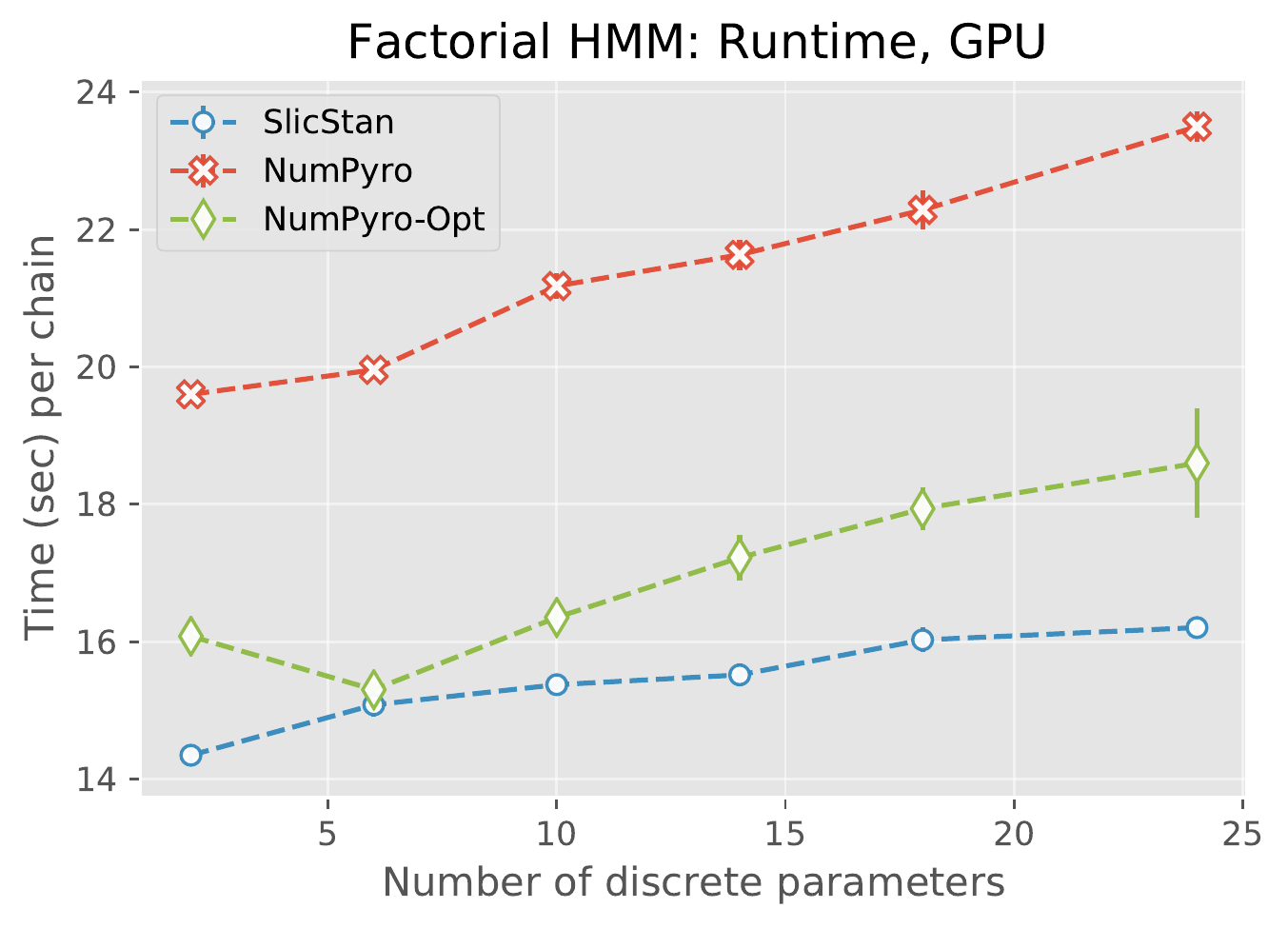}
	\end{subfigure}
	\begin{subfigure}{0.32\textwidth}
		\includegraphics[width=\textwidth]{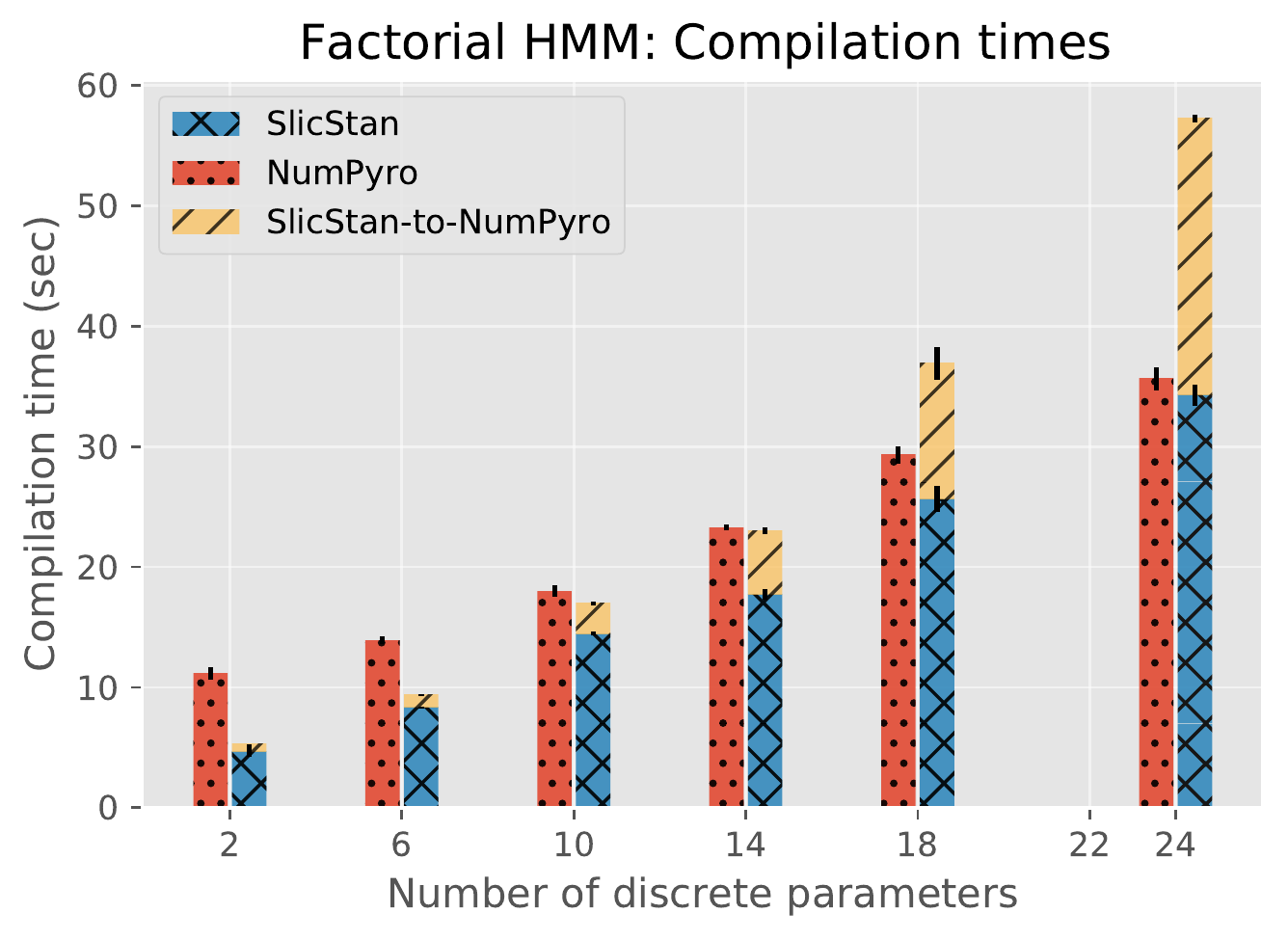}
	\end{subfigure}
	\vspace{-8pt}
	\caption{Factorial HMM results}
	\label{fig:hmm-factorial}
\end{figure}

\begin{figure}[!p]
	\begin{subfigure}{0.32\textwidth}
		\includegraphics[width=\textwidth]{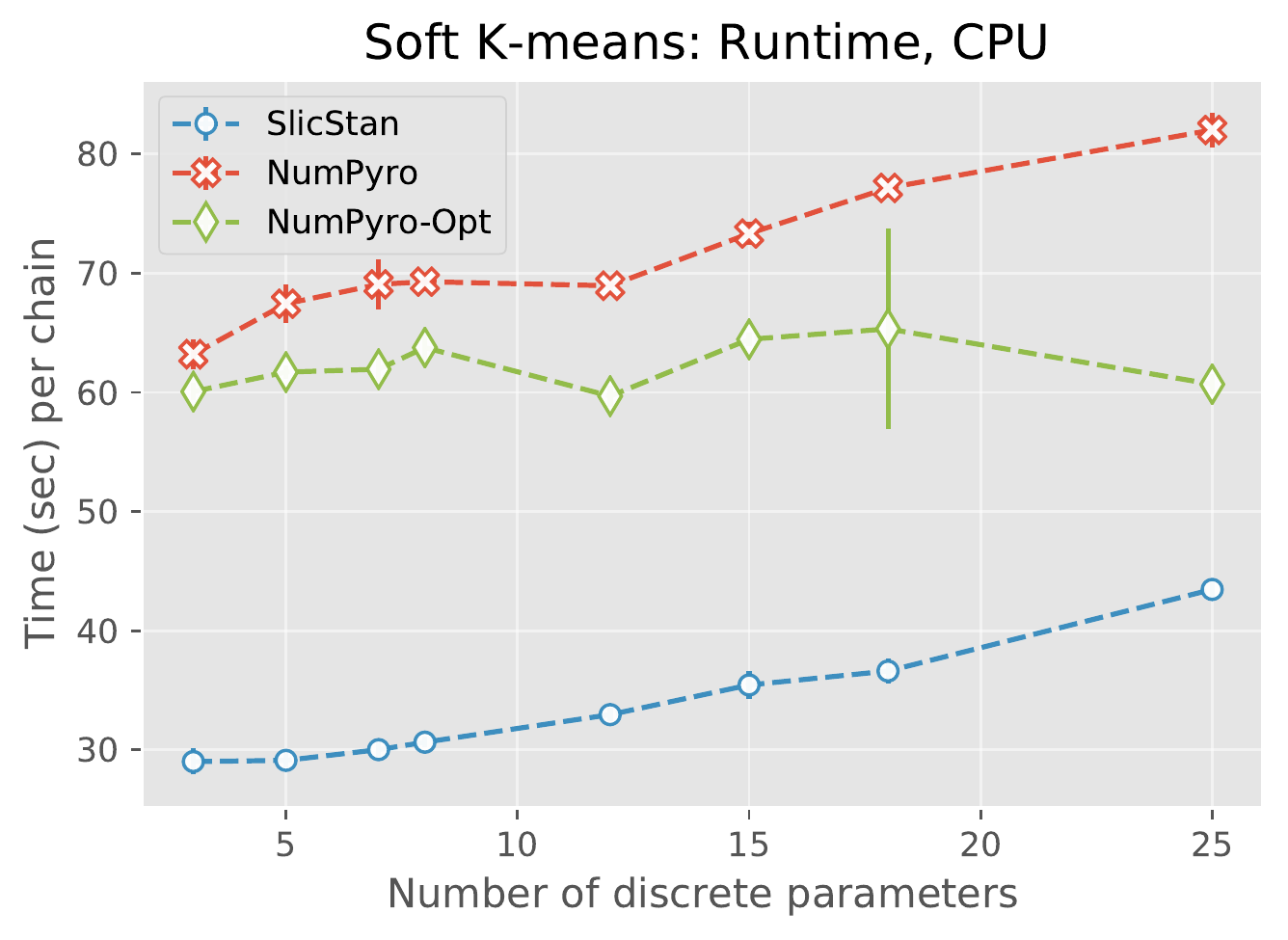}
	\end{subfigure}
	\begin{subfigure}{0.32\textwidth}
		\includegraphics[width=\textwidth]{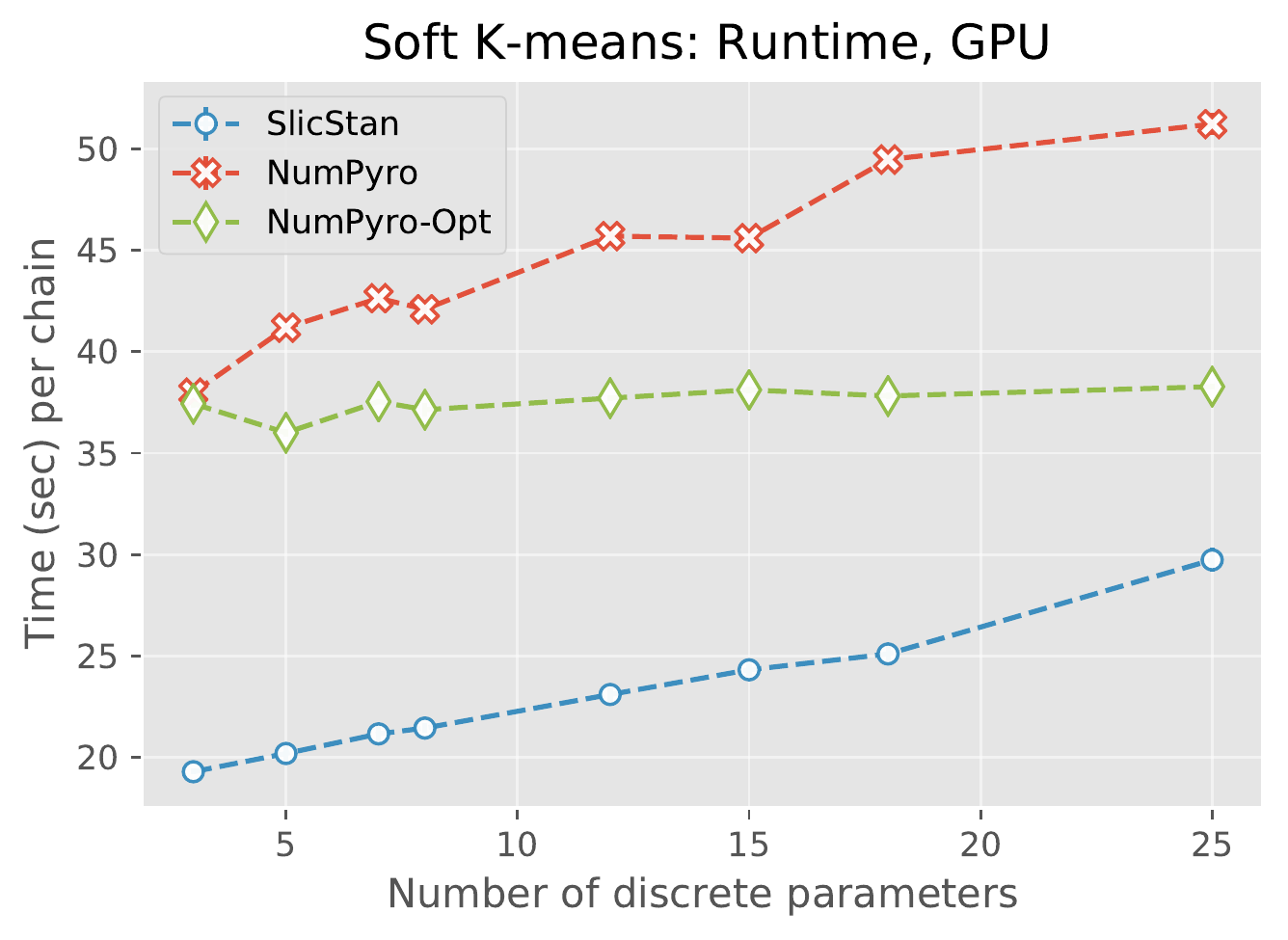}
	\end{subfigure}
	\begin{subfigure}{0.32\textwidth}
		\includegraphics[width=\textwidth]{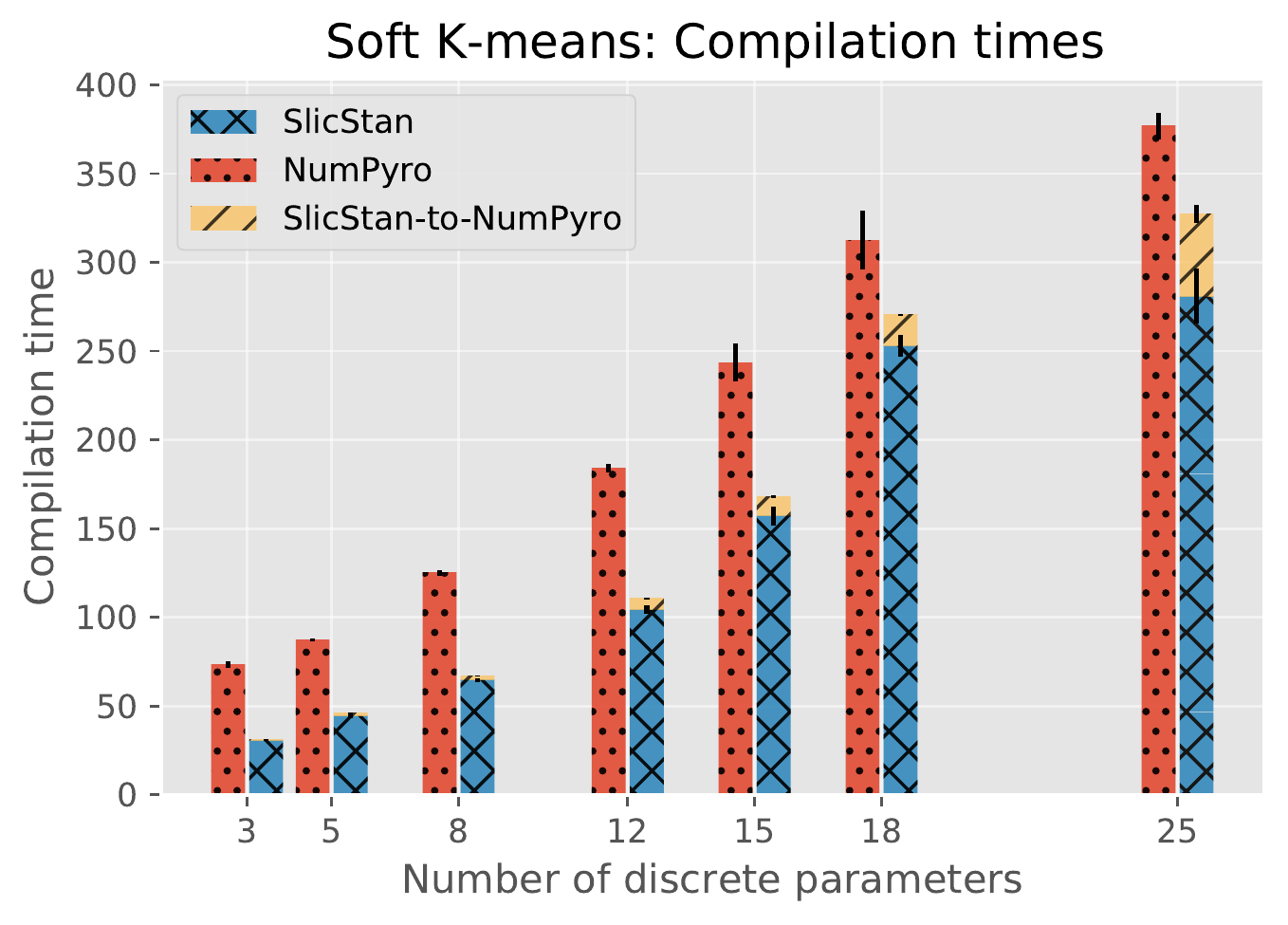}
	\end{subfigure}
	\vspace{-8pt}
	\caption{Soft K-means results}
	\label{fig:mixture}
\end{figure}

\begin{figure}[!p]
	\begin{subfigure}{0.32\textwidth}
		\includegraphics[width=\textwidth]{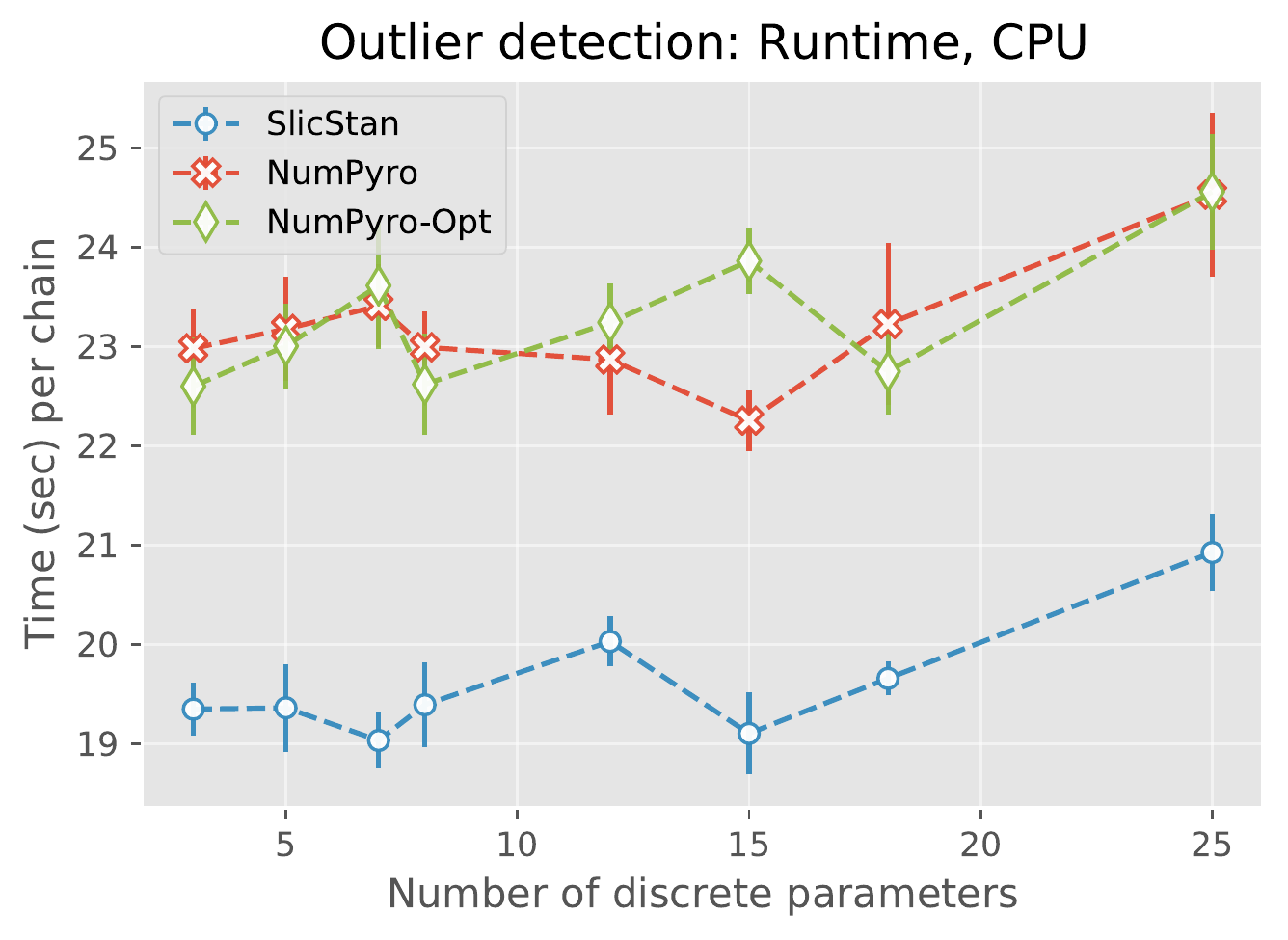}
	\end{subfigure}
	\begin{subfigure}{0.32\textwidth}
		\includegraphics[width=\textwidth]{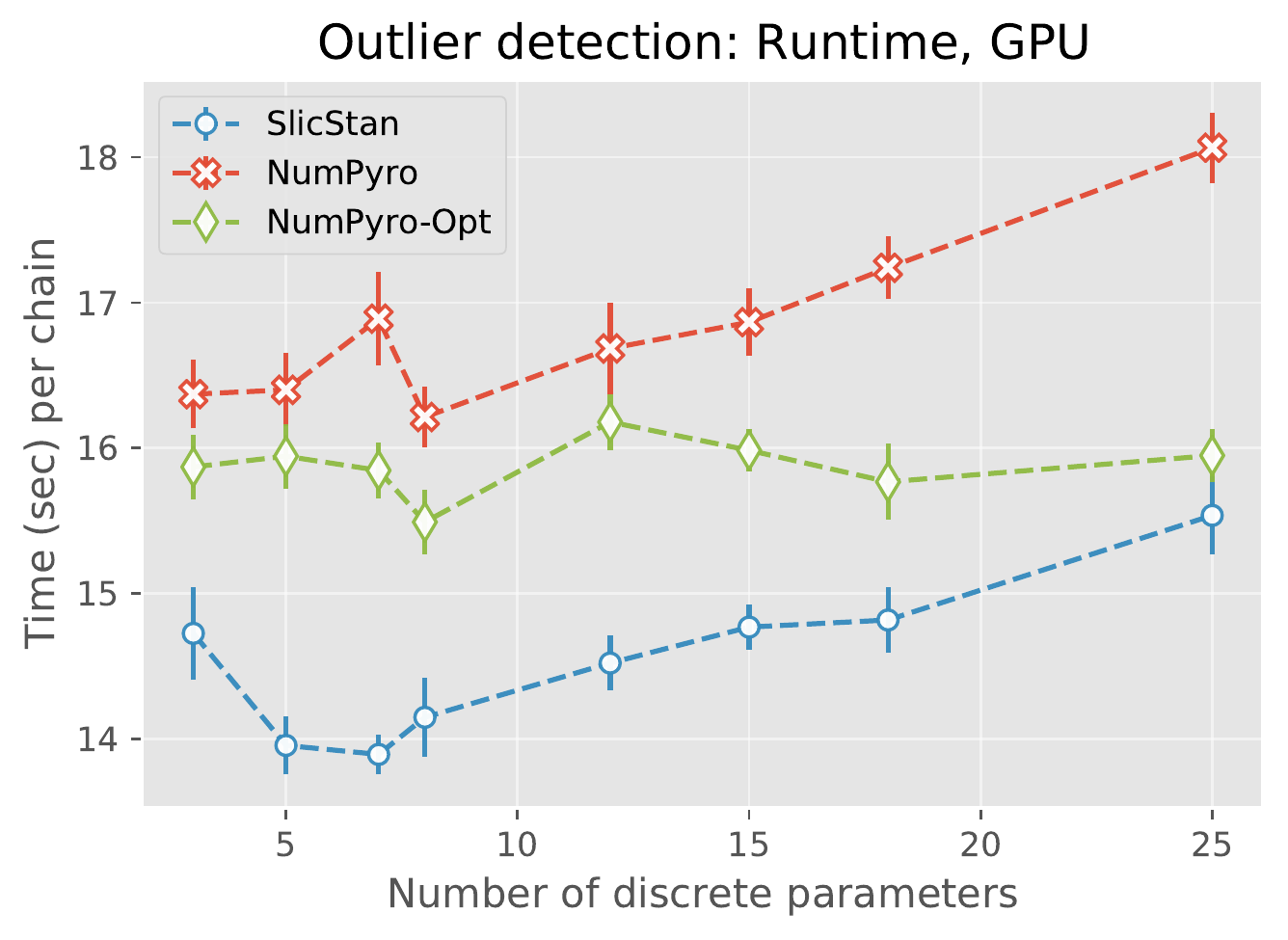}
	\end{subfigure}
	\begin{subfigure}{0.32\textwidth}
		\includegraphics[width=\textwidth]{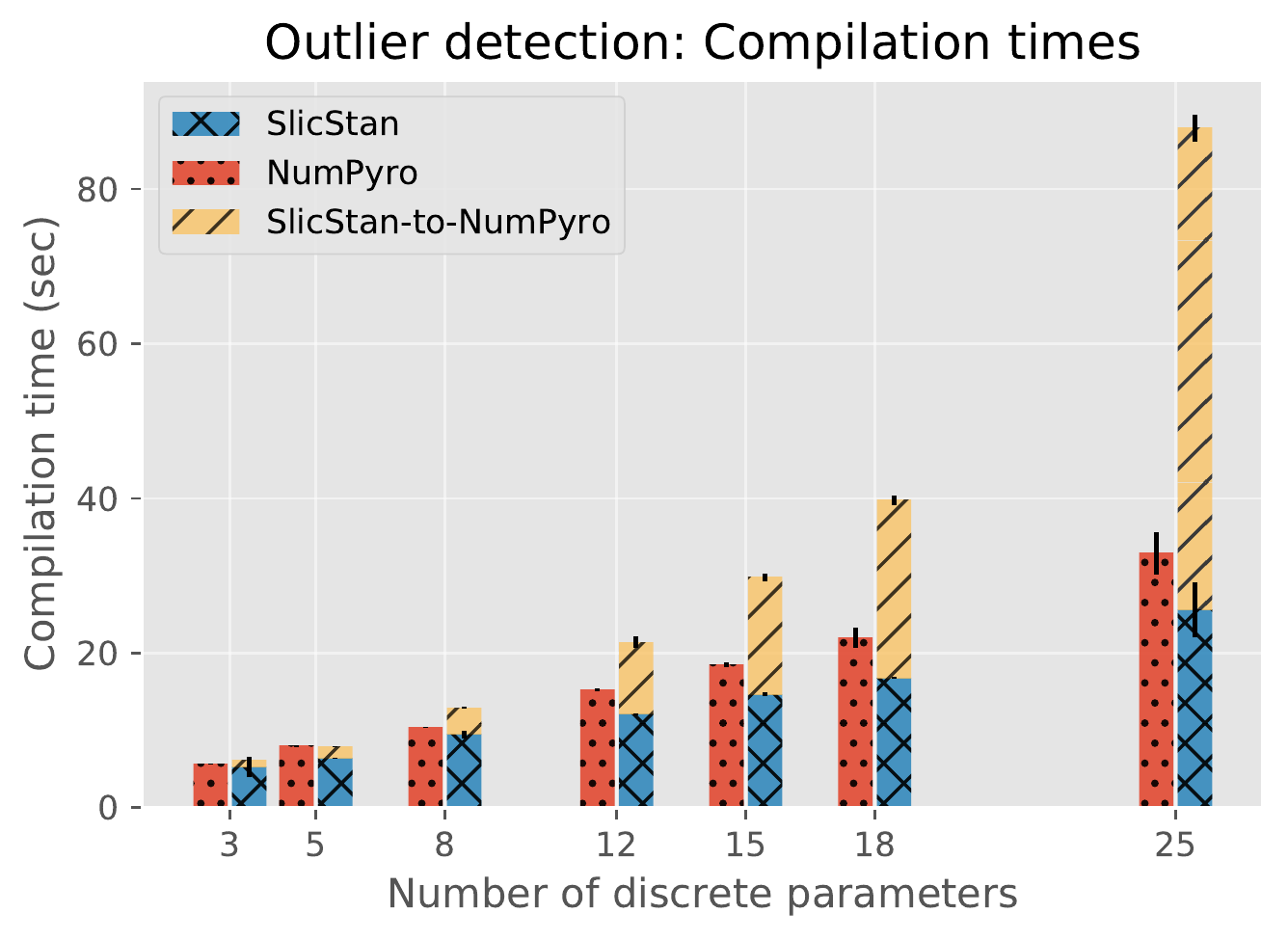}
	\end{subfigure}
	\vspace{-8pt}
	\caption{Outliers results}
	\label{fig:outliers}
\end{figure}

\paragraph{Factorial HMM}

In a factorial HMM, each data point $y_n$ is generated using two independent hidden states $z_n$ and $h_n$, each depending on the previous hidden states $z_{n-1}$ and $h_{n-1}$. 
\begin{gather*}
	\mu_k \sim \mathcal{N}(0, 1) \quad\text{for } k \in 1, \dots, K^2 \\
	z_1 \sim \mathrm{categorical}(\theta_{1}), \quad h_1 \sim \mathrm{categorical}(\theta_{1}) \\
	z_n \sim \mathrm{categorical}(\theta_{z_{n-1}}), \quad
	h_n \sim \mathrm{categorical}(\theta_{h_{n-1}}) \quad\text{for } n \in 2, \dots, N \\
	y_n \sim \mathcal{N}(\mu_{z_{n} * h_{n}}, 1) \quad\text{for } n \in 1, \dots, N \\
\end{gather*}

\vspace{-2pt}
We run the experiment for $K = 3$ and different length of the chain $N$, ranging from $N=1$ (2 discrete parameters) to $N=12$ (24 discrete parameters). 
We show the results in 
\autoref{fig:hmm-factorial}: similarly to before, SlicStan outperforms both NumPyro and NumPyro-Opt in terms of runtime. 
We also observe that, in the case of SlicStan, the time taken to sample a single chain increases more slowly as we increase the number of discrete variables.

\subsubsection{Mixture models} \label{sssec:mixture_eval}
Another useful application of mixed discrete and continuous variable models is found in mixture models.
We run experiments on two models: soft $K$-means clustering and linear regression with outlier detection.

\paragraph{Soft K-means}

The Gaussian mixture model underlines the celebrated soft $K$-means algorithm.
Here, we are interested in modelling some $D$-dimensional data that belongs to one of $K$ (unknown) Gaussian clusters. Each cluster $k$ is specified by a $D$-dimensional mean $\mu_{., k}$. Each data point $y_{.,n}$ is associated with a cluster $z_n$. 
\begin{gather*}
	\mu_{d, k} \sim \mathcal{N}(0, 1) \quad \text{for } d \in 1, \dots, D \text{ and } k \in 1, \dots, K\\
	z_n \sim \mathrm{categorical}(\pi) \quad \text{for } n \in 1, \dots, N \\
	y_{d, n} \sim \mathcal{N}(\mu_{d, z_n}, 1) \quad \text{for } d \in 1, \dots, D \text{ and } n \in 1, \dots, N
\end{gather*}

\vspace{-2pt}
We run the experiments for $K = 3$, $D = 10$, and $N = 3, \dots, 25$ and show the results in 
\autoref{fig:mixture}.
 
We observe a clear linear trend of the runtime growing with N, with SlicStan performing better and its runtime growing more slowly than that of NumPyro. While the SlicStan-translated code runs faster than NumPyro-Opt for $N \leq 25$, we observe that the SlicStan runtime grows faster than that of the manually optimised NumPyro-Opt code. 

\paragraph{Outlier detection} 

The final model we consider is a Bayesian linear regression that allows for outlier detection. The model considers data points $(x_n, y_n)$, where $y$ lies on the line $\alpha x + \beta$ with some added noise. The noise $\sigma_{z_n}$ depends on a Bernoulli parameter $z_n$, which corresponds to whether or not the point $(x_n, y_n)$ is an outlier or not. The noise for outliers ($\sigma_1$) and the noise for non-outliers ($\sigma_2$) are given as hyperparameters.
\begin{gather*}
	\alpha \sim \mathcal{N}(0, 10), \quad \beta \sim \mathcal{N}(0, 10)\\
	\pi^{(raw)}_1 \sim \mathcal{N}(0, 1), \quad \pi^{(raw)}_2 \sim \mathcal{N}(0, 1), \quad \pi = \frac{\exp{\pi^{(raw)}_1}}{\exp{\pi^{(raw)}_1} + \exp{\pi^{(raw)}_2}} \\
	z_n \sim \mathrm{bernoulli}(\pi) \quad \text{for } n \in 1, \dots, N \\
	y_{n} \sim \mathcal{N}(\alpha x_n + \beta, \sigma_{z_n}) \quad \text{for } n \in 1, \dots, N
\end{gather*}

Similarly to the earlier HMM models, SlicStan has the smallest runtime per chain, but at the expense of fast growing compile time (\autoref{fig:outliers}).

\subsection{Analysis and discussion} \label{ssec:results}


Our method can be applied to general models containing a fixed and known number of finite-support discrete parameters, which significantly reduces the amount of manual effort that was previously required for such models in languages like Stan \cite{StanHMM}.
In addition, as shown in Figures \ref{fig:hmm-1}--\ref{fig:outliers}, SlicStan outperforms both the NumPyro baseline and the hand-optimised NumPyro-Opt, in terms of runtime. 
This suggests that a static-time discrete variable optimisation, like the one introduced in this paper, is indeed beneficial and speeds up inference. 

One limitation of our experimental analysis is the relatively small number of discrete parameters we consider. 
Due to the array dimension limit imposed by PyTorch / NumPy,
Pyro cannot have more than 25 discrete variables (64 for CPU) unless the dependence between them is specified using \lstinline{markov} or \lstinline{plate} (as with NumPyro-Opt). For NumPyro this hardcoded limit is 32. Thus, it would not be possible to compare to the NumPyro baseline for a larger number of variables, though comparing to the hand-optimised NumPyro-Opt would still be possible.

Perhaps the biggest limitation of the discrete parameters version  of SlicStan is the exponentially growing compilation time.  
Using a semi-lattice instead of a lattice in the $\typingspecial$ level type analysis breaks the requirement of the bidirectional type system that ensures efficiency of type inference.
The constraints generated by the type system can no longer be resolved by SlicStan's original linear-time algorithm.  
While polynomial-time constraint-solving strategy may still exist, we choose to employ Z3 to automatically resolve the type inference constraints, and leave the consideration for more efficient type inference algorithm for future work. 
%

This also highlights the importance of a future SlicStan version that considers \emph{arrays} of discrete parameters. 
Our algorithm currently supports only individual discrete parameters.
In the cases where the size of an array of discrete parameters is statically known, the \ref{Elim Gen} procedure can be applied to a program where such arrays have been `flattened' into a collection of individual discrete variables, which is the strategy we adopt for the experiments in this section. 
But to be applicable more widely, the \ref{Elim Gen} rule needs to be generalised based on array element level dependence analysis, for example by incorporating ideas from the polyhedral model \cite{feautrier1992some}.
As the array level dependence analysis that would be required in most practical use-cases is very straightforward, we believe this would be a useful and feasible applied extension of our work. 
In addition, this would significantly decrease the number of program variables for which we need to infer a level type during the \ref{Elim Gen} transformation, thus making compilation practical for larger or arbitrary numbers of discrete parameters.

\section{Related work}

This paper provides a type system that induces conditional independence relationships, and it discusses one practical application of such type system: an automatic marginalisation procedure for discrete parameters of finite support.

\paragraph{Conditional independence}
The theoretical aim of our paper is similar to that of \citet{barthe2019probabilistic}, who discuss a separation logic for reasoning about independence, and the follow-up work of \citet{bao2021bunched}, who extend the logic to capture conditional independence. One advantage of our method is that the verification of conditional independence is automated by type inference, while it would rely on manual reasoning in the works of  \citet{barthe2019probabilistic} and \citet{bao2021bunched}. On the other hand, the logic approach can be applied to a wider variety of verification tasks.
\citet{amtoft2020theory} show a correspondence between variable independence and slicing a discrete-variables-only probabilistic program. 
The biggest difference to our work is that their work considers only conditional independence of variables given the observed data: that is CI relationships of the form $\mathbf{x}_1 \bigCI \mathbf{x}_2 \mid \data$ for some subsets of variables $\mathbf{x}_1$ and $\mathbf{x}_2$ and data $\data$.
The language of \citet{amtoft2020theory}  requires observed data to be specified \emph{syntactically} using an \lstinline{observe} statement. Conditional independencies are determined only given this observed data, and the method for determining how to slice a program is tied to the \lstinline{observe} statements. From the \citet{amtoft2020theory} paper:
``A basic intuition behind our approach is that an \lstinline{observe} statement can be removed if it does not depend on something on which the returned variable $x$ also depends.''
In contrast, we are able to find CI relationships given any variables we are interested in ($\mathbf{x_1} \bigCI \mathbf{x_2} \mid \mathbf{x}_3$ for some $\mathbf{x}_1$, $\mathbf{x}_2$, and $\mathbf{x}_3$), and type inference constitutes of a straightforward algorithm for finding such relationships. 
On the other hand, \citet{amtoft2020theory} permit unbounded number of variables (e.g. \lstinline{while (y > 0) $\,$ y ~ bernoulli(0.2))}, while it is not clear how to extend SlicStan/Stan to support this. While not in a probabilist programming setting, \citet{lobo2020programming} use taint analysis to find independencies between variables in a program, in order to facilitate easy trade off between privacy and accuracy in differential privacy context.

\paragraph{Automatic marginalisation}

The most closely related previous work, in terms of the automatic marginalisation procedure, is that of \citet{PyroDiscrete} and that of \citet{RaoBlackPPL}.
\citet{PyroDiscrete} implement efficient variable-elimination for plated factor graphs in Pyro \cite{Pyro}. Their approach uses effect-handlers and can be implemented in other effect-handling based PPLs, such as Edward2 \cite{Edward2}. 
\citet{RaoBlackPPL} introduce a `delayed sampling' procedure in Birch \cite{Birch}, which optimises the program via partial analytical solutions to sub-programs. Their method corresponds to automatic variable elimination and, more generally, automatic Rao–Blackwellization.
While we focus on discrete variable elimination only, our conditional independence type system can be directly used for more general analysis. The method from \autoref{sec:application} can be extended to marginalise out and sample continuous variables whenever they are part of an analytically-tractable sub-program, similarly to delayed sampling in Birch. 
One key difference of our approach is that the program re-writes are guided by the type system and happen at compile time, before inference is run. In contrast, both Pyro and Birch maintain a dynamic graph that guides the analysis at runtime. 

\paragraph{Symbolic inference}

Where a full analytical solution is possible, several probabilistic programming languages can derive it via symbolic manipulation, including Hakaru \cite{Hakaru} and PSI \cite{PSI, LambdaPSI}, while Dice \cite{holtzen2020dice}  performs exact inference for models with discrete parameters only, by analysing the program structure. 
In contrast, we focus on re-writing the program, and decomposing it into parts to be used with fast and more general asymptotically exact or approximate inference algorithms, like HMC, variational inference or others.

\paragraph{Extending HMC to support discrete parameters}

The idea of modifying HMC to handle discrete variables and discontinuities has been previously explored \cite{Zhou20, DHMCCharles, Pakman2013auxiliary, DHMC}.
More recently, \citet{DHMC-PPL} introduced the probabilistic programming language LF-PPL, which is designed specifically to be used with the Discontinuous Hamiltonian Monte Carlo (DHMC) algorithm \cite{DHMC}. The algorithm, and their framework can also be extended to support discrete parameters.
LF-PPL provides support for an HMC version that itself works with discontinuities. Our approach is to statically rewrite the program to match the constraints of Stan, vanilla HMC, and its several well-optimised extensions, such as NUTS \cite{NUTS}.

\paragraph{Composable and programmable inference}

Recent years have seen a growing number of techniques that allow for tailored-to-the-program compilation to an inference algorithm. For example, 
Gen \cite{Gen} can statically analyse the model structure to compile to a more efficient inference strategy. 
In addition, languages like Gen and Turing \cite{Turing} facilitate composable and programmable inference \cite{ProgrammableInference}, where the user  is provided with inference building blocks to implement their own model-specific algorithm. 
Our method can be understood as an automatic composition between two inference algorithms: variable elimination and HMC or any other inference algorithm that can be used to sample continuous variables. 

\section{Conclusion}

This paper introduces an information flow type system that can be used to check and infer conditional independence relationships in a probabilistic programs, through type checking and inference, respectively. We present a practical application of this type system: a semantics-preserving transformation that makes it possible to use, and to efficiently and automatically infer discrete parameters in SlicStan, Stan, and other density-based probabilistic programming languages. 
The transformed program can be seen as a hybrid inference algorithm on the original program, where continuous parameters can be drawn using efficient gradient-based inference methods, like HMC, while the discrete parameters are drawn using variable elimination. 

While the variable elimination transformation uses results on conditional independence of discrete parameters, our type system is not restricted to this usage. Conditional independence relationships can be of interest in many context in probabilistic modelling, including more general use of variable elimination, message-passing algorithms, Rao-Blackwellization, and factorising a program for a composed-inference approach. We believe conditional independence by typing can enable interesting future work that automates the implementation of such methods. 

\begin{acks} 
We thank Vikash Mansinghka for suggesting the outlier detection example, which we used for evaluation, 
as well as Lawrence Murry for clarifying the behaviour of Birch, and
anonymous reviewers whose helpful suggestions improved the paper. 
Maria Gorinova was supported by the EPSRC Centre for Doctoral Training
in Data Science, funded by the UK Engineering and Physical Sciences Research Council (grant EP/L016427/1) and the University of Edinburgh.
Matthijs V\'ak\'ar was funded by the European Union’s Horizon 2020 research and innovation
programme under the Marie\linebreak Skłodowska-Curie grant agreement No. 895827.

\end{acks}

\bibliography{bibfile}

\appendix
\clearpage

\newcommand{\allp}{\data, \params, Q}

\section{Definitions and Proofs} \label{ap:proofs}

\subsection{Definitions}

\begin{definition}[Assigns-to set $\assset(S)$] \label{def:assset}
	\assset(S) is the set that contains the names of global variables that have been assigned to within the statement S. It is defined recursively as follows:	\vspace{-10pt}
	\begin{multicols}{2}\noindent
		$\assset(x[E_1]\dots[E_n] = E) = \{x\}$ \\
		$\assset(S_1; S_2) = \assset(S_1) \cup \assset(S_2) $ \\
		$\assset(\kw{if}(E)\; S_1 \;\kw{else}\; S_2) = \assset(S_1)\cup \assset(S_2) $\\
		$\assset(\kw{for}(x\;\kw{in}\;E_1:E_2)\;S) = \assset(S) \setminus \{x\}$\\
		$\assset(\kw{skip}) = \emptyset $ \\
		$\assset(\kw{factor}(E))=\emptyset$\\
		$\assset(L \sim d(E_1,\ldots, E_n))=\emptyset $
	\end{multicols}\vspace{-12pt}
\end{definition} 

\begin{definition}[Reads set $\readset(S)$] \label{def:readset}
	\readset(S) is the set that contains the names of global variables that have been read within the statement S. It is defined recursively as follows:	\vspace{-10pt}
	\begin{multicols}{2}\noindent
		$\readset(x) = \{x\}$ \\
		$\readset(c) = \emptyset$ \\
		$\readset([E_1,\dots,E_n]) = \bigcup_{i=1}^n\readset(E_i)$ \\
		$\readset(E_1[E_2]) = \readset(E_1) \cup \readset(E_2) $ \\
		$\readset(f(E_1,\dots,E_n)) = \bigcup_{i=1}^n\readset(E_i)$\\
		$\readset([E| x\;\kw{in}\;E_1:E_2])=\readset(E)\cup\readset(E_1)\cup\readset(E_2)$\\
		$\readset(\kw{target}(S))=\readset(S)$\\
		$\readset(x[E_1]\dots[E_n] = E) = \bigcup_{i=1}^n\readset(E_i) \cup \readset(E)$ \\
		$\readset(S_1; S_2) = \readset(S_1) \cup \readset(S_2) $ \\
		$\readset(\kw{if}(E)\; S_1 \;\kw{else}\; S_2) = \readset(E)\cup\readset(S_1)\cup \readset(S_2) $\\
		$\readset(\kw{for}(x\;\kw{in}\;E_1:E_2)\;S) = \readset(E_1) \cup \readset(E_2) \cup \readset(S) \setminus \{x\}$\\
		$\readset(\kw{skip}) = \emptyset $ \\	
		$\readset(\kw{factor}(E)) = \readset(E)$\\
		$\readset(L\sim d(E_1,\ldots, E_n))= \readset(L)\cup \readset(E_1)\cup\cdots\cup \readset(E_n)$
\end{multicols}\vspace{-12pt}
\end{definition}

\begin{definition}[Samples-to set $\tildeset(S)$] \label{def:tildeset}
	$\tildeset(S)$ is the set that contains the names of global variables that have been sampled 
	within the statement S. It is defined recursively as follows:	\vspace{-10pt}
	\begin{multicols}{2}\noindent
		$\tildeset(L= E) = \emptyset$ \\
		$\tildeset(S_1; S_2) = \tildeset(S_1) \cup \tildeset(S_2) $ \\
		$\tildeset(\kw{if}(E)\; S_1 \;\kw{else}\; S_2) = \tildeset(S_1)\cup \tildeset(S_2) $\\
		$\tildeset(\kw{for}(x\;\kw{in}\;E_1:E_2)\;S) = \tildeset(S) \setminus \{x\}$\\
		$\tildeset(\kw{skip}) = \emptyset $ \\
		$\tildeset(\kw{factor}(E))=\emptyset$\\
		$\tildeset(x[E_1]\dots[E_n] \sim d(E_1,\ldots, E_n))=\{x\} $
	\end{multicols}\vspace{-12pt}
\end{definition} 

\begin{definition}[Free variables $\freevars(S)$]
	$\freevars(S)$ is the set that contains the free variables that are used in a
	statement $S$. It is recursively defined as follows:
\begin{multicols}{2}\noindent 
$\freevars(x)=\{x\}$\\
$\freevars(c)=\emptyset$ \\
$\freevars([E_1,\dots,E_n]) = \bigcup_{i=1}^n\freevars(E_i)$ \\
$\freevars(E_1[E_2]) = \freevars(E_1) \cup \freevars(E_2) $ \\
$\freevars(f(E_1,\dots,E_n)) = \bigcup_{i=1}^n\freevars(E_i)$\\
$\freevars([E| x\;\kw{in}\;E_1:E_2])=\freevars(E)\cup\freevars(E_1)\cup\freevars(E_2)$\\
$\freevars(\kw{target}(S))=\freevars(S)$\\
$\freevars(x[E_1]\dots[E_n] = E) = \bigcup_{i=1}^n\freevars(E_i) \cup \freevars(E)$ \\
$\freevars(S_1; S_2) = \freevars(S_1) \cup \freevars(S_2) $ \\
$\freevars(\kw{if}(E)\; S_1 \;\kw{else}\; S_2) = \freevars(E)\cup\freevars(S_1)\cup \freevars(S_2) $\\
$\freevars(\kw{for}(x\;\kw{in}\;E_1:E_2)\;S) = \freevars(E_1) \cup \freevars(E_2) \cup \freevars(S) \setminus \{x\}$\\
$\freevars(\kw{skip}) = \emptyset $ \\	
$\freevars(\kw{factor}(E)) = \freevars(E)$\\
$\freevars(L\sim d(E_1,\ldots, E_n))= \freevars(L)\cup \freevars(E_1)\cup\cdots\cup\\ \quad\freevars(E_n)$
\end{multicols}
\end{definition}

\begin{definition} \label{def:gammaE} 
	We overload the notation $\Gamma(L)$ that looks up the type of an L-value in $\Gamma$. When applied to a more general expression $E$, $\Gamma(E)$ looks up the \textit{type level} of $E$ in $\Gamma$:  
	\begin{multicols}{2}\noindent
	$\Gamma(x) = \ell, \text{ where } \ell \text{ is the level of } x \text{ in } \Gamma$ \\
	$\Gamma(c) = \lev{data}$ \\
	$\Gamma([E_1,\dots,E_n]) = \bigsqcup_{i=1}^n\Gamma(E_i)$ \\
	$\Gamma(E_1[E_2]) = \Gamma(E_1) \sqcup \Gamma(E_2) $ \\
	$\Gamma(f(E_1,\dots,E_n)) = \bigsqcup_{i=1}^n\Gamma(E_i)$\\
	$\Gamma([E| x\;\kw{in}\;E_1:E_2])=\Gamma(E)\sqcup\Gamma(E_1)\sqcup\Gamma(E_2)$
	\end{multicols}\vspace{-12pt}
\end{definition}

\begin{definition} \label{def:gammaEs}
	$\Gamma(E_1, \dots, E_n) \equiv \Gamma(E_1) \sqcup \dots \sqcup \Gamma(E_n)$.
\end{definition}

\begin{definition}[$\readset_{\Gamma \vdash \ell}(S)$]\label{def:read_level_set}
	$\readset_{\Gamma \vdash \ell}(S)$ is the set that contains the names of global variables that have been read at level $\ell$ within the statement $S$. It is defined recursively
	 as follows:\\	
		$\readset_{\Gamma \vdash \ell}(x[E_1]\dots[E_n] = E) =
		\left\{
		\begin{array}{lll}	
		\bigcup_{i=1}^n\readset(E_i) \cup \readset(E) & \textnormal{if} & \Gamma(x)=(\_,\ell)\\
		\emptyset & \textnormal{otherwise}
		\end{array}\right.$ \\
		$\readset_{\Gamma \vdash \ell}(S_1; S_2) = \readset_{\Gamma \vdash \ell}(S_1) \cup \readset_{\Gamma \vdash \ell}(S_2) $ \\
		$\readset_{\Gamma \vdash \ell}(\kw{if}(E)\; S_1 \;\kw{else}\; S_2) = \readset_{\Gamma \vdash \ell}(E)\cup\readset_{\Gamma \vdash \ell}(S_1)\cup \readset_{\Gamma \vdash \ell}(S_2) $\\
		$\readset_{\Gamma \vdash \ell}(\kw{for}(x\;\kw{in}\;E_1:E_2)\;S) = \readset_{\Gamma \vdash \ell}(E_1) \cup \readset_{\Gamma \vdash \ell}(E_2) \cup \readset_{\Gamma \vdash \ell}(S) \setminus \{x\}$\\
		$\readset_{\Gamma \vdash \ell}(\kw{skip}) = \emptyset $ \\	
		$\readset_{\Gamma \vdash \ell}(\kw{factor}(E)) =
		\left\{\begin{array}{lll}
			\readset(E) & \textnormal{if}&\ell=\lev{model}\\
		\emptyset & \textnormal{else}\end{array}\right.$\\
		$\readset_{\Gamma \vdash \ell}(L\sim d(E_1,\ldots, E_n))=
		\left\{
		\begin{array}{ll}
			\readset(L\sim d(E_1,\ldots, E_n)) & \textnormal{if}
			\quad \ell = \bigsqcup \{\ell'\mid 
			\exists x\in \freevars(L\sim d(E_1,\ldots,E_n)\\
			\quad\quad\exists \tau. \Gamma(x)=(\tau,\ell')\}\\
			\emptyset & \textnormal{otherwise.} 
		\end{array}\right.$
\end{definition}


\begin{definition}[$\assset_{\Gamma \vdash \ell}(S)$]\label{def:write_level_set}
	$\assset_{\Gamma \vdash \ell}(S) \deq \{x \in \assset(S) \mid \Gamma(x) = (\tau, \ell) \text{ for some } \tau\}$
	\end{definition}
	\begin{definition}[$\tildeset_{\Gamma \vdash \ell}(S)$]\label{def:sample_level_set}
		$\tildeset_{\Gamma \vdash \ell}(S) \deq \{x \in \tildeset(S) \mid \Gamma(x) = (\tau, \ell) \text{ for some } \tau\}$
		\end{definition}	


\begin{definition} 
	Given a statement $S$, we define the statement $\store(S)$  by structural 
	induction on $S$:\\
	\begin{tabular}{l}
		$\store(x[E_1]\dots[E_n] = E) =x[E_1]\dots[E_n] = E$ \\
		$\store(S_1; S_2) = \store(S_1) ; \store(S_2) $ \\
		$\store(\kw{if}(E)\; S_1 \;\kw{else}\; S_2) = \kw{if}(E)\; \store(S_1) \;\kw{else}\; \store(S_2)) $\\
		$\store(\kw{for}(x\;\kw{in}\;E_1:E_2)\;S) = \kw{for}(x\;\kw{in}\;E_1:E_2)\;\store(S)$\\
		$\store(\kw{skip}) = \kw{skip} $ \\
		$\store(\kw{factor}(E))=\kw{skip} $\\
		$\store(L \sim d(E_1,\ldots, E_n))=\kw{skip}  $
	\end{tabular}
\end{definition}

\begin{definition}[Neighbours of $z$, $\mathrm{ne}(\Gamma, \Gamma', z)$]  ~ \\
	For a $\vdash$ typing environment $\Gamma$, a $\typingspecial$ typing environment $\Gamma' = \Gamma_{\sigma}', \Gamma_{\mathbf{x}}'$ and a variable $z \in \dom(\Gamma_{\mathbf{x}}')$, the neighbours of $z$ are defined as: 
	$$\mathrm{ne}(\Gamma, \Gamma', z) \deq \{x : (\tau, \ell) \in \Gamma_{\mathbf{x}}' \mid \ell = \lev{l1} \text{ and } \Gamma(x) = (\kw{int}\langle K\rangle, \lev{model}) \text{ for some } K \}$$
\end{definition}

\subsection{Proofs}

\begin{restate}{Lemma~\ref{lem:noninterf} (Noninterference of $\vdash$)} 
Suppose $s_1 \models \Gamma$, $s_2 \models \Gamma$, and $s_1 \approx_{\ell} s_2$ for some $\ell$. Then for SlicStan statement $S$ and expression $E$:
\begin{enumerate}
	\item If $~\Gamma \vdash E:(\tau,\ell)$ and $(s_1, E) \Downarrow V_1$ and $(s_2, E) \Downarrow V_2$ then $V_1 = V_2$. 
	\item If $~\Gamma \vdash S:\ell$ and $(s_1, S) \Downarrow s_1', w_1$ and $(s_2, S) \Downarrow s_2', w_2$ then $s_1' \approx_{\ell} s_2'$.
\end{enumerate}
\end{restate}
\begin{proof}
(1)~follows by rule induction on the derivation $\Gamma \vdash E:(\tau, \ell)$, and using that if $\Gamma \vdash E:(\tau, \ell)$, $x \in \readset(E)$ and $\Gamma(x) = (\tau', \ell')$, then $\ell' \leq \ell$. (2)~follows by rule induction on the derivation $\Gamma \vdash S:\ell$ and using (1).

Most cases follow trivially from the inductive hypothesis. An exception is the \ref{Target} case, which we show below.\\
\begin{tabular}{ll}\ref{Target} & \parbox{0.87\linewidth}{ We use the premise $\forall \ell' > \ell. R_{\Gamma \vdash \ell'}(S) = \emptyset$, together with a lemma that for $S$, $s_1$ and $s_2$ such that $s_1, S \Downarrow s_1', w_1$, and $s_2, S \Downarrow s_2', w_2$, and $\forall x \in R(S). s_1(x) = s_2(x)$, we have that $w_1 = w_2$. (This lemma follows by structural induction on $S$.) In the case of \ref{Target}, $s_1, \mathrm{target}(S) \Downarrow w_1$, and $s_2, \mathrm{target}(S) \Downarrow w_2$ and $R(S) = \bigcup_{\ell'} R_{\Gamma \vdash \ell'}(S) = \left(\bigcup_{\ell' \leq \ell} R_{\Gamma \vdash \ell'}(S) \right) \cup \left(\bigcup_{\ell' > \ell} R_{\Gamma \vdash \ell'}(S) \right) = \bigcup_{\ell' \leq \ell} R_{\Gamma \vdash \ell'}(S)$.
Then, for any $x \in R(S)$, $x \in R_{\Gamma \vdash \ell'}(S)$ for some $\ell' \leq \ell$ , so $\Gamma(x) = (\tau, \ell_x)$ such that $\ell_x \leq \ell' \leq \ell$. And thus, by definition of $\approx_\ell$, $s_1(x) = s_2(x)$ for any $x \in R(S)$. By applying the lemma above, we then get $w_1 = w_2$, as required.}
\end{tabular}

\end{proof}

\begin{restate}{Lemma~\ref{lem:shredisleveled} (Shredding produces single-level statements)} 
	$$ S \shred[\Gamma] \shredded \implies \singlelevelS{\lev{data}}{S_D} \wedge \singlelevelS{\lev{model}}{S_M} \wedge \singlelevelS{\lev{genquant}}{S_Q}$$
\end{restate}
\begin{proof}
By rule induction on the derivation of $S \shred S_D, S_M, S_Q$.
\end{proof}

\begin{restate}{Lemma~\ref{lem:single-lev-prop} (Property of single-level statements)} ~\\
	Let $~\Gamma_{\sigma}, \Gamma_{\mathbf{x}}\vdash S$ be SlicStan program, such that $S$ is single-level statement of level $\ell$, $\Gamma \vdash \ell(S)$. Then there exist unique functions $f$ and $\phi$, such that for any $\sigma, \mathbf{x} \models \Gamma_{\sigma}, \Gamma_{\mathbf{x}}$: 
	$$  \sem{S}(\sigma)(x) = f(\sigma_{\leq \ell}, \mathbf{x}_{\leq \ell})\cup \sigma_{> \ell} , \hquad \phi(\sigma_{\leq \ell})(\mathbf{x}_{\leq \ell}), $$
	where we write $\sigma_{\leq \ell}=\{(x\mapsto V)\in \sigma\mid \Gamma_{\sigma}(x)=(\_,\ell)\}$ and $\sigma_{>\ell}=\sigma\setminus \sigma_{\leq \ell}$.
\end{restate}
\begin{proof}	
	This property follows from noninterference (\autoref{lem:noninterf}), if we understand factor and sample statements as assignments to a reserved weight variables of different levels. 
	Let $\Gamma, S$ be a SlicStan program and suppose we obtain $S'$ by:
	\begin{itemize}
		\item Substituting every $\kw{factor}(E)$ statement with $w_{\ell} = w_{\ell} * E$, where $\Gamma(E) = \kw{real}, \ell$ and $w_{\lev{data}}$, $w_{\lev{model}}$ and $w_{\lev{qenquant}}$ are write-only, distinct and reserved variables in the program.
		\item Substituting every $L \sim d(E_1, \dots, E_n)$ statement with $w_{\ell} = w_{\ell} * d_{\mathrm{pdf}}(L \mid E_1, \dots, E_n)$, where $\Gamma(d_{\mathrm{pdf}}(L \mid E_1, \dots, E_n)) = \kw{real}, \ell$.
	\end{itemize} 
	
	Then for all $\sigma, \mathbf{x} \models \Gamma$, we have $\semp{S}(\sigma)(\mathbf{x}) = \prod_{\ell}\sigma'(w_{\ell})$, where $\sigma' = \sems{S'}(\sigma, \forall \ell. w_{\ell} \mapsto 1)(\mathbf{x})$. By non-interference (\autoref{lem:noninterf}), for any level $\ell$ and store $\sigma_2 \approx_{\ell} \sigma$, if $\sigma_2' = \sems{S'}(\sigma_2, \forall \ell. w_{\ell} \mapsto 1)(\mathbf{x})$, then $\sigma_2' \approx_{\ell} \sigma'$. 
	Thus $\sigma_2'(w_{\ell'}) = \sigma_2(w_{\ell'})$ for $\ell' \leq \ell$, and therefore, when $S$ is a single-level statement of level $\ell$, $\sems{S'}(\sigma, \forall \ell. w_{\ell} \mapsto 1)(\mathbf{x}) = f(\sigma_{\leq \ell}, \mathbf{x}_{\leq{\ell}}), \sigma_{>\ell}, w_{\leq \ell} \mapsto \phi(\sigma_{\leq \ell}, \mathbf{x}_{\leq{\ell}}), w_{> \ell} \mapsto 1$ , for some functions $f$ and $\phi$. 
	Finally, this gives us $\sems{S}(\sigma, \mathbf{x}) = (f(\sigma_{\leq \ell}, \mathbf{x}_{\leq{\ell}}), \sigma_{>\ell})$, $\semp{S}(\sigma, \mathbf{x}) = \phi(\sigma_{\leq \ell}, \mathbf{x}_{\leq{\ell}})$.	
\end{proof}

\begin{restate}{Lemma~\ref{lem:shred} (Semantic Preservation of $\shred$)}~ \\
	If $~\Gamma \vdash S:\lev{data} $ and $ S \shred[\Gamma] \shredded $ then $\sem{S} = \sem{S_D; S_M; S_Q}$.
\end{restate}
\begin{proof}
	Follows by adapting proof from \cite{SlicStanPOPL}.
\end{proof}

\begin{restate}{Lemma~\ref{lem:shred2} (Semantic Preservation of $\shred$ 2)}~ \\
	If $~\Gamma \typingspecial S:\lev{l1} $ and $ S \shred[\Gamma] S_1, S_2, S_3 $ then $\sem{S} = \sem{S_1; S_2; S_3}$.
\end{restate}
\begin{proof}
	Follows by adapting proof from \cite{SlicStanPOPL}.
\end{proof}

\begin{lemma} \label{lem:expreads}
For a SlicStan expression $E$ and a function $\phi(x, y) = V$, where $V$ is a value such that $(\sigma, x, y), E \Downarrow V$ for every $x$ and $y$ and some $\sigma$, if $x \notin \readset(E)$, then:
$$\exists \phi' \text{ such that } \phi(x, y) = \phi'(y) \text{ for all } x, y$$
\end{lemma}
\begin{proof}
By induction on the structure of $E$.
\end{proof}

\begin{restate}{Theorem~\ref{th:shred_gen}~(Shredding induces a factorisation of the density).}  ~ \\
	Suppose $\Gamma \vdash S : \lev{data}$ and $~S \shred[\Gamma] S_D, S_M, S_Q$
	and $\Gamma = \Gamma_{\sigma} \cup \Gamma_{\data} \cup \Gamma_{\params} \cup \Gamma_{\quants}$.
	%
	For all $\sigma$, $\data$, $\params$, and $\quants$:
	if $\sigma, \data, \params, \quants \models \Gamma_{\sigma}, \Gamma_{\data}, \Gamma_{\params}, \Gamma_{\quants}$,
	and $\semp{S}(\sigma)(\data, \params, \quants) \propto p(\data, \params, Q)$
	and  $\tildeset(S_Q)=\dom(\Gamma_Q)$ then:
	\begin{enumerate}
		\item $\semp{S_M}(\sigma_D)(\data, \params, \quants) \propto p(\params, \data)$
		\item $\semp{S_Q}(\sigma_M)(\data, \params, \quants) = p(Q \mid \params, \data)$
	\end{enumerate}
	where $\sigma_D = \sems{S_D}(\sigma)(\data, \params, \quants)$
	and $\sigma_M = \sems{S_M}(\sigma_D)(\data, \params, \quants)$. 
\end{restate}

\begin{proof}
We prove this by establishing a more general result:

For $\sigma, \data, \params, \quants \models \Gamma_{\sigma}, \Gamma_{\data}, \Gamma_{\params}, \Gamma_{\quants}$, $A = \tildeset(S_Q) \subseteq Q$ and some $B \subseteq Q \setminus A$, if $\semp{S}(\sigma)(\data, \params, \quants) \propto p(\data, \params, A \mid B)$ then:
\begin{enumerate}
	\item $\semp{S_D}(\sigma)(\data, \params, \quants) = 1$
	\item $\semp{S_M}(\sigma_D)(\data, \params, \quants) = p(\params, \data)$
	\item $\semp{S_Q}(\sigma_M)(\data, \params, \quants) = p(A \mid \params, \data, B)$
\end{enumerate}

Note that in the case where $\tildeset(S_Q) = Q$, we have $A = Q$ and $B = \emptyset$, and the original statement of the theorem, $\semp{S_Q}(\sigma_M)(\data, \params, \quants) = p(Q \mid \params, \data)$, holds. 

We prove the extended formulation above by induction on the structure of $S$ and use of \autoref{lem:sem_properties}, \autoref{lem:shredisleveled} and \autoref{lem:single-lev-prop}, \autoref{lem:shred}.

Take any  $\sigma, \data, \params, \quants \models \Gamma_{\sigma}, \Gamma_{\data}, \Gamma_{\params}, \Gamma_{\quants}$ and let 
\begin{lstlisting}
	$\Phi(S, S_D, S_M, S_Q) \deq $
		$\Gamma \vdash S: \lev{data} \wedge S \shred S_D, S_M, S_Q \wedge A = \tildeset(S_Q)$ 
		$\implies \exists B \subseteq Q \setminus A. \forall \sigma_D, \sigma_M. \left( \right.$
			$\semp{S}(\sigma)(\data, \params, Q) \propto p(\data, \params, A \mid B) \wedge  \sem{S_D}(\sigma)(\data, \params, Q) = \sigma_D \wedge  \sem{S_M}(\sigma_D)(\data, \params, Q) = \sigma_M $
									$\implies \semp{S_D}(\sigma)(\data) = 1$
										 $\left.\wedge \; \semp{S_M}(\sigma_D)(\data, \params) = p(\params, \data)\right.$
										 $\left.\wedge \; \exists B \subseteq Q \setminus \tildeset(S_Q). \semp{S_Q}(\sigma_M)(\data, \theta, Q) = p(A \mid \params, \data, B) \right)$
\end{lstlisting}

Take any $\Gamma, S, S_D, S_M, S_Q$ such that $S \shred S_D, S_M, S_Q$, $A = \tildeset(S_Q)$, and take any $\sigma, \data, \params, \quants \models \Gamma_{\sigma}, \Gamma_{\data}, \Gamma_{\params}, \Gamma_{\quants}$, an unnormalised density $p$ and $B \subseteq Q \setminus A$, such that $\semp{S}(\sigma)(\allp) \propto p(\data, \params, A \mid B)$. We prove by rule induction on the derivation of $S \shred S_D, S_M, S_Q$ that $\Phi(S, S_D, S_M, S_Q)$.

\ref{Shred Seq}\qquad Let $S = S_1; S_2$ and $S_1 \shred S_{1D}, S_{1M}, S_{1Q}$ and $S_2 \shred S_{2D}, S_{2M}, S_{2Q}$. Thus $S \shred (S_{1D}; S_{2D}), (S_{1M}; S_{2M}), (S_{1Q}; S_{2Q})$. 
	
	Assume $\Phi(S_1, S_{1D}, S_{1M}, S_{1Q})$ and $\Phi(S_2, S_{2D}, S_{2M}, S_{2Q})$. 
	
	Let:
	\begin{itemize}
		\item $A_1 = \tildeset(S_{1Q})$ and $B_1 \subseteq Q \setminus A_1$ is such that $\semp{S_{1Q}}(\sigma_M)(\allp) = p_1(A_1 \mid \data, \params, B_1)$.
		
		\item $\sem{S_1}(\sigma)(\allp) = \sigma'$.
		
		\item $\semp{S_1}(\sigma)(\allp) \propto p_1(\data, \params, A_1 \mid B_1)$.
		
		\item $A_2 = \tildeset(S_{2Q})$ and $B_2 \subseteq Q \setminus A_2$ is such that $\semp{S_{2Q}}(\sigma_M)(\allp) = p_2(A_2 \mid \data, \params, B_2)$.

		\item $\semp{S_2}(\sigma')(\allp) \propto p_2(\data, \params, A_2 \mid B_2)$.
	\end{itemize}
	
	Thus, by Lemma~\ref{lem:sem_properties}, $\semp{S} = \semp{S_1; S_2} = \semp{S_1} \times \semp{S_2}$, so $p(\data, \params, A \mid B) \propto p_1(\data, \params, A_1 \mid B_1)p_2(\data, \params, A_2 \mid B_2)$.

	For (1), we have $\forall \sigma \models \Gamma_\sigma. \semp{S_{1D}}(\sigma)(\allp) = \semp{S_{2D}}(\sigma)(\allp) = 1$. Thus, by Lemma~\ref{lem:sem_properties}, $\semp{S_{1D}; S_{2D}} = \semp{S_{1D}} \times \semp{S_{2D}} = 1$.
	
	From $\Phi(S_1, S_{1D}, S_{1M}, S_{1Q})$ and $\Phi(S_2, S_{2D}, S_{2M}, S_{2Q})$ we also have:
	\begin{itemize}
	\item $\semp{S_{1Q}}(\sigma_M)(\allp) = p(A_1 \mid \params, \data, B_1)$
	\item $\semp{S_{2Q}}(\sigma_M')(\allp) = p(A_2 \mid \params, \data, B_2)$
	\end{itemize}
	
	$$A = \tildeset(S_Q) = \tildeset(S_{1Q}; S_{2Q}) = \tildeset(S_{1Q}) \cup \tildeset(S_{2Q}) = A_1 \cup A_2$$
	
	From $S$ well typed, it must be the case that $A_1 \cap A_2 = \emptyset$. Thus, we write $A = A_1, A_2$.
	
	We will prove that the property holds for $B = B_1 \cup B_2 \setminus A_1 \setminus A_2$.

	By semantic preservation of $\shred$ (Lemma \ref{lem:shred}), $\semp{S_1} = \semp{S_{1D}; S_{1M}; S_{1Q}} = \semp{S_{1D}} \times \semp{S_{1M}} \times \semp{S_{1Q}} \propto 1 \times p_1(\params, \data) \times p_1(A_1 \mid \params, \data, B_1)$. Similarly, $\semp{S_2} \propto 1 \times p_2(\params, \data) \times p_2(A_2 \mid \params, \data, B_2) = p_2(\params, \data) p_2(A_2 \mid \params, \data, A_1, B_1)$.
	
	But $p(\params, \data, A \mid B) \propto p_1(\params, \data, A_1 \mid B_1)p_2(\params, \data, A_2 \mid B_2)$, so:
	$$p(\params, \data, A \mid B)  
	\propto p_1(\params, \data) p_1(A_1 \mid \params, \data, B_1) p_2(\params, \data) p_2(A_2 \mid \params, \data, A_1, B_1)$$
	
	So, 
	\begin{align*}
	p(\params, \data) &= \int p(\data, \params, A \mid B) p(B) dA dB \\
	&\propto \int p_1(\params, \data) p_1(A_1 \mid \params, \data, B_1) p_2(\params, \data) p_2(A_2 \mid \params, \data, A_1, B_1) p(B) dA_1 dA_2 dB  \\
	&\propto p_1(\params, \data)p_2(\params, \data) \int p(B) p_1(A_1 \mid \params, \data, B_1) p_2(A_2 \mid \params, \data, A_1, B_1) dA_1 dA_2 dB\\
	&= p_1(\params, \data)p_2(\params, \data) \int p(B) \left( \int p_1(A_1 \mid \params, \data, B_1) \left( \int p_2(A_2 \mid \params, \data, A_1, B_1) dA_2 \right) dA_1 \right) dB \\
	&= p_1(\params, \data)p_2(\params, \data)\\
	&\propto p_1(\params, \data)p_2(\params, \data) 
	\end{align*}
$$\text{Thus } \semp{S_M} = \semp{S_{1M}; S_{2M}} \propto p_1(\params, \data)p_2(\params, \data) \propto p(\params, \data)$$

Finally, for last property on $S$, we use the chain rule of probability, semantics property of sequencing, and the result from above to get:
\begin{align*}
p(A \mid \data, \params, B) &= \frac{p(\data, \params, A \mid B)}{p(\data, \params \mid B)} \\
&\propto \frac{p_1(\data, \params)p_2(\data, \params)p_1(A_1 \mid \data, \params, B_1)p_2(A_2 \mid \data, \params, B_2)}{p(\data, \params)} \times \frac{p(B)}{p(B \mid \data, \params)} \\ 
&\propto p_1(A_1 \mid \data, \params, B_1)p_2(A_2 \mid \data, \params, B_2) \\
&= \semp{S_{1Q}} \semp{S_{2Q}} = \semp{S_Q}   
\end{align*}

Thus:
\begin{align*}
p(A \mid \data, \params, B) &= \frac{ p_1(A_1 \mid \data, \params, B_1)p_2(A_2 \mid \data, \params, B_2)}{Z} 
\end{align*}

Where: 
\begin{align*}
Z &= \int  p_1(A_1 \mid \data, \params, B_1)p_2(A_2 \mid \data, \params, B_2) dA \\
&= \int  p_1(A_1 \mid \data, \params, B_1)\left( \int p_2(A_2 \mid \data, \params, B_2) dA_2 \right) dA_1 \\
&= 1
\end{align*}

So $Z = 1$, and $p(A \mid \data, \params, B) = p_1(A_1 \mid \data, \params, B_1)p_2(A_2 \mid \data, \params, B_2) = \semp{S_Q}$.
	
Thus:
\begin{itemize}
\item $\semp{S_D} = \semp{S_{1D}; S_{2D}} = 1$
\item $\semp{S_M} = \semp{S_{1M}; S_{2M}} \propto p_1(\params, \data)p_2(\params, \data) = p(\params, \data)$
\item $\semp{S_Q} = \semp{S_{1Q}; S_{2Q}} = p_1(A_1 \mid \params, \data, B_1)p_2(A_2 \mid \params, \data, A_1, B_1) = p(A_1, A_2 \mid \params, \data, B)$
\end{itemize} 

$\Phi((S_1;S_2), (S_{1D}; S_{2D}), (S_{1M}; S_{2M}), (S_{1Q}; S_{2Q}))$ from here.
\end{proof}

\begin{restate}{Lemma~\ref{lem:shredisleveled2} (Shredding produces single-level statements 2)} 
	$$ S \shred[\Gamma] S_1, S_2, S_3 \implies \singlelevelS{\lev{l1}}{S_1} \wedge \singlelevelS{\lev{l2}}{S_2} \wedge \singlelevelS{\lev{l3}}{S_3}$$
\end{restate}
\begin{proof}
	By rule induction on the derivation of $S \shred S_1, S_2, S_3$.
\end{proof}

\begin{restate}{Lemma~\ref{lem:shred2} (Semantic preservation of $\shred$, $\typingspecial$)} ~ \\
	If $~\Gamma \typingspecial S:\lev{l1} $ and $ S \shred[\Gamma] S_1, S_2, S_3 $ then $\sem{S} = \sem{S_1; S_2; S_3}$.
\end{restate}
\begin{proof}
\end{proof}

\begin{restate}{Lemma~\ref{lem:single-lev-prop2} (Property of single-level statements 2)} ~\\
	Let $~\Gamma_{\sigma}, \Gamma_{\mathbf{x}}, S$ be a SlicStan program, and $\Gamma \typingspecial S : \lev{l1}$, and $S$ is single-level statement of level $\ell$, $\Gamma \typingspecial \ell(S)$. Then there exist unique functions $f$ and $\phi$, such that for any $\sigma, \mathbf{x} \models \Gamma_{\sigma}, \Gamma_{\mathbf{x}}$: 
	\begin{enumerate}
		\item If $\ell = \lev{l1}$, then $\sem{S}(\sigma)(x) \hquad = \hquad \left(f(\sigma_{\lev{l1}}, \mathbf{x}_{\lev{l1}}), \sigma_{\lev{l2}}, \sigma_{\lev{l3}} \right), \qquad \qquad \; \phi(\sigma_{\lev{l1}})(\mathbf{x}_{\lev{l1}})$
		\item If $\ell = \lev{l2}$, then $\sem{S}(\sigma)(x) \hquad = \hquad \left(\sigma_{\lev{l1}}, f(\sigma_{\lev{l1}}, \sigma_{\lev{l2}}, \mathbf{x}_{\lev{l1}}, \mathbf{x}_{\lev{l2}}), \sigma_{\lev{l3}} \right), \hquad \phi(\sigma_{\lev{l1}}, \sigma_{\lev{l2}})(\mathbf{x}_{\lev{l1}}, \mathbf{x}_{\lev{l2}})$
		\item If $\ell = \lev{l3}$, then $\sem{S}(\sigma)(x) \hquad = \hquad \left(\sigma_{\lev{l1}}, \sigma_{\lev{l2}}, f(\sigma_{\lev{l1}}, \sigma_{\lev{l3}}, \mathbf{x}_{\lev{l1}}, \mathbf{x}_{\lev{l3}}) \right), \hquad \phi(\sigma_{\lev{l1}}, \sigma_{\lev{l3}})(\mathbf{x}_{\lev{l1}}, \mathbf{x}_{\lev{l3}})$ 
	\end{enumerate}	
\end{restate}
\begin{proof}
	By understanding factor and sample statements as assignment to a reserved weight variables of different levels (similarly to \autoref{lem:single-lev-prop}) and noninterference (\autoref{lem:noninterf2}).
\end{proof}

\begin{restate}{Lemma~\ref{lem:exists} (Existence of \lev{model} to \lev{genquant} transformation)}
	For any SlicStan program $\Gamma, S$ such that $\Gamma \vdash S : \lev{l1}$, 
	and a variable $z \in \dom(\Gamma)$ such that $\Gamma(z) = (\kw{int}\langle K \rangle, \lev{model})$,
	there exists a SlicStan program $\Gamma', S'$, such that,  
	$$\Gamma, S \xrightarrow{z} \Gamma', S' 
	\quad \text{and} \quad 
	\Gamma'(z) = (\kw{int}\langle K \rangle, \lev{genquant})$$
\end{restate}
\begin{proof}
	Take a SlicStan program $\Gamma, S$, a typing environment $\Gamma_M$, a variable $z$, and statements $S_D, S_M$ and $S_Q$, such that:
	$$\Gamma(z) = (\kw{int}\langle K \rangle, \lev{model}) \quad
	\Gamma \vdash S : \lev{data} \quad 
	\Gamma \xrightarrow{z} \Gamma_M \quad 
	S \shred S_{D}, S_{M}, S_{Q} \quad
	\Gamma_M \typingspecial S_M : \lev{l1}
	$$
	
	Take also statements $S_1, S_2, S_3$, and $S_M'$, and a typing environment $\Gamma_{\mathrm{ne}}$ such that 
	$$S_{M} \shred[\Gamma_M] S_1, S_2, S_3 \quad 
	\Gamma_{\mathrm{ne}} = \mathrm{ne}(\Gamma, \Gamma_M, z)$$
	$$S_M' = 
	S_1; 
	f = \phi(\Gamma_{\mathrm{ne}}) \{\kw{elim}(\kw{int}\langle K \rangle z)~S_2 \}; 
	\kw{factor}(f[\dom(\Gamma_{\mathrm{ne}})]); 
	S_3; \kw{gen}(z) S_2; \store(S_2)$$
	
	Let $\Gamma'$ is such that $\dom(\Gamma') = \dom(\Gamma) \cup \{f\}$ and for all $x : \tau, \ell \in \Gamma$:
	$$
	\Gamma'(x) =
	\begin{cases}
	(\tau, \ell) & \text{if } \ell \neq \lev{model} \\
	(\tau, \ell) & \text{if } \ell = \lev{model} \text{ and } \Gamma_M(x) \neq (\tau, \lev{l2}) \\
	(\tau, \lev{genquant}) & \text{if } \ell = \lev{model} \text{ and } \Gamma_M(x) = (\tau, \lev{l2}) \\
	\end{cases}
	$$
	
	By semantic preservation of shredding (\autoref{lem:shred}, \autoref{lem:shred2}) and type preservation of the operational semantics (\cite{SlicStanPOPL}), $\Gamma \vdash S_D; S_1; S_2; S_3; S_Q : \lev{data}$, and thus, by \ref{Seq}, $\Gamma \vdash S_D : \lev{data}$, $\Gamma \vdash S_1 : \lev{data}, \dots, \Gamma \vdash S_Q : \lev{data}$.
	
	By definition of $\Gamma'$, $\Gamma'_{\lev{data}} \subset \Gamma_{\lev{data}}$. $S_D$ is single-level of level $\lev{data}$ and $\Gamma \vdash S_D : \lev{data}$, so $\Gamma_{\lev{data}} \vdash S_D : \lev{data}$ and thus $\Gamma' \vdash S_D : \lev{data}$. Similarly, $\Gamma \vdash S_1 : \data$ and $\Gamma \vdash S_3 : \data$.
	
	$\Gamma \vdash S_2 : \lev{data}$, so using \ref{Phi}, \ref{Elim} and \ref{Factor}, and noting that by definition $\dom(\Gamma_{\mathrm{ne}}) \subset \dom(\Gamma_{M, \lev{l1}})$, so $\Gamma_{\mathrm{ne}} \subset \Gamma$, we can derive:
	$$\Gamma' \vdash f = \phi(\Gamma_{\mathrm{ne}}) \{\kw{elim}(\kw{int}\langle K \rangle z)~S_2 \}; 
	\kw{factor}(f[\dom(\Gamma_{\mathrm{ne}})]) : \lev{data}$$
	
	By $\Gamma \vdash S_2 : \lev{data}$ and the definition of $\Gamma'$, and using \ref{Gen} and definition of $\store$, we also derive:
	$$\Gamma' \vdash \kw{gen}(z)~ S_2; \store(S_2) : \lev{genquant}$$
	
	Finally, $S_Q$ is a single-level statement of level $\lev{genquant}$ and for all $x : \tau, \ell \in \Gamma$, $x : \tau, \ell' \in \Gamma$, where $\ell \leq \ell'$. Therefore, $\Gamma \vdash S_Q: \lev{data}$ implies $\Gamma' \vdash S_Q : \lev{data}$.
	
	Altogether, this gives us $\Gamma' \vdash S_D; S_M'; S_Q$, and so by \ref{Elim Gen}, $\Gamma, S \xrightarrow{z} \Gamma', S_D; S_M', S_Q$.
	
\end{proof}

\begin{lemma} \label{lem:s2}
Let $\Gamma, S$ be a SlicStan program, such that $\sigma, \mathbf{x} \models \Gamma$, $\sems{S}(\sigma)(\mathbf{x}) = \sigma'$ and $\semp{S}(\sigma)(\mathbf{x}) = \psi(\mathbf{x})$ for some function $\psi$. If  $f \notin \dom(\Gamma)$ is a fresh variable, $z, z_1, \dots z_n \in \dom(\Gamma_{\mathbf{x}})$ are discrete variables of base types $\kw{int}\langle K \rangle, \kw{int}\langle K_1 \rangle, \dots, \kw{int}\langle K_n \rangle$ respectively, and $S'$ is a statement such that
$$
S' = \quad f = \phi(\kw{int}\langle K_1 \rangle z_1, \dots \kw{int}\langle K_n \rangle z_n) \{\kw{elim}(\kw{int}\langle K \rangle z)~S \}; \quad \kw{factor}(f[z_1, \dots, z_n]);
$$

then  $\sems{S'}(\sigma)(\mathbf{x}) = \sigma''$  with $\sigma''[-f] = \sigma'$ and $\semp{S'}(\sigma)(\mathbf{x}) = \sum_{z=1}^{K}\psi(\mathbf{x})$.
\end{lemma}
\begin{proof}
By examining the operational semantics of assignment, $\kw{factor}$, and the derived forms $\kw{elim}$ and $\phi$.
\end{proof}

\begin{lemma} \label{lem:sgen}
Let $\Gamma, S$ be a SlicStan program, such that $\sigma, \mathbf{x} \models \Gamma$, $\sems{S}(\sigma)(\mathbf{x}) = \sigma'$ and $\semp{S}(\sigma)(\mathbf{x}) = \psi(\mathbf{x})$ for some function $\psi$. If $z \in \dom(\Gamma_{\mathbf{x}})$ is a discrete variable of base type $\kw{int}\langle K \rangle$, and $S'$ is a statement such that
$$
S' = \quad \kw{gen}(z)~ S; \quad \store(S);
$$

then  $\sems{S'}(\sigma)(\mathbf{x}) = \sigma'$, $\psi(\mathbf{x})$ is normalisable with respect to $z$ with $\psi(\mathbf{x}) \propto p(z \mid \mathbf{x} \setminus \{z\})$, and $\semp{S'}(\sigma)(\mathbf{x}) = p(z \mid \mathbf{x} \setminus \{z\})$.
\end{lemma}
\begin{proof}
By examining the operational semantics of $\sim$ and $\kw{target}$, and by induction on the structure of $S$ to prove $\sems{\store(S)} = \sems{S}$ and $\semp{\store(S)} = 1$.
\end{proof}


\begin{display}{Typing Rules for Derived Forms:}
	\quad
	\staterule{Elim}
	{  \Gamma' \vdash S : \lev{data} \quad \readset_{\Gamma \vdash \lev{genquant}}(S) = \emptyset \quad \Gamma' = \Gamma[z \mapsto \kw{int}\langle K \rangle, \lev{model}]}
	{ \Gamma \vdash \kw{elim}(\kw{int}\langle K \rangle z)~ S : \lev{model}}\quad 
	
	\\[\GAP]\quad
	\staterule{Gen}
	{ \Gamma(z) = (\kw{int}, \lev{genquant}) \quad \Gamma \vdash S : \lev{data} \quad }
	{ \Gamma \vdash \kw{gen}(\kw{int}\langle K \rangle\,z)~S : \lev{genquant} }\qquad 
	
	\\[\GAP]\quad
	\staterule{Phi}
	{\Gamma' \vdash S : \lev{data} \quad \forall \ell' > \ell. \readset_{\Gamma \vdash \ell'}(S) = \emptyset \quad \Gamma' = \Gamma[z_1 \mapsto (\kw{int}\langle K_1 \rangle, \ell), \dots, z_N \mapsto (\kw{int}\langle K_N \rangle, \ell)]}
	{\Gamma \vdash \phi(\kw{int}\langle K_1 \rangle~z_1, \dots, \kw{int}\langle K_N \rangle~z_N)~S : \kw{real}, \ell}\qquad
\end{display}

\begin{restate}{Theorem~\ref{th:sempreservation} (Semantic preservation of $\xrightarrow{z}$)} $ $ \\
	For SlicStan programs $\Gamma, S$ and $\Gamma', S'$, and a discrete parameter $z$: 
	$\Gamma, S \xrightarrow{z} \Gamma', S' \rightarrow \sem{S} = \sem{S'}$.	
\end{restate}
\begin{proof} $ $ \\
	
	Let $\Gamma, S$ and $\Gamma', S'$ be SlicStan programs, and $z$ be a discrete parameter, such that 
	$\Gamma, S \xrightarrow{z} \Gamma', S'$.
	Let $S \shred S_D, S_M, S_Q$, $S \shred[\Gamma'] S_D', S_M', S_Q'$, and $S_M \shred[\Gamma''] S_1, S_2, S_3$ for $\Gamma''$ such that $\Gamma \xrightarrow{z} \Gamma''$ and $\Gamma'' \typingspecial S_M : \lev{l1}$.  
	
	Let $\Gamma = \Gamma_{\sigma}, \Gamma_{\lev{data}}, \Gamma_{\lev{model}}, \Gamma_{\lev{genquant}}$,
	$\Gamma' = \Gamma_{\sigma}', \Gamma_{\lev{data}}', \Gamma_{\lev{model}}', \Gamma_{\lev{genquant}}'$
	and\\ $\Gamma'' = \Gamma_{\sigma}'', \Gamma_{\lev{l1}}'', \Gamma_{\lev{l2}}'', \Gamma_{\lev{l3}}''$
	be the usual partitioning of each of the typing environments. 

	Let $z$ be a store such that $z \models \{z : \Gamma(z)\}$.
	
	Let $\data, \params$ and $\quants$ be stores such that $\data \models \Gamma_{\lev{data}}$, $z, \params \models \Gamma_{\lev{model}}$, and $\quants \models \Gamma_{\lev{genquant}}$.
	
	Let $\params_1, \params_2$ and $\params_3$ be a partitioning of $\params$, such that $\data, \params_1 \models \Gamma_{\lev{l1}}''$, $z, \params_2 \models \Gamma_{\lev{l2}}''$, and  $\params_3 \models \Gamma_{\lev{l3}}''$.
	
	Then, by definition of $\Gamma \xrightarrow{z} \Gamma''$, $\params_2 = z$. 

	By \autoref{th:shred_gen}: 
	\begin{itemize}
		\item $\semp{S_D}(\sigma)(\data, z, \params, \quants) = 1$
		\item $\semp{S_M}(\sigma_D)(\data, z, \params, \quants) \propto p(z, \params, \data)$
		\item $\semp{S_Q}(\sigma_M)(\data, z, \params, \quants) = p(\quants \mid z, \params, \data)$
	\end{itemize}
		
	
	$\Gamma,  S \xrightarrow{d} \Gamma', S'$, thus $S'$ must be of the form 
	$$S' = S_D; 
	\hquad S_{1}; 
	\hquad f = \phi(\Gamma_{\lev{l1}''}) \{\kw{elim}(\kw{int}\langle K \rangle z)~S_{2} \}; 
	\hquad \kw{factor}(f[\dom(\Gamma_{\lev{l1}}'')]); 
	\hquad S_{3}; 
	\hquad \kw{gen}(z) S_2; 
	\hquad \store(S_2);
	\hquad S_Q $$ 
	where 
	$ \Gamma \vdash S : \lev{data}, \hquad
	S \shred S_{D}, S_{M}, S_{Q}, \hquad
	\Gamma \xrightarrow{z} \Gamma'', \hquad
	\Gamma \typingspecial S_M : \lev{l1}, \hquad$ and
	$S_M \shred[\Gamma''] S_{1}, S_{2}, S_{3}$.
	
	The relation $\shred$ is semantics-preserving for well-typed programs with respect to both $\vdash$ and $\typingspecial$ (\autoref{lem:shred} and \autoref{lem:shred2}). Thus $\sem{S} = \sem{S_D; S_1; S_2; S_3; S_Q}$.

	We present a diagrammatic derivation of the change on store and density that each sub-part in the original and transformed program makes in \autoref{fig:proof}.
	
	\begin{figure*}[]
	\centering
	\resizebox{\linewidth}{!}{
	\begin{tikzpicture}
	\tikzstyle{block} = [rectangle, fill=blue!0,
	text centered, minimum height=1em]
	
	\node (sigma) {\footnotesize$\sigma$};
	\node [below=of sigma] (sigmad) {\footnotesize$\sigma^{(D)}, 1$};
	\node [below=of sigmad] (sigmad1) {\begin{varwidth}{5em}\footnotesize\centering		
		$\sigma^{(D1)}$,\\
		$\phi_1(\sigma^{(D1)}_{\lev{l1}})(\data, \params_1)$
		\end{varwidth}};
	\coordinate [below=of sigmad1] (mid);
	\node [block, below left=of mid, xshift=-1cm] (sigmad2) {\begin{varwidth}{15em}\footnotesize\centering		
		$\left(\sigma^{(D1)}_{\lev{l1}}, 
			   f_2(\sigma^{(D1)}_{\lev{l1},\lev{l2}}), 
			   \sigma^{(D1)}_{\lev{l3}}\right)$, \\
			   $\phi_2(\sigma^{(D1)}_{\lev{l1}},f_2(\sigma^{(D1)}_{\lev{l1},\lev{l2}}))(\data, \params_1, z)$
			   \end{varwidth}};
	\node [below right=of mid, xshift=1cm] (sigmad2p) {\begin{varwidth}{15em}\footnotesize\centering		
		$\left(\sigma^{(D1)}_{\lev{l1}}, 
		\sigma^{(D1)}_{\lev{l2}}, f \mapsto v, 
		\sigma^{(D1)}_{\lev{l3}}\right)$, \\
		$\sum_{z}\phi_2(\sigma^{(D1)}_{\lev{l1}},f_2(\sigma^{(D1)}_{\lev{l1},\lev{l2}}))(\data, \params_1, z)$
		\end{varwidth}};
	\node [block, below=of sigmad2, yshift=-2.25cm] (sigmam) {\begin{varwidth}{15em}\footnotesize\centering		
		$\left(\sigma^{(D1)}_{\lev{l1}}, 
		f_2(\sigma^{(D1)}_{\lev{l1},\lev{l2}}), 
		f_3(\sigma^{(D1)}_{\lev{l1},\lev{l3}})\right)$, \\
		$\phi_3(\sigma^{(D1)}_{\lev{l1}},f_3(\sigma^{(D1)}_{\lev{l1},\lev{l3}}))(\data, \params_1, \params_3)$
		\end{varwidth}};
	\node [below=of sigmad2p] (sigmamp) {\begin{varwidth}{15em}\footnotesize\centering		
		$\left(\sigma^{(D1)}_{\lev{l1}}, 
		\sigma^{(D1)}_{\lev{l2}}, f \mapsto v,
		f_3(\sigma^{(D1)}_{\lev{l1},\lev{l3}})\right)$, \\
		$\phi_3(\sigma^{(D1)}_{\lev{l1}},f_3(\sigma^{(D1)}_{\lev{l1},\lev{l3}}))(\data, \params_1, \params_3)$
		\end{varwidth}};
	\node [below=of sigmamp] (sigmagen) {\begin{varwidth}{15em}\footnotesize\centering		
		$\left(\sigma^{(D1)}_{\lev{l1}}, 
		f_2(\sigma^{(D1)}_{\lev{l1},\lev{l2}}), f \mapsto v,
		f_3(\sigma^{(D1)}_{\lev{l1},\lev{l3}})\right)$, \\
		$p(z \mid \data, \params_1)$
		\end{varwidth}};
	\node [below=of sigmagen, yshift=-0.19cm] (sigmapp) {\begin{varwidth}{15em}\footnotesize\centering		
		$f_g\left(\sigma^{(D1)}_{\lev{l1}}, 
		f_2(\sigma^{(D1)}_{\lev{l1},\lev{l2}}), 
		f_3(\sigma^{(D1)}_{\lev{l1},\lev{l3}})\right), f \mapsto v$, \\
		$\phi_{g}(\sigma^{(D1)}_{\lev{l1}}, 
		f_2(\sigma^{(D1)}_{\lev{l1},\lev{l2}}), 
		f_3(\sigma^{(D1)}_{\lev{l1},\lev{l3}}))(\data, \params, Q)$ \\
		$ = p(Q \mid \data, \params)$
		\end{varwidth}};
	\node [below=of sigmam] (sigmap) {\begin{varwidth}{15em}\footnotesize\centering		
		$f_g\left(\sigma^{(D1)}_{\lev{l1}}, 
		f_2(\sigma^{(D1)}_{\lev{l1},\lev{l2}}), 
		f_3(\sigma^{(D1)}_{\lev{l1},\lev{l3}})\right)$,  \\
		$\phi_{g}(\sigma^{(D1)}_{\lev{l1}}, 
		f_2(\sigma^{(D1)}_{\lev{l1},\lev{l2}}), 
		f_3(\sigma^{(D1)}_{\lev{l1},\lev{l3}}))(\data, \params, Q)$ \\
		$ = p(Q \mid \data, \params)$
		\end{varwidth}};
	\draw [->] (sigma) -- (sigmad) node[midway, left, seabornblue] {$S_D$};
	\draw [->] (sigmad) -- (sigmad1) node[midway, left, seabornblue] {$S_1$};
	\draw [->] (sigmad1) -- (sigmad2) node[midway, below right, seabornblue] {$S_2$} node[midway, above left] {by \autoref{lem:single-lev-prop2}};
	\draw [->] (sigmad1) -- (sigmad2p) node[midway, below left, seabornblue] {$S_2'$} node[midway, above right] {by \autoref{lem:s2}};
	\draw [->] (sigmad2) -- (sigmam) node[midway, right, seabornblue] {$S_3$} node[midway, left] {by \autoref{lem:single-lev-prop2}};
	\draw [->] (sigmad2p) -- (sigmamp) node[midway, left, seabornblue] {$S_3$} node[midway, right] {by \autoref{lem:single-lev-prop2} and $f$ fresh};
	\draw [->] (sigmamp) -- (sigmagen) node[midway, left, seabornblue] {$\mathrm{\textbf{gen}}(z) S_2$} node[midway, right] {by \autoref{lem:sgen}};
	\draw [->] (sigmam) -- (sigmap) node[midway, right, seabornblue] {$S_Q$} node[midway, left] {by \autoref{th:shred_gen}};
	\draw [->] (sigmagen) -- (sigmapp) node[midway, left, seabornblue] {$S_Q$} node[midway, right] {by \autoref{th:shred_gen} and $f$ fresh};
	%
	\node[below=of sigmap] (extending) {$ $};
	\node[below=of extending] (extending2) {$ $};
	\end{tikzpicture}}
	\caption{Diagrammatic proof of semantic preservation of $\xrightarrow{z}$}
	\vspace{30pt}
	\label{fig:proof}
	\end{figure*}

	Combining all of these results gives that:
	$$\sems{S'}(\sigma)(\data, \params, Q) = \sigma'' = \sigma'[f \mapsto v] = \sems{S}(\sigma)((\data, \params, Q))[f \mapsto v]$$
	In other words, the transformation $\xrightarrow{z}$ preserves store semantics (up to creating of one new fresh variable f).
	
	For the density, we get:
	\begin{align*}
	\semp{S'}&(\sigma)(\data, \params, Q) \\ 
	&= \phi_1(\data, \params_1) \left[ \sum_z \phi_2(\data, \params_1, z) \right] \phi_3(\data, \params_1, \params_3) p(z \mid \data, \params_1) p(Q \mid \data, \params) && \text{from \autoref{fig:proof}}  \\
	&= \left[\sum_z \phi_1(\data, \params_1)  \phi_2(\data, \params_1, z) \phi_3(\data, \params_1, \params_3)\right] p(z \mid \data, \params_1) p(Q \mid \data, \params) && \!\!\!\begin{array}{l}\text{by the distributive}\\ \text{law}\end{array} \\
	&\propto \left[\sum_z p(\data, \params_1, z, \params_2)\right] p(z \mid \data, \params_1) p(Q \mid \data, \params) && \!\!\!\begin{array}{l}\text{by \autoref{th:shred_gen}} \\\text{and \autoref{lem:shred2}}\end{array}\\
	&= p(\data, \params_1, \params_2) p(z \mid \data, \params_1) p(Q \mid \data, \params) && \text{marginalisation of } z \\
	&= p(\data, \params_1, \params_2) p(z \mid \data, \params_1, \params_3) p(Q \mid \data, \params) && \!\!\!\begin{array}{l}\text{by } z \bigCI \params_3 \mid \params_1\\ \text{ (\autoref{th:ci})}\end{array} \\
	&= p(\data, \params, Q) && \!\!\!\begin{array}{l}\text{by the chain rule} \\ \text{for probability}\end{array} \\
	&\propto \semp{S}(\sigma)(\data, \params, Q)	
	\end{align*}
	
	Together, this gives us $\sem{S} = \sem{S'}$ (up to $S'$ creating one new fresh variable f).
	
\end{proof}

\newpage
\section{Examples} \label{ap:examples}

\subsection{Sprinkler} 
Often, beginners are introduced to probabilistic modelling through simple, discrete variable examples, as they are more intuitive to reason about, and often have analytical solutions. Unfortunately, one cannot express such examples directly in PPLs that do not support discrete parameters. 
One well-known discrete variable example, often used in tutorials on probabilistic modelling, is the `Sprinkler' example. It models the relationship between cloudy weather, whether it rains, whether the garden sprinkler is on, and the wetness of the grass. In \autoref{fig:sprinkler}, we show a version of the sprinkler model written in SlicStan with discrete parameters (left) and the marginalisation part of its corresponding transformed version (right).

As $\mathrm{cloudy} \bigCI \mathrm{wet} \mid \mathrm{sprinkler}, \mathrm{rain}$, we do not need to include $\mathrm{wet}$ in the elimination of $\mathrm{cloudy}$, and the new factor is computed for different values of only $\mathrm{sprinkler}$ and $\mathrm{rain}$ (lines 2--6). The rest of the variables are eliminated one-by-one, involving all remaining variables (lines 7--15).

The snippet of the SlicStan code generated by our transformation is an exact implementation of the variable elimination algorithm for this model. This not only facilitates a platform for learning probabilistic programming using standard introductory models, but it can also be a useful tool for learning concepts such as marginalisation, conditional independence, and exact inference methods.  

\begin{figure*}
\centering 
\begin{multicols}{2} 
\textbf{Graphical model}	

\begin{center}
\resizebox{!}{2.4cm}{\hspace{-15pt}
	\begin{tikzpicture}
	\node[latent, text width=1cm, align=center] (cloudy) {\small cloudy};
	\coordinate[below=of cloudy, yshift=15pt] (mid) ;
	\node[latent, left=of mid, xshift=-15pt, text width=0.9cm, align=center] (sprinkler) {\small sprin- kler};
	\node[latent, right=of mid, xshift=15pt, text width=0.9cm, align=center] (rain) {\small rain};
	\node[latent, below=of mid, yshift=15pt, text width=0.9cm, align=center] (wet) {\small wet};
	\factoredge [] {} {cloudy} {sprinkler};
	\factoredge [] {} {cloudy} {rain};
	\factoredge [] {} {rain} {wet};
	\factoredge [] {} {sprinkler} {wet};
	\end{tikzpicture}
}\end{center}

\vspace{3pt}	
\textbf{SlicStan + discrete parameters support}
\vspace{-3pt}
\begin{lstlisting}[numbers=left,numbersep=\numbdist,numberstyle=\tiny\color{\numbcolor}]
	data real[2] p_rain, p_sprinkler;
	data real[2][2] p_wet;
	real p ~ beta(1, 1);
	int<2> cloudy ~ bern(p);
	int<2> sprinkler ~ bern(p_sprinkler[cloudy]);
	int<2> rain ~ bern(p_rain[cloudy]);
	int<2> wet ~ bern(p_wet[sprinkler][rain]);
\end{lstlisting}

\vspace{2cm}
\textbf{SlicStan}
\vspace{-5pt}
\begin{lstlisting}[numbers=left,numbersep=\numbdist,numberstyle=\tiny\color{\numbcolor}]
	...
	f1 = $\phi$(int<2> rain, int<2> sprinkler){
			elim(int<2> cloudy){
				cloudy ~ bern(p);
				sprinkler ~ bern(p_sprinkler[cloudy]);
				rain ~ bern(p_rain[cloudy]); }}		
	f2 = $\phi$(int<2> rain, int<2> wet){
			elim(int<2> sprinkler){
				factor(f1[rain,sprinkler]);
				wet ~ bern(p_wet[sprinkler,rain]); }}		
	f3 = $\phi$(int<2> wet){ elim(int<2> rain){
			factor(f2[rain,wet]); }}		
	f4 = $\phi$(){ elim(int<2> wet){
			factor(f3[wet]); }}
	factor(f4); 
	... $\vspace{-80pt}$
\end{lstlisting}
\end{multicols}
\vspace{-20pt}
\caption{The `Sprinkler' example.}
\label{fig:sprinkler}
\end{figure*}
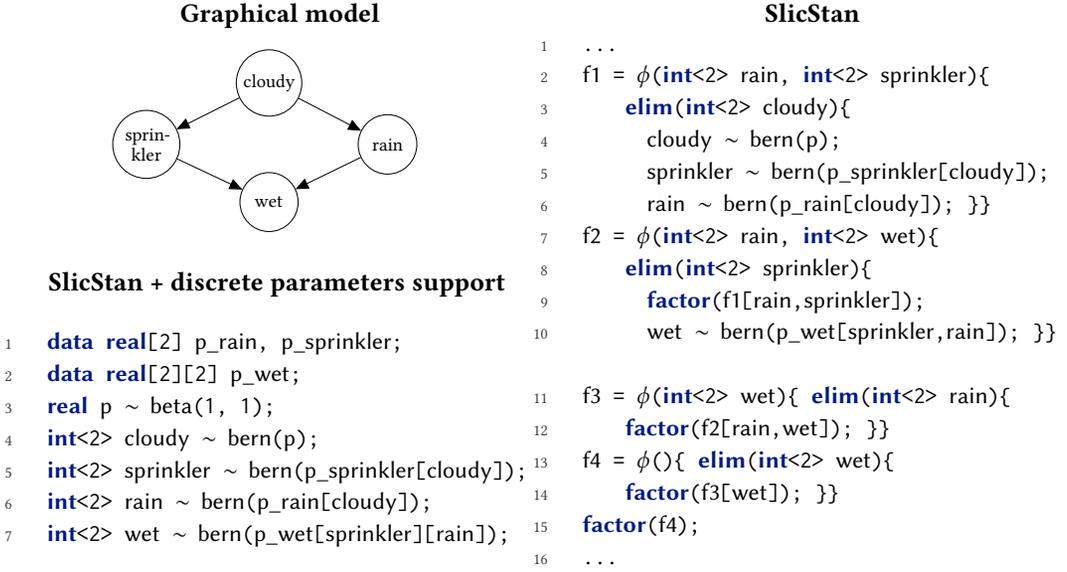

\subsection{Soft-K-means model} \label{ap:kmeans}
In \autoref{fig:softKmeans}, we present the standard soft-k-means clustering model 
as it is written in SlicStan with support for discrete model parameters (left).
The right column shows the resulting code that our program transformation 
generates.
This code consists of plain SlicStan code and no support for 
discrete model parameters is needed to perform inference on it.

The model can be used for (softly) dividing $N$ data points $\mathbf{y}$ in $D$-dimensional
Euclidean space into $K$ clusters which have means $\boldsymbol{\mu}$ and probability $\boldsymbol{\pi}$.

\vspace{8pt}
\begin{figure*}
\begin{multicols}{2} 
\textbf{SlicStan + discrete}
\vspace{-1pt}
\begin{lstlisting}
	data int D;  
	data int K; 
	data real[K] pi;
	data real N = 3;
	
	data real[D][N] y;  
	
	real[D][K] mu; 
	for(d in 1 : D) {
		for(k in 1 : K){
			mu[d][k] ~ normal(0, 1);
	}}
	
	int<K> z1 ~ categorical(pi);
	int<K> z2 ~ categorical(pi);
	int<K> z3 ~ categorical(pi);
	
	for(d in 1 : D) {
		y[d][1] ~ normal(mu[d][z1], 1);    
		y[d][2] ~ normal(mu[d][z2], 1);    
		y[d][3] ~ normal(mu[d][z3], 1);
	}
\end{lstlisting}

\vspace{3.5cm}
\textbf{SlicStan}
\vspace{-1pt}
\begin{lstlisting}
	...
	for(d in 1:D){
		for(k in 1:K){
			mu[d,k] ~ normal(0, 1);}}
	
	factor( elim(int<K> z1){
		z1 ~ categorical(pi);
		for(data int d in 1:D){
			y[d,1] ~ normal(mu[d,z1], 1);}});
	factor( elim(int<K> z2){
		z2 ~ categorical(pi);
		for(data int d in 1:D){
			y[d,2] ~ normal(mu[d,z2], 1);}});
	factor( elim(int<K> z3){
		z3 ~ categorical(pi);
		for(data int d in 1:D){
			y[d,3] ~ normal(mu[d,z3], 1);}});
	
	gen(int z3){
		z3 ~ categorical(pi);
		for(data int d in 1:D){
			y[d,3] ~ normal(mu[d,z3], 1);}}
	gen(int z2){
		z2 ~ categorical(pi);
		for(data int d in 1:D){
			y[d,2] ~ normal(mu[d,z2], 1);}}
	gen(int z1){
		z1 ~ categorical(pi);
		for(data int d in 1:D){
			y[d,1] ~ normal(mu[d,z1], 1);}}
\end{lstlisting}
\end{multicols}
\caption{Soft $K$-means.}
\label{fig:softKmeans}
\end{figure*}

\newpage
\subsection{A causal inference example} \label{ap:causal}

The question of how to adapt PPLs to causal queries, has been recently gaining popularity. One way to express interventions and reason about causality, is to assume a discrete variable specifying the direction (or absence of) causal relationship, and specify different behaviour for each case using if statements \cite{winn2012causality}. We show a simple causal inference example (\autoref{fig:causal}) written in SlicStan with direct support for discrete parameters (left) and the code that our transformation generates (right) on which we can perform inference using a combination of e.g. HMC and ancestral sampling.

This model can be read as follows.
Assume that we are in a situation where we want to answer a causal question.
We want to answer this question based on $N$ paired observations of $A$ and $B$, in some of which 
we might have intervened ($\mathrm{doB}$). 
Our model proceeds by drawing a (prior) probability that $A$ causes $B$
from a beta distribution, and then specifying $A$ and $B$ for different scenarios 
(intervention, $A$ causes $B$ and no intervention, $B$ causes $A$ and no intervention) 
using conditional statements.

\vspace{8pt}
\begin{figure*}
\begin{multicols}{2} 
\textbf{SlicStan + discrete}
\vspace{-1pt}
\begin{lstlisting}
	data real q;
	data int N;
	data int[N] A, B, doB;
	data real prob_intervention;
	
	real pAcausesB ~ beta(1, 1);
	int<2> AcausesB ~ bernoulli(pAcausesB);
	
	for (n in 1:N)
	 if(doB[n] > 0)
		B[n] ~ bernoulli(prob_intervention); 
	
	if (AcausesB > 1){
	 for (n in 1:N){
		A[n] ~ bernoulli(0.5);
		if (doB[n] < 1){
			if (A[n] > 0) { B[n] ~ bernoulli(q); } 
		else { B[n] ~ bernoulli(1 - q); }            
		}
	 }
	}
	else {
	 for (n in 1:N){
		if (doB[n] < 1){ B[n] ~ bernoulli(0.5); }        
		if (B[n] > 0){ A[n] ~ bernoulli(q); }
		else { A[n] ~ bernoulli(1 - q); }
	 }
	}
\end{lstlisting}

\textbf{SlicStan}
\vspace{-1pt}
\begin{lstlisting}
	data real q;
	data int N;
	data int[N] A, B, doB;
	data real prob_intervention;
	
	real pAcausesB ~ beta(1, 1);
	
	for(data int n in 1:N)
	 if(doB[n] > 0)
		B[n] ~ bernoulli(prob_intervention);
	
	factor(elim(int<2> AcausesB){
	 AcausesB ~ bernoulli(pAcausesB);
	 if(AcausesB > 1){
		for(data int n in 1:N){
			A[n] ~ bernoulli(0.5);
			if(doB[n] < 1){
				if(A[n] > 0){B[n] ~ bernoulli(q);}
		 else{ B[n] ~ bernoulli(1 - q); }
			}
		}
	 }    
	 else{
		for(data int n in 1:N){
			if(doB[n] < 1){ B[n] ~ bernoulli(0.5); }
			if(B[n] > 0){ A[n] ~ bernoulli(q); }
			else{ A[n] ~ bernoulli(1 - q); }
	}}});
\end{lstlisting}
\end{multicols}
\caption{A causal inference example.}
\label{fig:causal}
\end{figure*}

\end{document}